%% file: main.tex
\documentclass[11pt]{article} 
\usepackage{mathtools}
\usepackage[margin=1in]{geometry}
\usepackage[T1]{fontenc}
\usepackage[latin9]{inputenc}
\usepackage{verbatim}
\usepackage{float}

\usepackage{amssymb,amsfonts,amsmath}
\usepackage{enumerate}
\usepackage{enumitem}
\usepackage{amsthm}

\usepackage{color}
\usepackage{graphicx}
\usepackage{tikz}
\usepackage{tikzit}
\input{figures/sample.tikzstyles}

\usepackage{amsthm}
\usepackage[linesnumbered,ruled]{algorithm2e}

\SetCommentSty{mycommfont}
\usepackage{todonotes}
\usepackage{footnote}
\usepackage{subcaption}

\makesavenoteenv{algorithm}

\newfloat{Figure}{tbp}{loa}
 \floatname{Figure}{Figure}

\newcommand{\bbN}{\mathbb{N}}
\newcommand{\cT}{\mathcal{T}}
\newcommand{\cA}{\mathcal{A}}
\newcommand{\cE}{\mathcal{E}}
\newcommand{\cP}{\mathcal{P}}
\newcommand{\cN}{\mathcal{N}}
\newcommand{\cG}{\mathcal{G}}
\newcommand{\cC}{\mathcal{C}}

\newcommand{\cq}{\text{\rm CQ}}

\newcommand{\cL}{\mathcal{L}}
\newcommand{\leaves}[1]{\operatorname{leaves}(#1)}
\newcommand{\parent}[1]{\operatorname{parent}(#1)}

\newcommand{\trees}{\operatorname{Trees}}
\newcommand{\treesb}{\operatorname{Trees}_{\operatorname{B}}}
\newcommand{\sinit}{s_\text{init}}
\newcommand{\sfin}{s_\text{final}}
\newcommand{\St}{S}
\newcommand{\st}{s}
\newcommand{\str}{r}
\newcommand{\stq}{q}

\newcommand{\oSt}{\overline{\St}}
\newcommand{\ocT}{\overline{\cT}}
\newcommand{\oDelta}{\overline{\Delta}}

\newcommand{\aN}{\widetilde{N}}
\newcommand{\aT}{\widetilde{T}}
\newcommand{\aW}{\widetilde{W}}

\newcommand{\wt}[1]{\widetilde{#1}}

\newcommand{\virtsizew}{\operatorname{fsize}}
\newcommand{\virtsize}[1]{\virtsizew(#1)}

\newcommand{\duniv}{D_{1}}

\newcommand{\setc}{\mathbf{C}}
\newcommand{\setv}{\mathbf{V}}

\newcommand{\rsize}[1]{\|#1\|}

\newcommand{\Vand}{V_{\wedge}}
\newcommand{\nodesymbol}{}
\newcommand{\quicktree}[2]{
\mbox{\tikz[baseline=-2ex] \path[level distance=3mm,
	level 1/.style={sibling distance=1cm}]
	node[minimum size=2mm, inner sep=0mm] (r) {$\bullet$} 
	child { node[minimum size=0mm, inner sep=0mm] {#1}}
	child { node[minimum size=0mm, inner sep=0mm] {#2}};
	}
}
\newcommand{\singlenode}[1]{
	\mbox{\tikz[baseline=-0.7ex] \node [fill=white, draw=black, shape=circle, minimum size=2mm, inner sep=0.5mm] (1) {#1};
	}
}

\newcommand{\R}{\mathbb{R}}                     
               
\newcommand{\poly}{\text{poly}}
	
\newcommand{\eps}{\epsilon}

\newcommand{\pr}[1]{\text{\bf Pr}\normalfont\lbrack #1 \rbrack} 
\newcommand{\ex}[1]{\mathbb{E}\normalfont\lbrack #1 \rbrack}
\newcommand{\bpr}[1]{\text{\bf Pr}\normalfont \Big[#1 \Big]} 

\newcommand{\ttx}[1]{\texttt{#1}}

\theoremstyle{plain}
\newtheorem{theorem}{Theorem}[section]
\newtheorem{lemma}[theorem]{Lemma}
\newtheorem{corollary}[theorem]{Corollary}
\newtheorem{proposition}[theorem]{Proposition}

\theoremstyle{definition}
\newtheorem{definition}[theorem]{Definition}
\newtheorem{claim}[theorem]{Claim}
\newtheorem{example}[theorem]{Example}
\newtheorem{remark}[theorem]{Remark}

\newcommand{\rajesh}[1]{\todo[inline,linecolor=green,backgroundcolor=green!25,bordercolor=green]{Rajesh: #1}}

\newcommand{\eval}{W}

\newcommand{\FAIL}{\text{\rm {\bf FAIL}}}

\newcommand{\EstimatePartitionSize}{\text{\sc EstimatePartition}}
\newcommand{\Sample}{\text{\sc Sample}}

\newcommand{\SampleFromVertex}{\text{\sc SampleFromState}}
\newcommand{\astta}{\text{\sc FprasBTA}}

\newcommand{\DD}{\mathcal{D}}

\newcommand{\BB}{\mathcal{B}}
\newcommand{\FF}{\mathcal{F}}
\newcommand{\es}{\lambda}

\newcommand{\np}{\textsc{NP}}
\newcommand{\ptime}{\textsc{P}}
\newcommand{\bpp}{\textsc{BPP}}
\newcommand{\fp}{\textsc{FP}}
\newcommand{\wo}{\textsc{W}[1]}
\newcommand{\fpt}{\textsc{FPT}}
\newcommand{\sharpp}{\textsc{\#P}}
\newcommand{\logcfl}{\textsc{LogCFL}}
\newcommand{\nlog}{\textsc{NL}}
\newcommand{\rp}{\textsc{RP}}

\newcommand{\acone}{\textsc{AC}^1}

\newcommand{\skhw}{\#k\text{\rm -HW}}
\newcommand{\sohw}{\#1\text{\rm -HW}}
\newcommand{\skuhw}{\#k\text{\rm -UHW}}

\newcommand{\sacq}{\text{\rm \#ACQ}}

\newcommand{\suacq}{\text{\rm \#UACQ}}
\newcommand{\sta}{\text{\rm \#TA}}
\newcommand{\stta}{\text{\rm \#BTA}}

\newcommand{\snfa}{\text{\rm \#NFA}}
\newcommand{\snwa}{\text{\rm \#NWA}}
\newcommand{\aecsp}{\text{\rm AECSP}}

\newcommand{\saecsp}{\text{\rm \#AECSP}}

\newcommand{\skhwecsp}{\#k\text{\rm -HW-ECSP}}

\newcommand{\slp}{\#\text{\rm SuccinctNFA}}
\newcommand{\ssddnnf}{\#\text{\rm StructuredDNNF}}
\newcommand{\salp}{\#\text{\rm UnrolledSuccinctNFA}}
\newcommand{\nun}{\mathcal{N}_{\emph{unroll}}}

\newcommand{\prs}{\leq_{\text{\rm PAR}}}

\newcommand{\sol}{\text{\rm sol}}
\newcommand{\var}{\text{\rm var}}

\newcommand{\tw}{\text{\rm tw}}
\newcommand{\hw}{\text{\rm hw}}

\newcommand{\ext}{\text{\rm ext}}
\newcommand{\atoms}{\text{\rm atoms}}

\newcommand{\hill}{\text{\tt a}}
\newcommand{\ray}{\text{\tt b}}
\newcommand{\cso}{\text{\tt c1}}
\newcommand{\cst}{\text{\tt c2}}
\newcommand{\matho}{\text{\tt c3}}

\newcommand{\Grad}{\text{\it G}}

\newcommand{\Enr}{\text{\it E}}

\newcommand{\CourseCS}{\text{\it C}}
\newcommand{\CourseMath}{\text{\it M}}

\newcommand{\w}{{\bar{w}}}
\newcommand{\cd}{\Delta_C}
\newcommand{\id}{\Delta_I}
\newcommand{\rd}{\Delta_R}

\newcommand{\Vars}{\operatorname{Vars}}

\newcommand{\nt}[1]{{#1}}

\usepackage[colorlinks=true, linkcolor=blue, urlcolor=blue, citecolor=blue]{hyperref}
\usetikzlibrary{arrows,positioning,calc} 
\tikzset{
    rect/.style={
           rectangle,
           rounded corners,
           draw=black, 
           thick,
           text centered},
    rectw/.style={
           rectangle,
           rounded corners,
           draw=white, 
           thick,
           text centered},
    arrout/.style={
           ->,
           -latex,
           thick,
           },
    arrin/.style={
           <-,
           latex-,
           thick,
           },
    circ/.style={
           circle,
           draw=black, 
           thick,
           text centered},       
    circw/.style={
           circle,
           draw=white, 
           thick,
           text centered}
}

\allowdisplaybreaks
\author{
	Marcelo Arenas\\
	PUC \& IMFD Chile \\
	\texttt{marenas@ing.puc.cl} 
	\and
	Luis Alberto Croquevielle\\
	PUC \& IMFD Chile\\
	\texttt{lacroquevielle@uc.cl}
	\and 
	Rajesh Jayaram\\
	Carnegie Mellon University\\
	\texttt{rkjayara@cs.cmu.edu}
	\and
	Cristian Riveros\\
	PUC \& IMFD Chile\\
	\texttt{cristian.riveros@uc.cl}
}

\title{{\bf When is Approximate Counting \\
		 for Conjunctive Queries Tractable?}}

\date{}

\begin{document}
	
	\begin{titlepage}
\maketitle
\begin{abstract}   

Conjunctive queries are one of the most common class of queries used
in database systems, and the best studied in the literature. A seminal
result of Grohe, Schwentick, and Segoufin (STOC 2001) demonstrates
that \nt{for every class $\cG$ of graphs, the evaluation of all conjunctive queries whose underlying graph is in $\cG$ is tractable if, and only if, $\cG$ has bounded treewidth.}
In this work, we extend this characterization to the
counting problem for conjunctive queries. Specifically, for every class $\cC$ of conjunctive queries
with bounded treewidth, we introduce the first fully polynomial-time
randomized approximation scheme (FPRAS) for counting
answers to a query in $\cC$, and the first
polynomial-time algorithm for sampling answers uniformly from a query
in $\cC$. \nt{As a corollary, it follows that for every class $\cG$ of graphs, the counting problem for conjunctive queries whose underlying graph is in $\cG$ admits an FPRAS if, and only if, $\cG$ has bounded treewidth (unless $\bpp \neq \ptime$)}. In fact, our FPRAS is more general, and also applies
to conjunctive queries with bounded \textit{hypertree width}, as well
as unions of such queries.

The key ingredient in our proof is the resolution of a fundamental
counting problem from automata theory. Specifically, we demonstrate the first FPRAS and polynomial time 
sampler for the set of trees of size $n$ accepted by a \textit{tree
  automaton}, which improves the prior quasi-polynomial time
randomized approximation scheme (QPRAS) and sampling algorithm of
Gore, Jerrum, Kannan, Sweedyk, and Mahaney '97. We demonstrate how
this algorithm can be used to obtain an FPRAS for many hitherto open
problems, such as counting solutions to constraint satisfaction
problems (CSP) with bounded hypertree-width, counting the number of
error threads in programs with nested call subroutines, and counting
valid assignments to structured DNNF circuits.

\end{abstract}

\thispagestyle{empty}
\end{titlepage}

\newpage
\thispagestyle{empty}
\tableofcontents
\newpage

\pagenumbering{arabic}
\setcounter{page}{1}

\section{Introduction}
\label{sec:intro}
\input{intro-new.tex}


\subsection{Technical Overview}
\label{sec:tech}
\input{techniques}

\subsection{Other applications of the FPRAS}
\label{sec:introapp}
\input{app-intro}

\subsection{Related work}
\label{sec:related}
\input{related-work}

\section{Preliminaries}
\label{sec:preliminaries}
\input{preliminaries}
\section{From Conjunctive Queries to Tree Automata}
\label{sec:cqreduction}

\input{reduction}

\section{An FPRAS and Uniform Sampler for Tree Automata}
\label{sec:fpras}
\input{fpras2}

\section{Estimating Partition Sizes via Succinct NFAs} 
\label{sec:partition}
\input{partition-size2}

\section{Other Applications of our Main Results}
\label{sec:applications}

\input{applications-new}

\bibliography{cluster}

\end{document}

%% file: figures/sample.tikzstyles
\tikzstyle{hole}=[fill=white, draw=black, shape=circle, minimum size=9mm]
\tikzstyle{non_hole}=[fill={rgb,255: red,99; green,222; blue,255}, draw=black, shape=circle, minimum size=9mm]
\tikzstyle{fp_node}=[fill={rgb,255: red,255; green,190; blue,190}, draw=black, shape=circle, minimum size=9mm]

\tikzstyle{sm_hole}=[fill=white, draw=black, shape=circle, minimum size=7mm, minimum size=5mm, inner sep=0mm]
\tikzstyle{sm_non_hole}=[fill={rgb,255: red,99; green,222; blue,255}, draw=black, shape=circle, minimum size=7mm]
\tikzstyle{sm_fp_node}=[fill=green!20, draw=black, shape=circle, minimum size=5mm, inner sep=0mm]

\tikzstyle{dedge}=[->,-latex]
\tikzstyle{edge}=[draw=blue, thick=]
\tikzstyle{fedge}=[->,-latex, line width=0.6mm]
\tikzstyle{edges}=[-]

\tikzstyle{gate}=[fill=white, draw=black, shape=circle, minimum size=4mm]
\tikzstyle{inp}=[fill=white, draw=none, shape=circle, minimum size=5mm]
\tikzstyle{inpneg}=[fill=white, draw=none, shape=circle, minimum size=1mm,inner sep=1mm]

\tikzstyle{tnode}=[fill=green!20, draw=black, shape=circle, minimum size=2mm, inner sep=1mm,line width=1pt]

%% file: intro-new.tex

Conjunctive queries (CQ) are expressions
of the form $Q(\bar x) \leftarrow R_1(\bar y_1), \ldots, R_n(\bar y_n)$
where each $R_i$ is a relational symbol, each $\bar y_i$ is a tuple of
variables, and $\bar x$ is a tuple of output variables with $\bar x
\subseteq \bar y_1 \cup \cdots \cup \bar y_n$.
Conjunctive queries are
the most common class of queries used in database systems.
They correspond
to \textit{select-project-join} queries in relational algebra
and \textit{select-from-where} queries in SQL, and are closely related to \textit{constraint satisfaction problems} (CSPs). Therefore, the computational
complexity of tasks related to the evaluation of conjunctive
queries is a fundamental object of study.
Given as input a database instance $D$ and a conjunctive query $Q(\bar x)$,
the {\em query evaluation problem} is defined as the problem of
computing $Q(D) := \{ \bar a \mid D \models
Q(\bar a)\}$. Namely, $Q(D)$ is the set of answers $\bar a$ to $Q$ over $D$, where $\bar a$ is an assignment of the variables $\bar x$ which agrees with the relations $R_i$. The
corresponding \textit{query decision problem} is to verify whether or
not $Q(D)$ is empty. It is well known that even the query decision
problem is $\np$-complete for conjunctive queries~\cite{CM77}. Thus, a
major focus of investigation
in the area has been to find tractable special
cases~\cite{Y81,CR97,GLS98,GSS01,GLS02,FG06,GGLS16}.

In addition to evaluation, two fundamental problems for conjunctive queries are counting the number of answers to a query and uniformly sampling such answers.
The counting problem for CQ is of fundamental importance for query optimization~\cite{ramakrishnan2003database,PS13}. Specifically, the optimization process of a relational query engine requires, as input, an estimate of the number of answers to a query (without evaluating the query).  
Furthermore, uniform sampling is used to efficiently generate representative subsets of the data, instead of computing the entire query, which are often sufficient for data mining and statistical tasks~\cite{AD20}.
Starting with the work of Chaudhuri, Motwani and Narasayya~\cite{ChaudhuriMN99}, the study of random sampling from queries has attracted significant attention from the database community~\cite{ZhaoC0HY18,Chen020}.

Beginning with the
work in~\cite{Y81}, a fruitful line of research for finding tractable cases for \nt{CQs} has been to study the {\em degree of acyclicity} of a CQ.
In particular, the \textit{treewidth} $\tw(Q)$~\cite{CR97,GSS01} of a graph representing $Q$,
and more generally the \textit{hypertree width} $\hw(Q)$ of~$Q$~\cite{GLS02},
are two primary measurements of the degree of acyclicity. 
It is known that the query decision problem can be solved in polynomial time for every class
$\cC$ of CQs with bounded treewidth \cite{CR97,GSS01} or bounded
hypertree width \cite{GLS02}.\footnote{\nt{$\cC$ has bounded
	treewidth (hypertree width) if $\tw(Q) \leq k$ ($\hw(Q) \leq k$) for every $Q \in \cC$, for a fixed constant $k$.}}
A seminal result of Grohe, Schwentick,
and Segoufin~\cite{GSS01} demonstrates that
\nt{for every class $\cG$ of graphs, the evaluation of all conjunctive queries whose representing graph is in $\cG$ is tractable if, and only if, $\cG$ has bounded treewidth.}
Hence, the property of bounded treewidth \nt{provides a
characterization of} tractability of the query decision problem.

Unfortunately, uniform generation and exact counting are more challenging than query evaluation for CQs. Specifically,  given as input a conjunctive query $Q$
and database $D$, computing $|Q(D)|$ is $\sharpp$-complete even when $\tw(Q) = 1$~\cite{PS13} (that is, for
so called {\em acylic} CQs~\cite{Y81}).
Moreover, even approximate counting is intractable for queries with unbounded treewidth, since any
multiplicative approximation clearly solves the decision
problem. 
On the other hand, these facts do not preclude the existence of efficient
approximation algorithms for classes of CQs with bounded treewidth, as
the associated query decision problem is in $\ptime$.
Despite this possibility, to date no efficient
approximation algorithms for these classes are known.

In this paper, we fill this gap by demonstrating the existence of a
fully polynomial-time randomized approximation scheme (FPRAS) and a fully polynomial-time almost uniform sampler (FPAUS) for
every class of CQs with bounded hypertree width. Since $\hw(Q) \leq
\tw(Q)$ for every CQ $Q$~\cite{GLS02}, our result also
includes every class of CQs with bounded treewidth, as well as 
classes of CQs with bounded hypertree width but unbounded treewidth~\cite{GLS02}.
Specifically, we show the following.

\begin{theorem}[Theorem \ref{thm:acq-skhw-2} informal]\label{thm:intromain}
Let $\cC$ be a class of CQs with bounded hypertree width. Then there
	exists a fully polynomial-time randomized approximation scheme
	(FPRAS) that, given $Q \in \cC$ and a database $D$, estimates
	$|Q(D)|$ to multiplicative error $(1 \pm \eps)$. Moreover,
        there is a fully polynomial-time almost uniform sampler (FPAUS)
	that generates samples from $Q(D)$.
\end{theorem}

\noindent Our algorithm of Theorem \ref{thm:intromain} in fact holds for a
larger class of queries, including \textit{unions} of conjunctive
queries with bounded hypertree width
(Proposition \ref{prop:skuhw}). Note that, as defined
in \cite{JVV86}, an FPAUS samples from a distribution with variational
distance $\delta$ from uniform (see Section \ref{sec:fpras-fpaus} for
a formal definition).

An interesting question is whether there exists a larger class of queries $\mathcal{C}$ that admits an FPRAS. Since the decision problem for $\mathcal{C}$ is in $\bpp$ whenever  $\mathcal{C}$ admits an FPRAS, as a corollary of  Theorem \ref{thm:intromain} and the characterization of~\cite{GSS01}, \nt{we obtain the following answer to this question.}

\begin{corollary}[Corollary \ref{cor:characterization} informal]\label{cor:intromain}
	\nt{Let $\cG$ be a class of graphs and $\cC$ be the class of all CQs whose representing graph is in $\cG$.
	Then assuming $\wo \neq \fpt$ and $\bpp = \ptime$, the following are
	equivalent: (1) the problem of computing $|Q(D)|$ and sampling from $Q(D)$, given as input $Q \in \cC$ and a database $D$, admits an FPRAS and an FPAUS; and (2) $\cG$ has bounded treewidth.} 
%
\end{corollary}
\noindent
Corollary \ref{cor:intromain} shows that the results of \cite{GSS01}
can be extended to the approximate counting problem for
\nt{CQs}.
Perhaps surprisingly, this demonstrates that the classes of CQs
for which the decision problem is tractable, \nt{in the sense studied
in \cite{GSS01}}, are precisely \textit{the same} as the classes which
admit an FPRAS.
Besides, this gives a positive answer to the line
of research started \nt{in~\cite{ChaudhuriMN99}},
by \nt{providing a characterization of} the class of queries that
admit an almost uniform sampler.

\smallskip

\noindent \textbf{Tree automata.} The key to our results is the resolution of a fundamental counting problem from automata theory; namely, the counting problem for \textit{tree automata}. Specifically, we first demonstrate that the solution space $Q(D)$ of a conjunctive query with bounded hypertree-width can be efficiently expressed as the language accepted by a tree automaton $\mathcal{T}$. We then demonstrate the first FPRAS for the problem of counting the number of trees accepted by a tree automaton $\mathcal{T}$. 

Tree automata are the natural extension of \textit{non-deterministic finite automata} (NFA) from words to trees. This extension is a widely studied topic, since they have a remarkable capacity to model problems, while retaining many of the desirable computational properties of NFAs~\cite{S90,T97}.
 Beginning with the strong decidability result established by Rabin \cite{R69}, many important problems have been shown to be decidable via tree automata. Moreover, the fact that tree automata are equivalent to monadic second order-logic \cite{TR68} is a basic component of the proof of Courcelle's theorem~\cite{C90}. 
Further applications of tree automata, among others, include
model checking~\cite{emerson1991tree,V95}, program analysis~\cite{AEM04,AM04,AM09}, databases~\cite{N02,S07}, and knowledge representation~\cite{T99,calvanese1999reasoning,DBLP:conf/dlog/2003handbook} (see also~\cite{T97,tata2007} for a survey).

\smallskip

\noindent \textbf{The counting problem of tree automata.} Just as a non-deterministic finite automaton $N$ accepts a language $\cL(N)$ of words, a tree automata $\mathcal{T}$ accepts a language $\cL(\cT)$ of labeled trees.
Given a tree automaton $\cT$,  and an integer $n$, define the $n$-slice of $\cL(\cT)$ as $\cL_n(\cT) = \{t \in \cL(\cT) \mid |t| = n\}$, where $|t|$ is the number of vertices in $t$.  
While exactly computing the size of the $n$-slice for \textit{deterministic} finite automata and tree automata is tractable~\cite{mairson1994generating,bertoni2000random,kannan1995counting}, this is not the case for their \textit{non-deterministic} counterparts. In fact, given as input an NFA $\cA$ and a number $n$ in unary, the problem of computing $|\cL_n(\cA)|$ is $\sharpp$-hard \cite{AJ93}, which implies $\sharpp$-hardness for tree automata. Naturally, this does not rule out the possibility of efficient approximation algorithms. This observation was first exploited by Kannan, Sweedyk, and Mahaney, who gave a \textit{quasi polynomial-time approximation scheme} (QPRAS) for NFAs \cite{kannan1995counting}, which was later extended by the aforementioned authors, along with Gore and Jerrum~\cite{gore1997quasi}, to the case of tree automata.\footnote{The QPRAS of \cite{gore1997quasi} holds slightly more generally for \textit{context free grammars} (CFG).} 
Specifically, the algorithm of \cite{gore1997quasi} runs in time $\eps^{-2}(nm)^{O(\log(n))}$, where $m = |\mathcal{T}|$ is the size of the description of $\cT$, and $\eps$ is the error parameter. Improving the complexity of this algorithm to polynomial time has been a longstanding open problem. 

The algorithms of \cite{gore1997quasi} and \cite{kannan1995counting} are based on a recursive form of Karp-Luby sampling \cite{karp1989monte}, which is a type of rejection sampling. This approach has the drawback that the probability a sample is chosen is exponentially small in the depth of the recursion. Recently, using a different sampling scheme, it was shown that an FPRAS and an FPAUS exist for NFAs \cite{Arenas19}. However, the techniques in \cite{Arenas19} break down fundamentally (discussed in the following) when applied to tree automata. The main technical contribution of this work is to address the failing points of \cite{Arenas19}, and design an FPRAS for tree automata. 

\begin{theorem}[Theorem \ref{theo:fpras-bta} and  \ref{thm:samplemain} abbreviated]\label{thm1}
	Given a tree automaton $\mathcal{T}$ and $n \geq 1$, there is an algorithm which runs in time $\poly(|\cT|,n,\eps^{-1},\log(\delta^{-1}))$ and with probability $1-\delta$, outputs an estimate $\aN$ with: $	(1 - \eps) |\cL_n(\cT)| \leq \aN \leq (1 + \eps) |\cL_n(\cT)|$. Conditioned on the success of this event, there is a sampling algorithm where each call runs in time~$\poly(|\cT|,n,\log(\delta^{-1}))$, and either outputs a uniformly random tree $t \in \cL_n(\mathcal{T})$, or $\bot$. Moreover, it outputs $\bot$ with probability at most $1/2$. 
\end{theorem}
Note that conditioned on the success of the above FPRAS (run once), every subsequent call to the sampler generates a \textit{truly} uniform sample (or $\bot$). Observe that this notion of sampling is stronger than the standard notion of FPAUS (see Section \ref{sec:fpras-fpaus}). We note that the
existence of an FPAUS is in fact a corollary of the existence of an FPRAS for the above~\cite{JVV86}.

\smallskip
\noindent\textbf{Succinct NFAs.}
A key step in the proof of Theorem \ref{thm1} is a reduction to
counting and sampling from a \emph{succinct} NFA $\mathcal{N}$, which
is an NFA with succinctly encoded alphabet and transitions. Formally,
a succinct NFA $\mathcal{N}$ is a $5$-tuple $(S,\Sigma,\Delta,
s_{\text{init}}, s_{\text{final}})$, where $S$ is a set of states,
$\Sigma$ is an alphabet, $ s_{\text{init}}, s_{\text{final}}\in S$ are
the initial and final states, and $\Delta \subseteq S \times
2^\Sigma \times S $ is the transition relation, where each transition
is labeled by a set $A\subseteq \Sigma$. We assume that $\Sigma$ is
succinctly encoded via some representation (e.g. a DNF formula), and
likewise for each set $A$ such that $e = (s,A,s')$ is a transition in
$\Delta$.
Therefore, the size of the alphabet $\Sigma$ and the size of each such
set $A$
can be exponentially large in the representation of $\mathcal{N}$.  A
word $w = w_1 w_2 \dots w_n \in \Sigma^*$ is \textit{accepted} by
$\mathcal{N}$ if there is a sequence $s_{\text{init}} =
s_0,s_1,\dots,s_n = s_{\text{final}}$ of states such that there exists
a transition $(s_{i-1},A, s_{i}) \in \Delta$ with $w_i \in A$ for each
$i=1,2,\dots,n$.
Note that the special case where each transition $(s,A,s') \in \Delta$
satisfies $|A| = 1$ is precisely the standard definition of an NFA.
To solve the aforementioned problems for succinct NFA, we must assume
that the encodings of the label sets satisfy some basic
conditions. Specifically, we require that for each transition $(s,A,s')$,
we are given an oracle which can \textbf{(1)} test membership in
$A$, \textbf{(2)} produce an estimate of the size of $|A|$,
and \textbf{(3)} generate almost-uniform samples from $A$. Our full
algorithm is given formally in
Section \ref{sec:partition}.

\begin{theorem}[Theorem \ref{thm:progmain} informal]\label{thm:introlabeledpaths}
	Let $\mathcal{N} =(S,\Sigma,\Delta, s_{\text{init}}, s_{\text{final}})$
	be a succinct NFA and $n \geq 1$.  
	Suppose that the sets $A$ in each transition $(s,A,s') \in \Delta$ satisfy the properties described above. Then there is an FPRAS and an FPAUS for $\mathcal{L}_n(\mathcal{N})$.
\end{theorem}
While standard
(non-succinct) NFAs are known to admit an FPRAS by the results
of \cite{Arenas19}, Theorem \ref{thm:introlabeledpaths} is a strong
generalization of the main result of \cite{Arenas19}, and requires
many non-trivial additional insights and techniques.

\smallskip
\noindent\textbf{Applications.}
We demonstrate that the FPRAS of Theorem \ref{thm1} results in the
first polynomial-time randomized approximation algorithms for
many previously open problems in the fields of constraint satisfaction
problems,
verification of correctness of programs with nested calls to subroutines, and knowledge compilation.
We give a brief overview of these results in
Section \ref{sec:introapp}, and describe them in detail in
Section~\ref{sec:applications}.

\vspace{-.1in}

%% file: techniques.tex

\noindent \textbf{From acyclic conjuctive queries to tree automata.}
We first explain the role of trees when answering acyclic
conjuctive queries (acyclic CQs), and the role of tree automata when
counting the answers to such queries.  Consider a CQ:
$Q_1(x) \leftarrow \Grad(x), \Enr(x,y), \Enr(x,z), \CourseCS(y), \CourseMath(z)$.
This query is said to be acyclic as it can be encoded by a {\em
join tree}, that is, by a tree $t$ where each node is labeled by the
relations occurring in the query, and which satisfies the following
connectedness property: each variable in the query induces a connected
subtree of $t$~\cite{Y81}. In particular, a join tree for $Q_1(x)$ is
depicted in Figure \ref{fig-cq-w-aw-A}, where the connected subtree
induced by variable $x$ is marked in green.  An acyclic conjunctive
query $Q$ can be efficiently evaluated by using a join tree $t$
enconding it~\cite{Y81}; in fact, a tree {\em witnessing} the fact
that $\bar a \in Q(D)$ can be constructed in polynomial time.  For
example, if $\duniv = \{\Grad(\hill)$, $\Grad(\ray)$,
$\Enr(\hill, \cso)$, $\Enr(\ray, \cso)$, $\Enr(\ray, \cst)$,
$\Enr(\ray, \matho)$, $\CourseCS(\cso)$, $\CourseCS(\cst)$,
$\CourseMath(\matho)\}$, then $\ray$ is an answer
to $Q_1$ over $\duniv$. In fact, two witness trees for this answer are
shown in Figure~\ref{fig-cq-w-aw-B}. Notice that the assignments to
variable $y$ that distinguish these two trees are marked in blue.

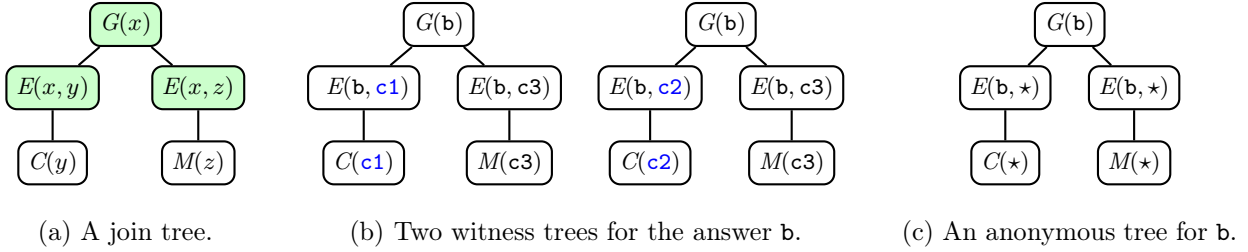
\begin{figure}[t]
	\begin{subfigure}{0.20\textwidth}
		\begin{center}
			\begin{tikzpicture}
			\node[rect, fill=green!20] (a1) {{\footnotesize $\Grad(x)$}};
			\node[rectw, below=4mm of a1] (a2) {};
			\node[rect ,fill=green!20, left=2mm of a2] (a3) {{\footnotesize $\Enr(x,y)$}}
			edge [thick] (a1);
			\node[rect, fill=green!20, right=2mm of a2] (a4) {{\footnotesize $\Enr(x,z)$}}
			edge [thick] (a1);
			\node[rect, below=4mm of a3] (a5) {{\footnotesize $\CourseCS(y)$}}
			edge [thick] (a3);
			\node[rect, below=4mm of a4] (a6) {{\footnotesize $\CourseMath(z)$}}
			edge [thick] (a4);
			\end{tikzpicture}
		\end{center}
	\vspace{-2mm}
		\caption{A join tree.\label{fig-cq-w-aw-A}}
	\end{subfigure}
	\begin{subfigure}{0.51\textwidth}
		\begin{center}
			\begin{tikzpicture}
			\node[rect] (wa1) {{\footnotesize $\Grad(\ray)$}};
			\node[rectw, below=4mm of wa1] (wa2) {};
			\node[rect, left=1mm of wa2] (wa3) {{\footnotesize \ $\Enr(\ray,{\color{blue} \cso})$}}
			edge [thick] (wa1);
			\node[rect, right=1mm of wa2] (wa4) {{\footnotesize $\Enr(\ray,\matho)$}}
			edge [thick] (wa1);
			\node[rect, below=4mm of wa3] (wa5) {{\footnotesize $\CourseCS({\color{blue} \cso})$}}
			edge [thick] (wa3);
			\node[rect, below=4mm of wa4] (wa6) {{\footnotesize $\CourseMath(\matho)$}}
			edge [thick] (wa4);

			\node[rect, right=28mm of wa1] (swa1) {{\footnotesize $\Grad(\ray)$}};
			\node[rectw, below=4mm of swa1] (swa2) {};
			\node[rect, left=1mm of swa2] (swa3) {{\footnotesize $\Enr(\ray,{\color{blue} \cst})$}}
			edge [thick] (swa1);
			\node[rect, right=1mm of swa2] (swa4) {{\footnotesize $\Enr(\ray,\matho)$}}
			edge [thick] (swa1);
			\node[rect, below=4mm of swa3] (swa5) {{\footnotesize $\CourseCS({\color{blue} \cst})$}}
			edge [thick] (swa3);
			\node[rect, below=4mm of swa4] (swa6) {{\footnotesize $\CourseMath(\matho)$}}
			edge [thick] (swa4);
			\end{tikzpicture}
		\end{center}
	\vspace{-2mm}
		\caption{Two witness trees for the answer $\ray$.\label{fig-cq-w-aw-B}}
	\end{subfigure}
	\begin{subfigure}{0.27\textwidth}
		\begin{center}
			\begin{tikzpicture}
			\node[rect] (awa1) {{\footnotesize $\Grad(\ray)$}};
			\node[rectw, below=4mm of awa1] (awa2) {};
			\node[rect, left=1mm of awa2] (awa3) {{\footnotesize $\Enr(\ray,\star)$}}
			edge [thick] (awa1);
			\node[rect, right=1mm of awa2] (awa4) {{\footnotesize $\Enr(\ray,\star)$}}
			edge [thick] (awa1);
			\node[rect, below=4mm of awa3] (awa5) {{\footnotesize $\CourseCS(\star)$}}
			edge [thick] (awa3);
			\node[rect, below=4mm of awa4] (awa6) {{\footnotesize $\CourseMath(\star)$}}
			edge [thick] (awa4);
			\end{tikzpicture}
		\end{center}
		\vspace{-2mm}
		\caption{An anonymous tree for~$\ray$.\label{fig-cq-w-aw-C}}
	\end{subfigure}
	\caption{Join, witness and anonymous trees for a CQ.\label{fig-cq-w-aw}} \vspace{-.2 in}
\end{figure}
Consider the problem 
$\sacq$, which is to count, given an acyclic CQ $Q$ and
a database $D$, the number of answers to $Q$ over $D$. 
Since the number of witness trees can be counted in polynomial
time, one might think
that $\sacq$ can also be solved in polynomial time. 
 However, there is no one-to-one correspondence between the
answers to an acyclic CQ and their witness trees; as shown in
Figure~\ref{fig-cq-w-aw-B}, two trees may witness the same answer. In fact, $\sacq$ is
$\sharpp$-complete~\cite{PS13}. 

However, 
we first observe that in a witness tree $t$, if only output variables are given actual values
and non-output variables are assigned an anonymous symbol $\star$,
then there \textit{will} be a one-to-one correspondance between answers to a
query and witnesses. Let's us denote such structures as {\em
anonymous} trees, an example of which is given in
Figure \ref{fig-cq-w-aw-C}. But how can we specify when an anonymous
tree is valid? For example, if $t'$ is the anonynomous tree
obtainined by replacing $\ray$ by $\hill$ in
Figure \ref{fig-cq-w-aw-C}, then $t'$
is not a valid anonymous tree, because $\hill$ is not an answer to $Q_1$
over $\duniv$. 
We demonstrate 
that tree automata provide the
right level of abstraction to specify the validity of such anonymous trees,
so that $\sacq$ can be reduced to a counting problem over tree
automata. In Section \ref{sec:cqreduction} we present this construction, where we consider the more general notion of
bounded hypertree width. In what follows, we focus on
the approximate counting problem for tree automata. 

\smallskip
\noindent \textbf{Tree automaton.} To capture the essence of the problem, in this section we consider a simplified version of tree automata. Specifically,
 we restrict the discussion to \textit{unlabeled} binary ordered trees, which are sufficient to present the main ideas of the algorithm. A binary ordered tree $t$ (or just tree) is a rooted binary tree where the children of each node are ordered; namely, one can distinguish between the left and right child of each non-leaf node. For a non-leaf node $u$ of $t$, we write $u1$ and $u2$ to denote the left and right children of $u$, respectively, and we denote the root of any tree $t$ by $\es$. We will write $u \in t$ to denote that $u$ is a node of $t$, and $|t|$ to denote the number of nodes of $t$. For example, Figure~\ref{fig:trees} depicts a binary ordered tree $t_1$ with $|t_1| = 9$, and another tree $t_2$ with $|t_2| = 13$, where 
  the ordering on the children is given from left to right. 
\begin{figure}
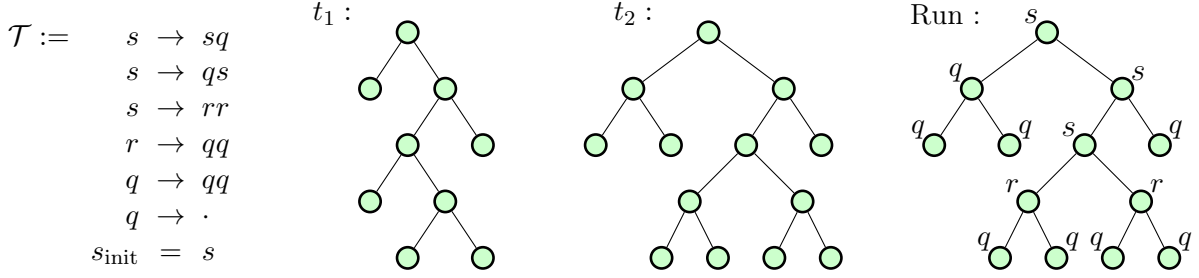

	\ctikzfig{trees} 
	\vspace{-5mm}
	\caption{A tree automata $\cT$, binary ordered trees $t_1$ and $t_2$, and a run of $\cT$ over $t_2$.}
	\label{fig:trees}
	\vspace{-.1 in}
\end{figure}


A tree automaton $\cT$ (over binary ordered trees) is a tuple $(\St, \Delta, \sinit)$ where $\St$ is a finite set of states, $\Delta \subseteq (\St \times \St \times \St) \cup \St$ is the transition relation, and $\sinit \in \St$ is the initial state.\footnote{We omit the alphabet in this definition because we consider unlabeled trees in this discussion, see Section~\ref{sec:preliminaries} for a definition of tree automata over labeled trees.} A \textit{run} $\rho$ of $\cT$ over a tree $t$ is a function $\rho:t \to \St$ mapping nodes to states that respects the transition relation. Namely, for every node $u$ of $t$ we have $\rho(u) \in \Delta$ whenever $u$ is a leaf, and $(\rho(u), \rho(u1), \rho(u2)) \in \Delta$, otherwise. 
We say that $\cT$ accepts $t$ if there exists a run $\rho$ of $\cT$ over $t$ such that $\rho(\es) = \sinit$, and such a run $\rho$ is called an accepting run of $\cT$ over $t$. 
The set of all trees accepted by $\cT$ is denoted by $\cL(\cT)$, and the $n$-slice of  $\cL(\cT)$, denoted by $\cL_n(\cT)$, is the set of trees $t \in \cL(\cT)$ with size $n$. 
For the sake of presentation, in the following we write $s \rightarrow qr$ to represent the transition $(s,q,r) \in \Delta$ and $s \rightarrow \cdot$ to represent $s \in \Delta$. Note that transitions of the form $s \rightarrow \cdot$ correspond to leaves that have no children, and can be thought as ``final states'' of a run. 

Figure~\ref{fig:trees} gives an example of a tree automaton $\cT$ with states $\{s, r, q\}$. The right-hand side of Figure~\ref{fig:trees} shows an example of an accepting run of $\cT$ over $t_2$. One can easily check from the transitions of $\cT$ in this example that a tree $t$ is in $\cL(\cT)$ if, and only if, there exists a node $v \in t$ such that both children of $v$ are internal (non-leaf) nodes.  For example, $t_2$ satisfies this property and $t_2 \in \cL(\cT)$.  
On the other hand, all nodes $v \in t_1$ have at least one child that is a leaf, and thus there is no accepting run of $\cT$ over $t_1$, so $t_1 \notin \cL(\cT)$. 

\smallskip
\noindent \textbf{Unrolling the Automaton.}
Fix $n \geq 1$ and a tree automaton $\cT = (S,\Delta,s_{init})$ as defined above. Our first step will be to \textit{unroll} the automaton, so that each state is restricted to only producing trees of a fixed size. Specifically, we construct an automaton $\ocT = (\oSt,\oDelta,\sinit^n)$, where each state $s \in S$ is duplicated $n$ times into $s^1,s^2,\dots,s^n \in \oSt$, and where $s^i$ is only allowed to derive trees of size $i$. To enforce this, each transition $s \to rq$ in $\Delta$ is replaced with $s^i \to r^jq^k \in \oDelta$ for all $j,k > 0$ such that $i=j+k+1$, and each transition $s \to \cdot$ in $\Delta$ is replaced with $s^1 \to \cdot \in \oDelta$. 
Now for every $s \in \St$, let $T(s^i)$ be set of trees that can be derived beginning from the state $s^i$ (all of which have size $i$). When $i > 1$, we can then define $T(s^i)$ via the relation
\begin{equation}\label{eqn:prod}
T(s^i) = \bigcup_{(s^i \to r^j q^k) \in \overline{\Delta}} \left(T(r^j) \otimes T(q^k) \right)	
\end{equation}
where $T(r^j) \otimes T(q^k)$ is a shorthand to denote the set of all trees that can be created by taking every $t_1 \in T(r^j)$ and $t_2 \in T(q^k)$ and forming the tree \!\!\quicktree{$t_1$}{$t_2$}\!\!\!. 
This fact allows us to define each set $T(s^i)$ recursively as a union of  ``products'' of other such sets. Our goal is then to estimate $|T(\sinit^n)|$ and sample from $T(\sinit^n)$. 

We remark that for so-called ``bottom-up deterministic'' automata $\mathcal{T}$ \cite{tata2007}, the sets $T(r^j) \otimes T(q^k)$ in the union in Equation~(\ref{eqn:prod}) are disjoint, so $|T(s^i)| = \sum_{(s^i \to r^jq^k) \in \oDelta} |T(r^j)| \cdot |T(q^k)|$ and one can then compute the values $|T(s^i)|$ exactly via dynamic programming. 
Thus, the core challenge is the \textit{ambiguity} of the problem: namely, the fact that trees $t \in T(\sinit^n)$ may admit exponentially many runs in the automata. 
For example, the tree automaton $\cT$ from Figure~\ref{fig:trees} can accept $t_2$ by two different runs. In what follows, we will focus on the problem of uniform sampling from such a set $T(s^i)$, since given a uniform sampler the problem of size estimation is routine.


\smallskip
\noindent \textbf{A QPRAS via Karp-Luby Sampling.} 
To handle the problem of sampling with ambiguous derivations, Gore, Jerrum, Kannan, Sweedyk, and Mahaney \cite{gore1997quasi} utilized a technique known as \textit{Karp-Luby} sampling.
This technique is a form of rejection sampling, where given sets $T_1,\dots,T_k$ and $T = \cup_i T_i$, one can sample from $T$ by: \textbf{(1)} sampling a set $T_i$ with probability proportional to $|T_i|$, \textbf{(2)} sampling an element $t$ uniformly from $T_i$, \textbf{(3)} accepting $t$ with probability $1/m(t)$, where $m(t)$ is the total number of sets $T_j$ which contain $t$. 
 The QPRAS of \cite{gore1997quasi} applied this procedure recursively, using approximations $\widetilde{N}(T_i)$ in the place of $|T_i|$, where the union $T = \cup_i T_i$ in question is just the union in Equation~(\ref{eqn:prod}), and each $T_i$ is a product of smaller sets $T_i = T_{i,1} \otimes T_{i,2}$ which are themselves unions of sets at a lower depth. So to carry out \textbf{(2)}, one must recursively sample from $T_{i,1}$ and $T_{i,2}$. The overall probability of acceptance in \textbf{(3)} is now exponentially small in the sampling depth. Using a classic depth reduction technique~\cite{valiant1983fast},  
they can reduce the depth to $\log(n)$, but since $m(t)$ can still be as large as $\Omega(n|\mathcal{T}|)$ at each step, the resulting acceptance probability is quasi-polynomially small.



\noindent \textbf{A Partition Based Approach.} 
The difficult with Karp-Luby sampling is that it relies on a rejection step to compensate for the fact that some elements can be sampled in multiple ways. Instead, our approach will be to \textit{partition} the sets in question, so that no element can be sampled in more than one way. Simply put, to sample from $T$, we will first partition $T$ into disjoint subsets $T_1',\dots,T_\ell'$. Next, we sample a set $T_i'$ with probability (approximately) proportional to $|T_i'|$, and lastly we set $T \leftarrow T_i'$ and now recursively sample from the new $T$. The recursion ends when the current set $T$ has just one element. Clearly no rejection procedure is needed now for the sample to be approximately uniform. To implement this template, however, there are two main implementation issues which we must address. Firstly, how to partition the set $T$, and secondly, how to efficiently estimate the size of each part $T_i$. In the remainder, we will consider these two issues in detail. 

\vphantom\\


\noindent
\fbox{\parbox{\textwidth}{
		\textbf{Our High-Level Sampling Template}\\
		\ttx{Input:} Arbitrary set $T$.
		\begin{enumerate}[topsep=0pt,itemsep=-1ex,partopsep=1ex,parsep=1ex]
			\item If $|T| = 1$, return $T$. Otherwise, find some partition $T = \cup_{i=1}^\ell T_i'$. 
			\item Call subroutine to obtain estimates $\aN(T_i') \approx |T_i'|$
			\item Set $T \leftarrow T_i'$ with probability $\aN(T_i') / \sum_j \aN(T_j')$, and recursively sample from $T$. 
		\end{enumerate}			
}}

\vphantom\\

\smallskip
\noindent \textbf{Algorithmic Overview and Setup.}
For the rest of the section, fix some state $s^i$. It will suffice to show how to generate a uniform sample from the set $T(s^i)$. To implement the above template, we will rely on having inductively pre-computed estimates of $|T(r^j)|$ for every $r \in \St$ and $j<i$. Specifically, our algorithm proceeds in rounds, where on the $j$-th round we compute an approximation $\wt{N}(r^j) \approx |T(r^j)|$ for each state $r \in \St$. In addition to these estimates, a key component of our algorithm is that, on the $i$-th round, we also store \textit{sketches} $\wt{T}(r^j)$ of each set $T(r^j)$ for $j<i$, which consist of polynomially many uniform samples from $T(r^j)$. 
One can uses these sketches $\aT(r^j)$ to aid in the generation of uniform samples for the larger sets $T(s^i)$ on the $i$-th round. For instance, given a set of trees $T = \cup_{j=1}^k T_{j}$ for some sets $T_1,\dots,T_k$ where we have estimates $\aN(T_j) \approx |T_j|$ and sketches $\aT_j \subseteq T_j$, one could estimate~$|T|$ by the value 
\begin{equation}\label{eqn:inclusion}
\sum_{j =1}^k \wt{N}(T_{j}) \left( 	\frac{\left| \aT_j \setminus \bigcup_{j' < j} T_{j'} \right|	}{\left|\aT_j 	\right|}	\right)
\end{equation} 
Here, the term in parenthesis in~\ref{eqn:inclusion} estimates the fraction of the set $\aT_{j}$ which is not already contained in the earlier sets $T_{j'}$. 


\smallskip
\noindent \textbf{The Partition Scheme for NFA.}
The above insight of sketching the intermediate subproblems $T(r^j)$ of the dynamic program and applying~\ref{eqn:inclusion} was made by \cite{Arenas19} in their FPRAS for non-deterministic finite automata (NFA). Given an NFA $\mathcal{N}$ with states $S$, $\Sigma = \{0,1\}$, and any state $s \in S$ of $\mathcal{N}$, one can similarly define the intermediate subproblem $W(s^i)$\footnote{We use $W$ to denote sets of words, and $T$ for sets of trees.} as the set of \textit{words} of length $i$ that can be derived starting at the state $s$. The FPRAS of \cite{Arenas19} similarly pre-computes sketches for these sets in a bottom-up fashion. To sample a string $w = w_1 \cdots w_i \in W(s^i)$, they sampled the symbols in $w$ bit by bit, effectively ``growing'' a prefix of $w$. 
First, $W(s^i)$ is partitioned into $W(s^i,0)\cup W(s^i,1)$, where $W(s^i,b) \subseteq W(s^i)$ is the subset of strings $x  = x_1 \cdots x_i \in W(s^i)$ with first bit $x_1$ equal to $b$. If for any prefix $w'$, we define $R_{w'} \subseteq S$ to be the set of states $r$ such that there is a path of transitions from $s$ to $r$ labeled by $w'$, then observe that $W(s^i,b) = \{b\} \cdot \cup_{r \in R_b} W(r^{i-1})$, where $\cdot$ is the concatenation operation for sets of words.
Thus $|W(s^i,b)|$ can be estimated directly by Equation~\ref{eqn:inclusion} in polynomial time. After the first bit $w_1 = b$ is sampled, they move on to sample the second bit $w_{2}$ conditioned on the prefix $w_1 =b$. By partitioning the strings again into those with prefix equal to either $b 0$ or $b 1$, each of which is described compactly as $\{ b b' \} \cdot \cup_{r \in R_{b b'}} W(r^{i-2})$ for $b' \in \{0,1\}$, one can use Equation~\ref{eqn:inclusion} again to sample $w_{2}$ from the correct distribution, and so on.

The key ``victory'' in the above approach is that for NFAs, one can compactly condition on a prefix $w'$ of a word $w \in W(s^i)$ as a union $\cup_{r \in R_{w'}}W(r^{i - |w'|})$ taken over some easy to compute subset of states $R_{w'} \subseteq S$. In other words, to condition on a partial derivation of a word, one need only remember a subset of states. This is possible because, for NFAs, the overall configuration of the automata at any given time is specified only by a single current state of the automata.  However, this fact breaks down fundamentally for tree automata. 
Namely, at any intermediate point in the derivation of a tree, the configuration of a tree automata is described not by a single state, but rather by the combination of states $(r_{t_1}^{j_1},\dots,r_{t_k}^{j_k})$ assigned to the (possibly many) leaves of the partially derived tree. So the number of possible configurations is exponential in the number of leaves of the partial tree. Consequentially, the number of sets in the union of Equation~\ref{eqn:inclusion} is exponentially large.\footnote{By being slightly clever about the order in which one derives the tree, one can reduce the number of ``active'' leafs to $O(\log n)$, which would result in a quasi-polynomial $|S|^{O(\log n)}$ time algorithm following the approach of \cite{Arenas19}, which in fact is a slight improvement on the $(|S|n)^{O(\log n)}$ obtained from \cite{gore1997quasi}.} Handling this lack of a compact representation is the main challenge for tree automata, and will require a substantially different approach to sampling. 


\smallskip
\noindent \textbf{The Partition Scheme for Tree Languages.}
Similarly at a high level to the word case, our approach to sampling will be to ``grow'' a tree $t$ from the root down. However, unlike in the word case, there is no longer any obvious method to partition the ways to grow a tree (for words, one just partitions by the next bit in the prefix). Our solution to this first challenge is to partition based on the \textit{sizes} of the subtrees of all the leaves of $t$. Namely, at each step we expand one of the leaves $\ell$ of $t$, and choose what the final sizes of the left and right subtrees of $\ell$ will be. By irrevocably conditioning on the final sizes of the left and right subtrees of a leaf $\ell$, we partition the set of possibles trees which $t$ can grow into based on the sizes that we choose. Importantly, we do \textbf{not} condition on the states which will be assigned to any of the vertices in $t$, since doing so would no longer result in a partition of $T(s^i)$.

More formally, we grow a \textit{partial tree} $\tau$, which is an ordered tree with the additional property that some of its leaves are labeled with positive integers, and these leaves are referred to as \textit{holes}. 
For an example, see the leftmost tree in Figure \ref{fig:fundamental_pathintro}. A partial tree $\tau$ is called complete if it has no holes.
For a hole $H$ of $\tau$, we denote its integral label by $\tau(H) \geq 1$, and call $\tau(H)$ the \textit{final size} of $H$, since $\tau(H)$ will indeed be the final size of the subtree rooted at $H$ once $\tau$ is complete.
Intuitively, to complete $\tau$ we must replace each hole $H$ of $\tau$ with a subtree of size exactly $\tau(H)$.
Because no states are involved in this definition, a partial tree $\tau$ is by itself totally independent of the automata.


We can now define the set $T(s^i, \tau) \subseteq  T(s^i) $ of \textit{completions} of $\tau$ as the set of trees $t \in  T(s^i)$ such that $\tau$ is a subtree of $t$ sharing the same root, and such that for every hole $H \in \tau$ the subtree rooted at the corresponding node $H \in t$ has size $\tau(H)$. Equivalently, $t$ can be obtained from $\tau$ by replacing each hole $H \in \tau$ with a subtree $t_H$ of size $\tau(H)$.
 If $\!\singlenode{$i$}\!\!$ is the partial tree consisting of a single hole with final size $i$, then we have $T(s^i, \!\text{\singlenode{$i$}}\!\!\!)  =  T(s^i)$. So at each step in the construction of $\tau$, beginning with $\tau = \!\text{\singlenode{$i$}}\!\!\!$, we will attempt to sample a tree $t$ uniformly from $T(s^i,\tau)$. To do so, we can pick any hole $H \in \tau$, and expand it by adding left and right children and fixing the final sizes of the subtrees rooted at those children. 
 There are $\tau(H)$ ways of doing this: namely, we can fix the final size of the left and right subtrees to be $j$ and $\tau(H)-j-1$ respectively, for each $j \in \{0,1,\dots,\tau(H)-1\}$. So let $\tau_j$ be the partial tree resulting from fixing these final sizes to be $j$ and $\tau(H)-j-1$, and notice that $T(s^i,\tau_0),\dots,T(s^i,\tau_{\tau(H)-1})$ partitions the set $T(s^i,\tau)$. 
  Thus it will now suffice to efficiently estimate the sizes $|T(s^i, \tau_j)|$ of each piece in the partition.


\begin{figure}
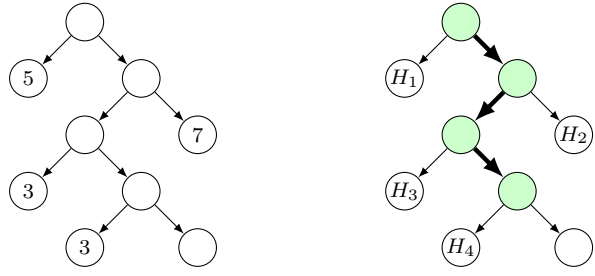

	\ctikzfig{fig_unlabel}
	\vspace{-5mm}
	\caption{Two examples of partial trees. 
		The left-hand side tree shows the label of each hole written inside the node.
		The right-hand side tree
		illustrates the main path, where non-white 
		(green) 
		nodes and thick arcs are used to highlight the vertices and edges on the main path. 
		\label{fig:fundamental_pathintro}}
\end{figure}

\smallskip
\noindent \textbf{Estimating the number of completions via the Main Path.}  
The remaining challenge can now be rephrased in following way: given any partial tree $\tau$, design a subroutine to estimate the number of completions $|T(s^i,\tau)|$.
The key tool in our approach to this is a reduction which allows us to represent the set $T(s^i,\tau)$ as the language generated by a succinct NFA, whose transitions are labeled by large sets which are succinctly encoded (see earlier definition before Theorem \ref{thm:introlabeledpaths}).  In our reduction, on round $i \leq n$, the alphabet $\Sigma$ of the succinct NFA will be the set of all ordered trees of size at most $i$. Note that this results in $\Sigma$ and the label sets $A$ being exponentially large in $n$, preventing one from applying the algorithm of \cite{Arenas19}.



Our first observation is that by always choosing the hole $H$ at the lowest depth to expand in the partitioning scheme, the resulting holes $H_1,\dots,H_k \in \tau$ will be \textit{nested} within each other. Namely, for each $i>1$, $H_i$ will be contained in the subtree rooted at the sibling of $H_{i-1}$. Using this fact, we can define a distinguished path $P$ between the parent of $H_1$ and the parent of $H_k$. Observe that each hole $H_j$ must be a child of some node in $P$. We call $P$ the \textit{main path} of $\tau$ (see Figure \ref{fig:fundamental_pathintro}). For simplicity, assume that each vertex $v \in P$ has exactly one child that is a hole of $\tau$,\footnote{Extra care should be taken when this is not the case.} and label the vertices of the path $P = \{v_1,v_2,\dots,v_k\}$, so that $H_j$ is the child of $v_j$. Notice by the above nestedness property, the holes $H_{j},H_{j+1},\dots,H_{k}$ are all contained in the subtree rooted at $v_j$. 

Observe that any completed tree $t\in T(s^i,\tau)$ can be uniquely  represented by the trees $(t_1,\dots,t_k)$, such that $t$ is obtained from $\tau$ by replacing each hole $H_i \in \tau$ by the tree $t_i$. Thinking of each tree $t_i$ as a symbol in the alphabet $\Sigma$ of all ordered trees, we can thus specify the tree $t$ by a \textit{word} $t_1 t_2 \cdots t_k \in \Sigma^*$. So our goal is to show that the set of words $T(s^i,\tau) = \{t_1 \cdots t_k \in \Sigma^*  \mid 
t_1 \cdots t_k \in T(s^i,\tau) \}$
 is the language accepted by an succinct NFA $\mathcal{N}$ over the alphabet of ordered trees $\Sigma$ with polynomially many states and set-labeled transitions. 

Now NFAs can only express labeled paths (i.e., words) and not trees. However, the key observation is that if we restrict ourselves to the main path $P$, then the sequences of states from $\mathcal{T}$ which can occur along $P$ can indeed be expressed by an NFA. Informally, for every vertex $v_j \in P$ with (wlog) left child $H_j$, and for every transition $s \to r s' $ in the tree automata $\mathcal{T}$ which could occur at $v_j$, we create a unique transition $s \to s'$ in the succinct NFA. Here, the two states $s,s'$ are assigned to the vertices $v_j,v_{j+1}$ on the main path $P$, and the state $r$ is placed inside of the hole $H_j$.  Now the set of trees $t_j$ which could be placed in $H_j$ by this transition \textit{only} depends on the state $r$. Specifically, this set of trees is exactly $T(r^{\tau(H_j)})$. Thus, if we label this transition $s \to s'$ in the succinct NFA by the set $T(r^{\tau(H_j)})$, the language accepted by the NFA will be precisely $T(s^i,\tau)$. 
 The full details can be found in Section \ref{sec:partition}. 
 
The crucial fact about this construction is that the transition labels of the succinct NFA are all sets of the form $T(r^j)$ for some $r \in S$ and $j < i$.  Since $j< i$, our algorithm has already pre-computed the sketches $\wt{T}(r^j)$ and estimates $\wt{N}(T^j)$ of the label sets $T(r^j)$ at this point. We will use these sketches and estimate to satisfy the ``oracle'' assumptions of Theorem \ref{thm:introlabeledpaths}.

\smallskip
\noindent \textbf{An FPRAS for Succinct NFAs.} 
Now that we have constructed the succinct NFA $\mathcal{N}$ which recognizes the language $T(s^i,\tau)$ as its $k$-slice, we must devise a subroutine to approximate the size of the $k$-slice of $\mathcal{N}$. Let $S',\Delta'$ be the states and transitions of $\mathcal{N}$.
 In order to estimate $|\mathcal{L}_k(\mathcal{N})|$, we mimic the inductive, dynamic programming approach of our ``outside'' algorithm.\footnote{We think of this subroutine to estimate $|T(s^i,\tau)|$ as being the ``inner loop'' of the FPRAS.} Namely, we define partial states of a dynamic program on  $\mathcal{N}$, by setting $W(x^\ell)$ to be the set of words of length $\ell$ accepted by $\mathcal{N}$ starting from the state $x \in S'$. We then similarly divide the computation of our algorithm into rounds, where on  round $\ell$ of the subroutine, we inductively pre-compute new NFA sketches $\wt{W}(x^\ell)$ of $W(x^\ell)$ and estimates $\wt{N}(x^\ell)$ of $|W(x^\ell)|$ for each state $x \in S'$. Given these estimates and sketches, our procedure for obtaining the size estimates $\wt{N}(x^{\ell})$ is straightforward. Thus, similar to the outside algorithm, the central challenge is to design a polynomial time algorithm to sample from the set $W(x^\ell)$, allowing us to construct the sketch $\wt{W}(x^\ell)$. 
 
For a string $u \in \Sigma^*$, define $W(x^\ell, u)$ to be the set of strings $w \in W(x^\ell)$ with prefix equal to $u$.
Recall the approach of \cite{Arenas19} to this problem for standard NFAs began by partitioning $W(x^\ell)$ into $\bigcup_{\alpha \in \Sigma} W(x^\ell, \alpha)$ and estimating the size $|W(x^\ell, \alpha)|$ for each $\alpha \in \Sigma$. Then one chooses $\alpha$ with probability (approximately) $\pr{\alpha} = |W(x^\ell, \alpha)| / \sum_{\beta \in \Sigma} |W(x^\ell, \beta)|$ and recurses into the set $W(x^\ell, \alpha)$. Clearly we can no longer follow this strategy, as $|\Sigma|$ is of exponential size with respect to $\mathcal{N}$. Specifically, we cannot estimate $|W(x^\ell, \alpha)|$ for each $\alpha \in \Sigma$. Instead, our approach is to approximate the behavior of the ``idealistic'' algorithm which \textit{does} estimate all these sizes, by sampling from~$\Sigma$ without explicitly estimating the sampling probabilities $\pr{\alpha}$. 
Namely, for a prefix $u$ we must sample a string $v \sim W(x^\ell, u)$, by first sampling the next symbol $\alpha \sim \Sigma$ from a distribution $\wt{\mathcal{D}}(u)$ which is close to the true distribution $\mathcal{D}(u)$ over $\Sigma$ given by $\pr{\alpha} = |W(x^\ell, u \cdot \alpha)| /|W(x^\ell, u)|$ for each $\alpha \in \Sigma$.  
 
 
To do this, first note that we can write 
$W(x^\ell,u) = \{u\} \cdot \cup_{y \in R(x,u)} W(y^{\ell - |u|})$, where $R(x,u) \subseteq S'$ is the set of states $y$ such that there is a a path of transitions from $x$ to $y$ labeled by sets $A_1 \dots A_{|u|}$ with $u_j \in A_{j}$ for each $j \in \{1, \ldots, |u|\}$.  
Thus the set of possible symbols $\alpha$ that we can append to $u$ is captured by the sets of labels of the transitions out of some state $y \in R(x,u)$.
Now consider the set of transitions $\{(y, A, z) \in \Delta' \mid y \in R(x,u)\}$, namely, all transitions out of some state in $R(x,u)$.
Furthermore, suppose for the moment that we were given an oracle which generates uniform samples from each label set $A$ of a transition $(y, A, z)$, and also provided estimates $\wt{N}(A)$ of the size of that set $|A|$.  Given such an oracle,  we design a multi-step rejection procedure to sample a symbol $\alpha$ approximately from  $\mathcal{D}(u)$, based on drawing samples from the external oracle and then rejecting them based on intersection ratios of our pre-computed internal NFA sketches $\wt{W}(y^{\ell - |u|})$.

 Since $\alpha$ is generated by a transition out of $R(x,u)$, we first sample such a transitions with probability proportional to the number of remaining suffixes which could be derived by taking that transition.
More specifically, the number of suffixes that can be produced by following a transition $(y, A, z)$ is given by $|A| \cdot |W(z^{\ell-|u|-1})|$, which can be approximated by $\wt{N}(A) \cdot \wt{N}(z^{\ell-|u|-1})$ using the oracle and our internal estimates. 
Then if $Z$ is the sum of the estimates $\wt{N}(A) \cdot \wt{N}(z^{\ell-|u|-1})$ taken over all transitions $\{(y, A, z) \in \Delta' \mid y \in R(x,u)\}$,  we choose a transition $(y,A,z)$ with probability $\wt{N}(A) \cdot \wt{N}(z^{\ell-|u|-1}) / Z$ and then call the oracle to obtain a sample $\alpha \sim A$. The sample $\alpha$ now defines a piece $W(x^\ell, u\cdot \alpha)$ of the partition of $W(x^\ell, u)$ which the idealistic algorithm would have estimated and potentially chosen. However, at this point $\alpha$ is not drawn approximately from the correct distribution $\mathcal{D}(u)$, since the sample from the oracle does not taken into account any information about the other transitions which could also produce $\alpha$. To remedy this, we show that it suffices to accept the symbol $\alpha$ with probability:
\begin{equation}\tag{$\dagger$}
\frac{\big|\aW(z^{\ell-|u|-1}) \setminus \bigcup_{\zeta \in \BB(\alpha) \,:\, \zeta \prec z} W(\zeta^{\ell - |u| - 1}) \big|}{\big|\aW(z^{\ell-|u|-1}) \big|} \label{rej-prob}
\end{equation}
where $\prec$ is an ordering over $S'$ and  $\mathcal{B}(\alpha)$ is the set of all states that can be reached from $R(x,u)$ by reading $\alpha$, namely, all states $\zeta$ such that there exists a transition $(\eta, B, \zeta) \in \Delta'$ with $\eta \in R(x, u)$ and $\alpha \in B$.
Otherwise, we reject $\alpha$. 
Intuitively, probability (\ref{rej-prob}) is small when the sets of suffixes which could be derived following transitions $\mathcal{B}(\alpha)$ that could also produce $\alpha$ intersect heavily. If this is the case, we have ``overcounted'' the contribution of the set $W(x^\ell, u \cdot \alpha)$ in the partition, and so the purpose of the probability (\ref{rej-prob}) is to compensate for this fact. We show that this procedure results in samples $\alpha$ drawn from a distribution $\wt{\mathcal{D}}(u)$ which is close in statistical distance to the exact distribution $\mathcal{D}(u)$. 
Furthermore, one can bound the acceptance probability by $\text{(\ref{rej-prob})}\geq 1/\poly(n)$ in expectation over the choice of $\alpha$, so after repeating the oracle call
 polynomially many times, we will accept a sample $\alpha$. Once $\alpha$ is accepted, we condition on it and move to the next symbol, avoiding any recursive rejection sampling.

We now return to the assumption of having a oracle to sample from and approximate the size of the label sets $A$. 
By construction, $A$ is a set of trees $T(s^j)$ for which we have pre-computed sketches and estimates $\wt{T}(s^j),\wt{N}(s^j)$ from the external algorithm. To simulate this oracle, we reuse the samples within the sketches $\wt{T}(s^j)$ for each call to the succinct NFA sub-routine, pretending that they are being generated fresh and on the fly. 
However, since the same sketches must be reused on each call to the subroutine, we lose independence between the samples generated within subsequent calls.  Ultimately, though, all that matters is that the estimate of  $|T(s^i,\tau)|$ produced by the subroutine is correct.
So to handle this, we show that one can condition on a deterministic property of the sketches $\{\wt{T}(s^j)\}_{s \in S, j < i}$, so that \textit{every} possible run of the succinct NFA subroutine will yield a good approximation, allowing us to ignore these dependencies. 

Lastly, we handle the propagation of error resulting from the statistical distance between $\wt{\mathcal{D}}(u)$ and $\mathcal{D}(u)$. 
This statistical error feeds into the error for the estimates $\wt{N}(x^{\ell+1})$ on the next step, both of which feed back into the statistical error when sampling from $W(x^{\ell+1})$, doubling the error at each step. We handle this by introducing an approximate rejection sampling step, inspired by an exact rejection sampling technique due to~\cite{jerrum1986random} (the exact version was also used in \cite{Arenas19}). This approximately corrects the distribution of each sample $w$, causing the error to increases linearly in the rounds instead of geometrically, which will be acceptable for our purposes.

%% file: app-intro.tex


\noindent \textbf{Constraint satisfaction problems.} Constraint satisfaction problems (CSPs) offer a general and natural setting to represent a large number of problems where solutions must satisfy some constraints, and which can be found in different areas~\cite{V00,DBLP:books/daglib/0004131,rossi2006handbook,DBLP:books/daglib/0013017,DBLP:series/faia/2009-185,russell2016artificial}. 
The most basic task associated to a CSP is the problem of verifying whether it has a solution, which corresponds to an assignment of values to the variables of the CSP that satisfies all the constraints of the problem. Tightly related with this task is the problem of counting the number of solution to a CSP. In this work, we consider this counting problem in the usual setting where a projection operator for CSPs is allowed, so that it is possible to indicate the output variables of the problem. We denote this setting 
as ECSP.

As counting the number of solutions of an ECSP is $\sharpp$-complete and cannot admit an FPRAS (unless $\np = \rp$), we focus 
on two well known 
notions of acyclicity 
that ensure that solutions can be found in polynomial time~\cite{GLS00,GLS02}. More precisely, we define $\saecsp$ as the problem of counting, given an acyclic ECSP $\cE$, the number of solutions to $\cE$. Moreover, given a fixed $k \geq 0$, we define $\skhwecsp$ as the problem of counting, given an ECSP $\cE$ whose hypertree-width is at most $k$, the number of solution for $\cE$. Although both problems are known to be $\sharpp$-complete~\cite{PS13}, we obtain as a consequence of Theorem~\ref{thm1} that both $\saecsp$ and $\skhwecsp$ admit FPRAS.

\smallskip
\noindent \textbf{Software verification.} Nested words have been proposed as a model for the formal verification of correctness of structured programs that can contain nested calls to subroutines~\cite{AEM04,AM04,AM09}. In particular, the execution of a program is viewed as a linear sequence of states, but where a matching relation is used to specify the correspondence between each point during the execution at which a procedure is called with the point when we return from that procedure call. This idea gives rise to the notion of nested word, which is defined as a regular word accompanied by a matching relation. Moreover, properties of programs to be formally verified are specified by using nested word~automata (NWA). 
The emptiness problem for nested word automata ask whether, given a NWA $\cN$, there exists a nested word accepted by $\cN$. This is a fundamental problem when looking for faulty executions of a program with nested calls to subroutines; if $\cN$ is used to encode the complement of a property we expect to be satisfied by a program, then a nested word accepted by $\cN$ encodes a bug of this program. In this sense, the following is also a very relevant problem for understanding how faulty a program~is. Define $\snwa$ as the problem of counting, given a nested word automaton $\cN$ and a string $0^n$, the number of nested words of length $n$ accepted by $\cN$. As expected, $\snwa$ is a $\sharpp$-complete problem. Interestingly, from Theorem \ref{thm1} and the results in \cite{AM09} showing how nested word automata can be represented by using tree automata over binary trees, it is possible to prove that $\snwa$ admits an FPRAS.

\smallskip
\noindent \textbf{Knowledge compilation.} Model counting is the problem of counting the number of satisfying assignments given a propositional formula. Although this problem is $\#\textsc{P}$-complete~\cite{valiant1979complexity}, there have been several approaches to tackle it~\cite{GomesSS09}.
One of them comes from the field of \emph{knowledge compilation}, a subarea in artificial intelligence~\cite{darwiche2002knowledge}.
Roughly speaking, this approach consists in dividing the reasoning process in two phases. 
The first phase is to compile the formula into a target language (e.g. Horn formulae, BDDs, circuits) that has good algorithmic properties. The second phase is to use the new representation to solve the problem efficiently.
The main goal then is to find a target language that is expressive enough to encode a rich set of propositional formulae and, at the same time, that allows for efficient algorithms to solve the counting problem.

A target language for knowledge compilation that has attracted a lot of attention is the class of DNNF circuits \cite{darwiche2001decomposable}. DNNF has good algorithmic properties in terms of satisfiability and logical operations. Furthermore, DNNF can be seen as a generalization of DNF formulae and, in particular, of binary decision diagrams (BDD), in the sense that every BDD can be transformed into a DNNF circuit in polynomial time. Moreover, DNNF is exponentially more succinct than DNF or BDD, and then it is a more appealing language for knowledge compilation.
Regarding model counting, DNNF circuits can easily encode \#P-complete problems (e.g. \#DNF) and, therefore, researchers have look into subclasses of DNNF where counting can be done more efficiently. One such a class that has recently received a lot of attention is the class of structured DNNF~\cite{pipatsrisawat2008new}, which has been used for efficient enumeration~\cite{AmarilliBJM17,AmarilliBMN19}, and has proved to be appropriate to compile propositional CNF formulae with bounded width (e.g. CV-width)~\cite{oztok2014cv}.
Unfortunately, the problem of computing the number of propositional variable assignments that satisfy a structured DNNF circuit is a $\sharpp$-complete problem, as these circuits include the class of DNF formulae.
However, and in line with the idea that structured DNNF circuits allow for more efficient counting algorithms, 
we prove that the counting problem of structured DNNF circuits admits a fully-polynomial time randomized approximation schema as a consequence of Theorem~\ref{thm1}.


%% file: related-work.tex

Several works have looked into the counting problem for CQs
(and the related problems we listed above, like CSPs). In order to
clarify the discussion, we will give a rough characterization of the
research in this area. This will better illustrate how our results
relate to previous work. So as a first idea, when counting solutions
to CQs, an importante source of difficulty is the presence of
existentially quantified variables. Consider the query we used in
Section \ref{sec:tech}:
\begin{eqnarray*}
	Q_1(x) & \leftarrow & \Grad(x), \Enr(x,y), \Enr(x,z), \CourseCS(y), \CourseMath(z).
\end{eqnarray*}
Notice that there are three variables $x$, $y$ and $z$ in the
right-hand side, while only $x$ is present in the left-hand
side. Thus, $x$ is an output variable, while $y$ and $z$ are
existentially quantified variables. An alternative notation for CQs
makes the quantification even more explicit:
\begin{eqnarray*}
	Q_1(x) & \leftarrow & \exists y\exists z \,
	(\Grad(x)\wedge \Enr(x,y)\wedge \Enr(x,z)\wedge \CourseCS(y)\wedge \CourseMath(z)).
\end{eqnarray*}
As we mentioned in Section \ref{sec:tech}, when variables are
existentially quantified, there is no one-to-one correspondence
between the answers to a CQ and their witness trees. This introduces a
level of ambiguity (i.e. potentially several witness trees for each
answer) into the counting problem, which makes it more difficult, even
though it does not make the evaluation problem any harder. In fact, it
is proved in Theorem 4 in \cite{PS13} that the counting problem is
$\sharpp$-complete for acyclic CQs over graphs (i.e. bounded arity),
even if queries are allowed a single existentially quantified variable
(and an arbitrary number of output variables). In contrast, it is
known (e.g. \cite{DALMAU2004315}) that for each class of CQs with
bounded treewidth and without existentially quantified variables, the
counting problem can be solved exactly in polynomial time.

It was open what happens with the counting problem when CQs are
considered with all their features, that is, when output and
existentially quantified variables are combined.
In particular, it was open whether the counting problem admits an
approximation in that case. Our paper aims to study precisely that
case, in contrast with previous work that does not consider such
output variables combined with existentially quantified variables
\cite{10.1145/3389390, DALMAU2004315}.

As a second idea, approaches to make the counting or evaluation
problem for CQs tractable usually revolve around imposing some
structural constraint on the query, in order to restrict its degree of
cyclicity. Most well-known is the result in \cite{Y81}, which proves
that the evaluation problem is tractable for acyclic queries. In
generalizations of this result (e.g. \cite{GSS01}), the acyclicity is
usually measured as the width of some query decomposition. Specific to
the counting problem, this type of notion is used
in \cite{DBLP:journals/mst/0001M15} to characterize tractable
cases. Notice, however, that they rely not only on the width of
different query decompositions, but also on a measure of how free
variables are spread in the query, which they call \textit{quantified
star size}. In contrast, we rely only on the structural width of the
hypertree decomposition.

%% file: preliminaries.tex


In this section, we introduce the main terminology used in this paper. 

\subsection{Intervals, strings, trees and tree automata}
\paragraph{Basic notation.}
Given $m \leq n$ with $n,m \in \mathbb{N}$, we use notation $[m,n]$ for the set $\{m, m+1, \ldots, n\}$, and notation $[n]$ for the set~$[1,n]$.
Moreover, given $u,\eps \in \mathbb{R}$ with $\eps \geq 0$, let $(u \pm \eps)$ denote the real interval $[u- \eps, u+ \eps]$. In general, we consider real intervals of the form $(1 \pm \eps)$, and we use $x(1 \pm \eps)$ to denote the range $[x- x\eps, x+ x\eps]$, and $x = (1 \pm \eps )y$ to denote the containment $x \in [y - \eps y, y+ \eps y]$.

\paragraph{Strings and Sequences.}
 Given a finite alphabet $\Sigma$, a finite string over $\Sigma$ is a sequence $w = w_1 \ldots w_n$ such that $n \geq 0$ and $w_i \in \Sigma$ for every $i \in [n]$. Notice that if $n = 0$, then $w$ is the empty word, which is denoted by~$\lambda$. We write $|w|= n$ for the length of $w$.
As usual, we denote by $\Sigma^*$ all strings over $\Sigma$. For two sets $A, B \subseteq \Sigma^*$ we denote by $A \cdot B = \{u \cdot v \mid u \in A, v\in B\}$, where $u \cdot v$ is the concatenation of two strings $u$ and $w$, and by $A^i$ the concatenation of $A$ with itself $i$ times, that is, $A^0 = \{\es\}$ and $A^{i+1} = A \cdot A^i$ for every $i \in \mathbb{N}$.

\paragraph{Ordered Trees.}
Fix $k \in \bbN$ with $k\geq 1$.
A finite ordered $k$-tree (or just a $k$-tree) is a prefix-closed non-empty finite subset $t \subseteq [k]^*$, namely,  if $w \cdot i \in t$ with $w \in [k]^*$ and $i \in [k]$, then $w \in t$ and $w\cdot j \in t$ for every $j \in [i]$. For a $k$-tree $t$, $\es \in t$ is the called the root of $t$ and every maximal element in $t$ (under prefix order) is called a leaf. We denote by $\leaves{t}$ the set of all leaves of $t$.
For every $u,v \in t$, we say that $u$ is a child of $v$, or that $v$ is the parent of $u$, if $u = v \cdot i$ for some $i \in [k]$. We say that $v$ has $n$ children if $v\cdot1, \ldots, v\cdot n \in t$ with $n = \max_{v\cdot i \in t}\{i\}$. 
We denote by $v = \parent{u}$ when $v$ is the parent of $u$ (if $u$ is the root, then $\parent{u}$ is undefined). 
Furthermore, we say that $v$ is an ancestor of $u$, or $u$ is a descendant of $v$, if $v$ is a prefix of $u$. The size of $t$, i.e. the number of nodes, is denoted by $|t|$.

Let $\Sigma$ be a finite alphabet and $t$ be a $k$-tree. Slightly abusing notation, we also use $t$ to denote a $k$-tree labeled over $\Sigma$. That is, we also consider $t$ as a function such that for every $u \in t$, it holds that $t(u) \in \Sigma$ is the label assigned to node $u$.
For $a \in \Sigma$, we denote just by $a$ the tree consisting of one node labeled with $a$.
For labeled $k$-trees $t$ and $t'$, and  a leaf $\ell \in t$, we define $t[\ell \rightarrow t']$ the labeled $k$-tree resulting from ``hanging'' $t'$ on the node $\ell$ in $t$. Formally, we have that $t[\ell \rightarrow t'] = t \, \cup \, (\{\ell\} \cdot t')$, $t[\ell \rightarrow t'](u) = t(u)$ whenever $u \in (t \smallsetminus \{\ell\})$ and $t[\ell \rightarrow t'](\ell \cdot u) = t'(u)$ whenever $u \in t'$. 
Note that the leaf $\ell$ takes in $t[\ell \rightarrow t']$ the label on $t'$ instead of its initial label on~$t$. 
When $t$ consists of just one node with label $a$ and with two children, we write $a(t_1, t_2)$ for the tree defined as $t[1 \rightarrow t_1][2 \rightarrow t_2]$, namely, the tree consisting of a root $a$ with $t_1$ and $t_2$ hanging to the left and right, respectively.
In particular, $t = a(b, c)$ is the tree with three nodes such that $t(\lambda) = a$, $t(1) = b$, and $t(2) = c$.
Finally, we denote by $\trees_k[\Sigma]$ the set of all $k$-trees labeled over $\Sigma$ (or just $k$-trees over $\Sigma$). 

\paragraph{Tree Automata.}
A (top-down) tree automaton $\cT$ over $\trees_k[\Sigma]$ is a tuple $(\St, \Sigma, \Delta, \sinit)$ where $\St$ is a finite set of states, $\Sigma$ is the finite alphabet, $\Delta \subseteq \St \times \Sigma \times (\cup_{i=0}^k \St^i)$ is the transition relation, and  $\sinit \in \St$ is the initial state.
We will usually use $\st$, $\stq$, and $\str$ to denote states in $\St$.
A run $\rho$ of $\cT$ over a $k$-tree $t$ is a function $\rho: t \rightarrow \St$ that assigns states to nodes of $t$ such that for every $u \in t$, if $u\cdot1, \ldots, u\cdot n$ are the children of $u$ in $t$, then $(\rho(u), t(u), \rho(u\cdot1) \rho(u\cdot 2) \ldots \rho(u\cdot n)) \in \Delta$. 
In particular, if $u$ is a leaf, then it holds that $(\rho(u), t(u), \es) \in \Delta$. 
We say that $\cT$ accepts $t$ if there exists a run of $\cT$ over $t$ with $\rho(\lambda) = \sinit$, and we define $\cL(\cT) \subseteq \trees_k[\Sigma]$ as the set of all $k$-trees over $\Sigma$ accepted by $\cT$.  We write $\cL_n(\cT)$ to denote the $n$\textit{-slice} of $\cL(\cT)$, namely $\cL_n(\cT)$ is the set $\{ t \in \cL(\cT) \mid |t| = n\}$ of all $k$-trees of size $n$ in $\cL(\cT)$.

Give a state $\st \in \St$, we will usually parameterize $\cT$ by the initial state $s$, specifically, we write $\cT[s] = (\St, \Sigma, \Delta, s)$ for the modification of $\cT$ where $s$ is the new initial state.
Furthermore, let 
$\tau = (s, a, w) \in \Delta$
be any transition. We denote by $\cT[\tau] = (\St, \Sigma, \Delta \cup \{(s^\star, a, w)\}, s^\star)$ where $s^\star$ is a fresh state not in $Q$. In other words, $\cT[\tau]$ is the extension $\cT$ that recognizes trees where runs are forced to start with transition $\tau$.  

A binary labeled tree $t$ is a labeled $2$-tree such that every node has two children or is a leaf. Notice that $2$-trees are different from binary trees, as in the former a node can have a single child, while in the latter this is not allowed.
For every non-leaf $u \in t$, we denote by $u\cdot1$ and $u\cdot 2$ the left and right child of $u$, respectively. Similar than for $k$-trees, we denote by $\treesb[\Sigma]$ the set of all binary trees.
We say that a tree automaton $\cT = (\St, \Sigma, \Delta, \sinit)$ is over $\treesb[\Sigma]$ if $\Delta \subseteq \St \times \Sigma \times (\{\lambda\} \cup \St^2)$.

%
%

\subsection{Approximate Counting, Almost Uniform Sampling, and \\
	 Parsimonious Reductions}
\label{sec:fpras-fpaus}
\paragraph{Definition of FPRAS.}
Given an input alphabet $\Sigma$, a \textit{randomized approximation scheme} (RAS) for a function $f : \Sigma^* \to \R$ is a randomized algorithm $\cA : \Sigma^* \times (0,1) \to  \R$ such that for every $w \in \Sigma^*$ and $\eps \in (0,1)$:
\begin{eqnarray*}
\pr{|\cA(w,\eps) - f(w)| \leq \eps \cdot f(w)} & \geq & \frac{3}{4}.
\end{eqnarray*}

\noindent
A randomized algorithm $\cA : \Sigma^* \times (0,1) \to  \R$ is a \textit{fully polynomial-time randomized approximation scheme} (FPRAS)~\cite{JVV86} for $f$, if it is a randomized approximation scheme for $f$ and, for every $w \in \Sigma^*$ and $\eps \in (0,1)$, $\cA(w,\epsilon)$ runs in polynomial time over $|w|$ and $\eps^{-1}$.
Thus, if $\cA$ is an FPRAS for $f$, then $\cA(w, \eps)$ approximates the value $f(w)$ with a relative error of $(1 \pm \eps)$, and it can be computed in polynomial time in the size $w$ and $\eps^{-1}$. 

\paragraph{Definition of FPAUS.}
In addition to polynomial time approximation algorithms, we also consider polynomial time (almost) uniform samplers.
Given an alphabet $\Sigma$ and a finite universe $\Omega$, let $g: \Sigma^* \to 2^\Omega$. 
We say that $g$ admits a \textit{fully polynomial-time almost uniform sampler} (FPAUS)~\cite{JVV86}  if there is a randomized algorithm $\cA : \Sigma^* \times (0,1) \to  \Omega \cup \{\bot\}$ such that for every $w \in \Sigma^*$ with $g(w) \neq \emptyset$, and $\delta\in (0,1)$, $\cA(w,\delta)$ outputs a value $x^* \in g(w) \cup \{\bot\}$  with
\[	\pr{x^* = x} = (1 \pm \delta)\frac{1}{|g(w)|} \quad \text{ for all } x \in g(w)\]
and, moreover, $\cA(w,\delta)$ runs in polynomial time over $|w|$ and $\log \frac{1}{\delta}$. If $g(w) = \emptyset$, a FPAUS must output a symbol $\bot$ with probability $1$.
The symbol $\bot$ can be thought of as a ``failure'' symbol, where the algorithm produces no output. Notice that whenever $g(w)$ admits a deterministic polynomial time membership testing algorithm (i.e. to test if $x \in g(w)$), it is easy to ensure that a sampler only outputs either a element $x \in g(w)$ or $\bot$. Also notice that the conditions imply that if $g(w) \neq \emptyset$, we have $\pr{x^* = \bot }\leq \delta$.  Given a set $S = g(w)$, when the function $g$ and the input $w$ is clear from context, we will say that the set $S$ admits an FPAUS to denote the fact that $g$ admits an FPAUS. 

For an example of an FPAUS, $w$ could be the encoding of a non-deterministic finite automata $\mathcal{N}$ and a number $n \in \mathbb{N}$ given in unary, and $g(w)$ could be the set of strings of length $n$ accepted by $\mathcal{N}$. A poly-time almost uniform sampler must then generate a string from $\cL_n(\cN)$ from a distribution which is pointwise a $(1 \pm \delta)$ approximation of the uniform distribution over $\cL_n(\cN)$, output $\bot$ with probability at most $\delta$, and run in time $\poly(|\cN|,n,\log\frac{1}{\delta})$. Notice that an FPAUS must run in time $\poly(\log \frac{1}{\delta})$, whereas an FPRAS may run in time $\poly(\frac{1}{\epsilon})$.

\paragraph{Parsimonious Reduction.}
Finally, given functions $f,g : \Sigma^* \to \mathbb{N}$, a polynomial-time parsimonious reduction from $f$ to $g$ is a polynomial-time computable function $h: \Sigma^* \to \Sigma^*$ such that, for every $w \in \Sigma^*$, it holds that $f(w) = g(h(w))$. If such a function $h$ exists, then we use notation $f \prs g$.
Notice that if $f \prs g$ and $g$ admits an FPRAS, then $f$ admits an FPRAS.

  
\smallskip

\subsection{The counting problems for tree automata} 
The following is the main counting problem studied in this paper regarding tree automata:
\begin{center}
	\framebox{
		\begin{tabular}{ll}
			\textbf{Problem:} &  $\sta$\\
			\textbf{Input:} & A tree automaton $\cT$ over $\trees_k[\Sigma]$ and a string $0^n$ \\
			\textbf{Output:} & $|\cL_n(\cT)|$ 
		\end{tabular}
	}
\end{center}
By the results in~\cite{tata2007} about encoding $k$-trees as binary trees using an extension operator $@$, it is possible to conclude the following:
\begin{lemma}\label{lem-tata}
	Let $\Sigma$ be a finite alphabet and $@ \notin \Sigma$. Then there exists a polynomial-time algorithm that, given a tree automata $\cT$ over $\trees_k[\Sigma]$, produces a tree automaton $\cT'$ over $\treesb[\Sigma \cup \{@\}]$ such that, for every $n \geq 1$:
	\begin{eqnarray*}
	\big|\{t \mid t \in \cL(\cT) \text{ and } |t| = n\}\big|  & = & \big|\{t' \mid t' \in \cL(\cT') \text{ and } |t'| = 2n-1\}\big|
	\end{eqnarray*}
\end{lemma}
\noindent
Therefore, we also consider in this paper the following problem:
\begin{center}
	\framebox{
		\begin{tabular}{ll}
			\textbf{Problem:} &  $\stta$\\
			\textbf{Input:} & A tree automaton $\cT$ over $\treesb[\Sigma]$ and a string $0^n$ \\
			\textbf{Output:} & $|\cL_n(\cT)|$ 
		\end{tabular}
	}
\end{center} 
As we know from Lemma \ref{lem-tata} that there exists a polynomial-time parsimonious reduction from $\sta$ to $\stta$, we can show that $\sta$ admits an FPRAS by proving that $\stta$ admits an FPRAS.

%% file: reduction.tex


We now provide the formal link between Conjunctive Queries (CQ) and
tree automata. As was briefly mentioned in the introduction, it is
possible to reduce $\sacq$ to $\sta$, where $\sacq$ is the problem of
counting the number of solutions to an acyclic CQ. Hence, the
existence of an FPRAS for $\sacq$ is inferred from the existence of an
FPRAS for $\sta$.
In this section, we formalize these claims and prove them.

Our results apply to a more general notion of acyclicity, known as the 
\textit{hypertree width} of a CQ, and the formal analysis will
focus on the more general setting. We start by formalizing conjunctive
queries, and introducing this more general notion of
acyclicity. Assume that there exist disjoint (countably) infinite sets
$\setc$ and $\setv$ of constants and variables, respectively.  Then a
conjunctive query (CQ) is an expression of the form:
\begin{eqnarray} \label{eq:cq-def-2}
	Q(\bar x)  & \leftarrow & R_1(\bar u_1), \ldots, R_n(\bar u_n),
\end{eqnarray} 
where for every $i \in [n]$, $R_i$ is a $k_i$-ary relation symbol
($k_i \geq 1$) and $\bar u_i$ is a $k_i$-ary tuple of variables and
constants (that is, elements from $\setv$ and $\setc$), and $\bar x =
(x_1,\dots,x_m)$ is a tuple of variables such that each variable $x_i$
in $\bar x$ occurs in some $\bar u_i$.
The symbol $Q$ is used as the name of the query, and $\var(R_i)$ is used to denote the set of variables in relation symbol $R_i$. Moreover, $\var(Q)$ denotes the set of all variables appearing in the query (both in the left- and right-hand sides).

Intuitively, the right-hand side $R_1(\bar u_1), \ldots, R_n(\bar u_n)$ of $Q$ is used to specify a pattern over a database, while the tuple $\bar x$ is used to store the answer to the query when such a pattern is found. More precisely, a database $D$ is a set of facts of the form $T(\bar a)$ where $\bar a$ is a tuple of constants (elements from $\cC$), which indicates that $\bar a$ is in the table $T$ in $D$. Then a homomorphism from $Q$ to $D$ is a function $h$ from the set of variables occurring in $Q$ to the constants in $D$ such that for every $i \in [n]$, it holds that $R_i(h(\bar u_i))$ is a fact in $D$,
where $h(\bar u_i)$ is obtained by applying $h$ to each component of $\bar u_i$ leaving the constants unchanged. Moreover, given such a homomorphism $h$, the tuple of constants $h(\bar x)$ is said to be an answer to $Q$ over the database $D$, and $Q(D)$ is defined as the set of answers of $Q$ over $D$.


\paragraph{Notions of Acyclicity.}
In Section \ref{sec:tech}, we consider a CQ as \textit{acyclic} if it
can be encoded by a join tree. We now introduce a more general notion
of acyclicity. Let $Q$ be a CQ of the form $Q(\bar x) \leftarrow
R_1(\bar u_1), \ldots, R_n(\bar u_n)$. A hypertree for $Q$ is a triple
$\langle T, \chi, \xi\rangle$ such that $T = (N,E)$ is a rooted tree,
and $\chi$ and $\xi$ are node-labeling functions such that for every
$p \in N$, it holds that $\chi(p) \subseteq \var(Q)$ and
$\xi(p) \subseteq \{R_1, \ldots, R_n\}$. Moreover, $\langle
T, \chi, \xi\rangle$ is said to be a hypertree decomposition for
$Q$ \cite{GLS02} if the following conditions~hold:
\begin{itemize}
	\item for each atom $i \in [n]$, there exists $p \in N$ such that $\var(R_i) \subseteq \chi(p)$;
	
	\item for each variable $x \in \var(Q)$, the set $\{ p \in N \mid x \in \chi(p)\}$ induces a (connected) subtree of $T$;
	
	\item for each $p \in N$, it holds that 
	\begin{eqnarray*}
		\chi(p) & \subseteq & \bigcup_{R \in \xi(p)} \var(R)
	\end{eqnarray*}
	
	\item for each $p \in N$, it holds that
	\begin{eqnarray*}
		\bigg(\bigcup_{R \in \xi(p)} \var(R)\bigg) \cap \bigg(\bigcup_{p' \,:\, p' \text{ is a descendant of } p \text{ in } T} \chi(p')\bigg) &  \subseteq & \chi(p)
	\end{eqnarray*}
\end{itemize}
The width of the hypertree decomposition $\langle T, \chi, \xi\rangle$ is defined as the maximum value of $|\xi(p)|$ over all vertices $p \in N$. Finally, the hypertree width $\hw(Q)$ of CQ $Q$ is defined as the minimum width over all its hypertree decompositions~\cite{GLS02}. 

\begin{example}
Consider the CQ $Q(x,y,z) \leftarrow R(x,y), S(y,z), T(z,x)$. It is
easy to see that $Q$ is a non-acyclic query (it cannot be represented
by a join tree as defined in Section \ref{sec:tech}), but we can still
study its degree of acyclicity using the idea of hypertree width. In
particular, the following is a hypertree decomposition for $Q$, where
the values of $\chi(p)$ and $\xi(p)$ are shown on the left- and
right-hand sides of the rectangle for node $p$:
\begin{center}
	\begin{tikzpicture}
		\node[rect] (a3) {$\{x,y,z\}$, $\{R,S\}$};
		\node[rect, below=8mm of a3] (a5) {$\{x,z\}$, $\{T\}$}
		edge[arrin] (a3);
	\end{tikzpicture}
\end{center}
Notice that the width of this hypertree decomposition is 2, as
$|\xi(p)| = 2$ for the root. And in fact, no hypertree
decomposition of width~1 can be constructed for $Q$, so that $\hw(Q) =
2$ (otherwise, $Q$ would be acyclic). In some way, we were forced to
bundle two of the atoms ($R$ and $S$) together and in the process
increase the width, in order to create a \textit{join tree}-like
structure.\qed
\end{example}

It was shown in \cite{GLS02} that a CQ $Q$ is acyclic if and only if
$\hw(Q) = 1$. Thus, the notion of hypertree width generalizes the
notion of acyclicity given before. We are interested in classes of
queries with bounded hypertree width, for which it has been shown that
the evaluation problem can be solved efficiently~\cite{GLS02}.
More precisely, for every $k \geq 1$
define the following counting problem.
\begin{center}
	\begin{tabular}{c}
		\framebox{
			\begin{tabular}{lp{10.5cm}}
				\textbf{Problem:} &  $\skhw$\\
				\textbf{Input:} & A conjunctive query $Q$ such that $\hw(Q) \leq k$ and a database $D$\\
				\textbf{Output:} & $|Q(D)|$
			\end{tabular}
		}
	\end{tabular}	
\end{center}
It is important to notice that $\sacq = \sohw$, as it is proved in \cite{GLS02} that a CQ is acyclic if and only if $\hw(Q) = 1$. However, we will keep both languages for historical reasons, as acyclic conjunctive queries were defined two decades earlier, and are widely used
in databases.
Both $\sacq$ and $\skhw$, for a fixed $k \geq 1$, are known to be $\sharpp$-complete~\cite{PS13}. On the positive side, based on the relationship with tree automata that we show below, we can conclude that these problems admit FPRAS and a FPAUS, as formalized in Section~\ref{sec:fpras-fpaus}.
\begin{theorem}\label{thm:acq-skhw-2}
	\sacq\ admits an FPRAS and a FPAUS, and for every constant $k \geq 1$, $\skhw$ admits an FPRAS and a FPAUS.
\end{theorem}
\nt{
Before presenting the proof of Theorem \ref{thm:acq-skhw-2}, we show
how it results in the characterization of CQs over graphs which admit
an FPRAS. For the sake of presentation, we focus on CQs without
constants.  Given a CQ $Q(\bar x) \leftarrow R_1(\bar y_1), \ldots,
R_n(\bar y_k)$, define a graph $G_Q$ represeting $Q$ as follows. The
set of vertices in $Q$ is the set of variables $\bar
y_1 \cup \cdots \cup \bar y_k$, and there exists an edge between two
variables $x$ and $y$ if, and only if, there exists $i \in [1,k]$ such
that both $x$ and $y$ occur in $\bar y_i$.
Notice that a class $\cG$ of graphs has bounded treewidth if there
exists a constant $k$ such that $\tw(G) \leq k$ for every
$G \in \cG$. Moreover, define $\cq(\cG)$ as the class of all
conjunctive queries $Q$ whose represeting graph $G_Q$ is in $\cG$.  By
the results of \cite{GSS01}, assuming that $\wo \neq \fpt$, for every
class $\cG$ of graphs, the evaluation of $\cq(\cG)$ is tractable if,
and only if, $\cG$ has bounded treewidth.
Since an FPRAS or FPAUS for the set $Q(D)$ of answers of a conjunctive
query results in a $\bpp$ algorithm for the query decision problem
(that is, to verify whether $Q(D) \neq \emptyset$), it follows that if
$\bpp = \ptime$, it is not possible to obtain an FPRAS or an FPAUS
for any class of CQs of the form $\cq(\cG)$ for a class of graphs
$\cG$ with unbounded treewidth. This demonstrates that, in a sense,
the class of CQs with bounded treewidth is precisely the class of CQs
which admit efficient approximation algorithms and~samplers.

\begin{corollary}\label{cor:characterization}
		Let $\cG$ be a class of graphs.
	Then assuming $\wo \neq \fpt$ and $\bpp = \ptime$, the following are
	equivalent:
	\begin{enumerate}[noitemsep,topsep=0pt]
		\item The problem of computing $|Q(D)|$ and sampling from $Q(D$), given as input $Q \in \cq(\cG)$ and a database $D$, admits an FPRAS and an FPAUS.
		\item $\cG$ has bounded treewidth. 
	\end{enumerate}
\end{corollary}
\begin{proof}
$(2) \Rightarrow (1)$ is implied immediately by
Theorem \ref{thm:acq-skhw-2}, noting that if $\cG$ is a class of
graphs with bounded treewidth, then $\cq(\cG)$ is a class of CQs with
bounded hypertree width~\cite{GLS02}.
For the other direction $(1) \Rightarrow
(2)$, let $\cG$ be a class of graph with unbounded treewidth.
Then by Corollary $19$ of \cite{GSS01}, assuming $\wo \neq \fpt$, it
holds that the problem of deciding whether $Q(D) = \emptyset$ for
$Q \in \cq(\cG)$ is not in $P$. Now suppose we had a FPRAS for
$\cq(\cG)$. Such an algorithm gives a $(1 \pm 1/2)$ approximation to
$|Q(D)|$ with probability $3/4$. Thus, such an FPRAS distinguishes
whether or not $Q(D) = \emptyset$ with probability $3/4$ with
two-sided error. Assuming $\bpp=\ptime$, there must exist
a \textit{deterministic} polynomial time algorithm for testing whether
$Q(D) = \emptyset$, which is a contradiction. Additionally, notice
that a sample $x \sim Q(D)$ which comes from an FPAUS with $\delta =
1/2$ implies that $Q(D) \neq \emptyset$, whereas a FPAUS must fail to
output any sample if $Q(D) = \emptyset$. Such an FPAUS therefore also
yields a $\bpp$ algorithm for the decision problem on $\cq(\cG)$, which
again is a contradiction of $\bpp = \ptime$.
\end{proof}
}

We are now ready to present the proof of Theorem \ref{thm:acq-skhw-2}.

\begin{proof}[Proof of Theorem \ref{thm:acq-skhw-2}]
	\input{reduction2}
\end{proof}

\subsection{Union of conjunctive queries}
An important and well-studied extension of the class of conjunctive queries is obtained by adding the union operator. A union of conjunctive queries (UCQ) is an expression of the form:
\begin{eqnarray}\label{eq:ucq_def}
Q(\bar x) & \leftarrow & Q_1(\bar x) \vee \cdots \vee Q_m(\bar x),
\end{eqnarray}
where $Q_i(\bar x)$ is a conjunctive query for each $i \in [m]$, and the same tuple $\bar x$ of output variables is used in the CQs $Q_1(\bar x)$, $\ldots$, $Q_m(\bar x)$. As for the case of CQs, the symbol $Q$ is used as the name of the query. 
A tuple $\bar a$ is said to be an answer of UCQ $Q$ in \eqref{eq:ucq_def} over a database $D$ if and only if $\bar a$ is an answer to $Q_i$ over $D$ for some $i \in [m]$. Thus, we have that:
\begin{eqnarray*}
Q(D) &=& \bigcup_{i=1}^m Q_i(D)
\end{eqnarray*} 
As expected, the problem of verifying, given a UCQ $Q$, a database $D$ and a tuple of constants $\bar a$, whether $\bar a$ is an answer to $Q$ over $D$ is an $\np$-complete problem~\cite{CM77}. Also as expected, the evaluation problem for union of acyclic conjunctive queries can be solved in polynomial time, given that the evaluation problem for acyclic CQs can be solved in polynomial time.
Concerning to our investigation, we are interested in the following  problem associated to the evaluation problem for union of acyclic conjunctive queries:
\begin{center}
	\framebox{
		\begin{tabular}{ll}
			\textbf{Problem:} &  $\suacq$\\
			\textbf{Input:} & A union of acyclic conjunctive queries $Q$ and a database $D$\\
			\textbf{Output:} & $|Q(D)|$
		\end{tabular}
	}
\end{center}
As expected from the result for conjunctive queries, $\suacq$ is \sharpp-complete~\cite{PS13}. However, $\suacq$ remains $\sharpp$-hard even if we focus on the case of UCQs without existentially quantified variables, that is, UCQs of the form \eqref{eq:ucq_def} where $\bar x$ consists of all the variables occurring in CQ $Q_i(\bar x)$ for each $i \in [m]$. Notice that this is in sharp contrast with the case of CQs, where $\sacq$ can be solved in polynomial time if we focus on case of CQs without existentially quantified variables~\cite{PS13}. 
However, by using Theorem \ref{thm:acq-skhw-2}, we are able to provide a positive result about the possibility of efficiently approximating $\suacq$.
\begin{proposition}\label{prop:uacq}
	\suacq\ admits an FPRAS and an FPAUS.
\end{proposition}

\begin{proof}
	We need to prove that there exists a randomized algorithm $\cA$ and a polynomial $p(x,y)$ such that $\cA$ receives as input a union of acyclic conjunctive queries $Q$, a database $D$ and $\eps \in (0,1)$, $\cA$ works in time $p(\|Q\|+\|D\|,\eps^{-1})$, where $\|Q\|+\|D\|$ is the size of $Q$ and $D$, 
	and $\cA$ satisfies the following condition:
	\begin{eqnarray*}
		\pr{|\cA(Q,D,\eps) - |Q(D)|| \leq \eps \cdot |Q(D)|} & \geq & \frac{3}{4}.
	\end{eqnarray*}
	Assume that $Q$ is of the form \eqref{eq:ucq_def}, from which we have that $Q(D) = \bigcup_{i=1}^m Q_i(D)$ and, therefore, $|Q(D)| = |\bigcup_{i=1}^m Q_i(D)|$. Thus, we know from \cite{KL83} that the algorithm $\cA$ can be constructed if three conditions are satisfied: (a) there exists a polynomial-time algorithm that verifies whether $\bar a \in Q_i(D)$; (b) there exists a randomized polynomial-time algorithm that generates an element in $Q_i(D)$ with uniform distribution; and (c) there exists a polynomial-time algorithm that computes $|Q_i(D)|$. In our case, property (a) holds as each $Q_i(\bar x)$ is an acyclic conjunctive query, while condition (c) cannot hold unless $\fp = \sharpp$, given that $\sacq$ is $\sharpp$-complete. However, as shown in \cite{gore1997quasi}, the existence of algorithm $\cA$ can still be guaranteed under condition (a) and the existence of an FPRAS for the function $(Q_i,D) \mapsto |Q_i(D)|$, as this latter condition also implies the existence of a fully polynomial-time almost uniform sampler for $Q_i(D)$~\cite{JVV86}. Therefore, we conclude that algorithm $\cA$ exists from Theorem~\ref{thm:acq-skhw-2}.
\end{proof}

As a final fundamental problem, we consider the problem of counting the number of solutions of a union of conjunctive queries of bounded hypertree width.
\begin{center}
	\framebox{
		\begin{tabular}{lp{11cm}}
			\textbf{Problem:} &  $\skuhw$\\
			\textbf{Input:} & A union of conjunctive query $Q(\bar x) \leftarrow Q_1(\bar x) \vee \cdots \vee Q_m(\bar x)$ such that $\hw(Q_i) \leq k$ for every $i \in [m]$, and a database $D$\\
			\textbf{Output:} & $|Q(D)|$
		\end{tabular}
	}
\end{center}
By using the same ideas as in the proof of Proposition \ref{prop:uacq}, we obtain from Proposition \ref{thm:acq-skhw-2} that:
\begin{proposition}\label{prop:skuhw}
	For every $k \geq 1$, it holds that $\skuhw$ admits an FPRAS and an FPAUS.
\end{proposition}

%% file: reduction2.tex
Fix $k \geq 1$. We provide a polynomial-time parsimonious reduction from $\skhw$ to $\sta$. In Section~\ref{sec:fpras}, we will show that $\sta$ admits an FPRAS (see Corollary \ref{cor:fpras-ta-bta}), which proves that $\skhw$ admits an FPRAS as well. Moreover, the reduction will be performed in such a way that given a tree accepted by the constructed tree automata $\cT$, one can uniquely construct a corresponding $x \in Q(D)$ in polynomial time. As a result, an FPAUS for tree automata implies an FPAUS for CQ's with bounded hypertree width. 

Let $D$ be a database and $Q(\bar x)$ a CQ over $D$ such that its atoms are of the form $R(\bar t)$ and $\hw(Q) \leq k$. We have from~\cite{GLS02} that there exists a polynomial-time algorithm that, given $Q$, produces a hypertree decomposition $\langle T, \chi, \xi\rangle$ for $Q$ of width $k$, where $T = (N,E)$.
Moreover, $\atoms(Q)$ is used to denote the set of atoms occurring in the right-hand side of $Q$, and for every $R \in \atoms(Q)$, notation $\bar t_R$ is used to indicate the tuple of variables in atom $R$. Whenever we have atoms indexed like $R_i$, we shall refer to $\bar t_{R_i}$ as $\bar t_i$ for the sake of clarity. Also, for every tuple $\bar x$ of variables, we use $\var(\bar x)$ to denote its set of variables, i.e., $\var((x_1, \ldots, x_r)) = \{x_1,\ldots, x_r\}$.
Finally, we can assume that $\langle T, \chi, \xi\rangle$ is a complete hypertree decomposition in the sense that for every $R \in \atoms(Q)$, there exists $p \in N$ such that $\var(\bar t_R) \subseteq \chi(p)$ and $R \in \xi(p)$~\cite{GLS02}. Finally, let $n = |N|$.

In what follows, we define a tree automaton $\cT = (S,\Sigma,\Delta,S_0)$ such that
\begin{eqnarray*}
	|Q(D)| &=& |\{ t \in \cL(\cT) \mid |t| = n\}|.
\end{eqnarray*}
Notice that for the sake of presentation, we are assuming that $\cT$ has a set $S_0$ of initial states, instead of a single initial state. Such an automaton can be translated in polynomial time into a tree automaton with a single initial state. Given a tuple of variables $\bar x = (x_1, \ldots, x_r)$ and a tuple of constants $\bar a = (a_1, \ldots, a_r)$, we use notation $\bar x \mapsto \bar a$ to indicate that variable $x_i$ is assigned value $a_i$ for every $i \in [r]$. Notice that $\bar x$ can contain repeated variables, and if this is the case then each occurrence of a repeated variable is assigned the same value.  For example, $(x,y,x,y) \to (a,b,a,b)$ is an assignment, while $(x,y,x,y) \to (a,b,a,c)$ is not an assignment if $b \neq c$.
Besides, notice that $\emptyset \mapsto \emptyset$ is an assignment.
Moreover, two such assignments $\bar x \mapsto \bar a$ and $\bar y \mapsto \bar b$ are said to be consistent if for every variable $z$ that occurs both in $\bar x$ and $\bar y$, it holds that the same value is assigned to $z$ in $\bar x \mapsto \bar a$ and in $\bar y \mapsto \bar b$. Then for every $p \in N$ such that:
\begin{eqnarray}\label{eq-chi-xi-cq-1}
	\chi(p) &=& \{y_1, \ldots, y_r\}\\
	\xi(p) &=& \{R_1, \ldots, R_s\}, \label{eq-chi-xi-cq-2}
\end{eqnarray}
and assuming that $\chi(p) \cap \var(\bar x) = \{z_1, \ldots, z_o\}$, $\bar y = (y_1, \ldots, y_r)$ and $\bar z = (z_1, \ldots, z_o)$, we define
\begin{align*}
	&S(p) \ = \ \big\{\big[p,\, \bar y \mapsto \bar a,\, \bar z \mapsto \bar b,\, \bar t_1 \mapsto \bar c_1,\, \ldots,\, \bar t_s \mapsto \bar c_s\big]\ \big|\\ 
	& \hspace{180pt} \bar R_i(\bar c_i) \text{ is a fact in } D \text{ for every } i \in [s],\\ 
	& \hspace{180pt} \bar y \mapsto \bar a \text{ is consistent with } \bar z \mapsto \bar b,\\
	&\hspace{180pt}\bar y \mapsto \bar a \text{ is consistent with } \bar t_i \mapsto \bar c_i \text{ for every } i \in [s],\\ 
	&\hspace{180pt}\text{and }  \bar t_i \mapsto \bar c_i \text{ is consistent with } \bar t_j \mapsto \bar c_j \text{ for every } i,j \in [s]\},
\end{align*}
and
\begin{eqnarray*}
	\Sigma(p) & = & \big\{\big[p,\, \bar z \mapsto \bar b\big] \mid \exists \bar a \exists \bar c_1 \cdots \exists \bar c_s : \big[p,\, \bar y \mapsto \bar a,\, \bar z \mapsto \bar b,\, \bar t_1 \mapsto \bar c_1,\, \ldots,\, \bar t_s \mapsto \bar c_s\big] \in S(p) \big\}
\end{eqnarray*}
With this terminology, we define $S_0 = S(p_0)$, where $p_0$ is the root of the hypertree decomposition $\langle T, \chi, \xi\rangle$, and we define:
\begin{eqnarray*}
	S &=& \bigcup_{p \in N} S(p)\\
	\Sigma &=& \bigcup_{p \in N} \Sigma(p)
\end{eqnarray*}
Finally, the transition relation $\Delta$ is defined as follows.
Assume again that $p \in N$ satisfies \eqref{eq-chi-xi-cq-1} and~\eqref{eq-chi-xi-cq-2}. If $p$ has children $p_1$, $\ldots$, $p_\ell$ in $T$, where $\ell \geq 1$ and for every $i \in [\ell]$:
\begin{eqnarray*}
	\chi(p_i) &=& \{u_{i,1}, \ldots, u_{i,r_i}\}\\
	\xi(p_i) &=& \{R_{i,1}, \ldots, R_{i,s_i}\},
\end{eqnarray*}
with $s_i \leq k$. Then assuming that $\chi(p_i) \cap \var(\bar x) = \{w_{i,1}, \ldots, w_{i,o_i}\}$, $\bar u_i =  (u_{i,1}, \ldots, u_{i,r_i})$ and $\bar w_i = (w_{i,1}, \ldots, w_{i,o_i})$ for each $i \in [\ell]$, the following tuple is included in $\Delta$
\begin{multline*}
	\big(\big[p,\, \bar y \mapsto \bar a,\, \bar z \to \bar b, \, \bar t_1 \mapsto \bar c_1,\, \ldots,\, \bar t_s \mapsto \bar c_s\big], \ \big[p,\, \bar z \mapsto \bar b\big],\\
	\big[p_1,\, \bar u_1 \mapsto \bar d_1,\, \bar w_1 \mapsto \bar e_1, \, \bar t_{1,1} \mapsto \bar f_{1,1},\, \ldots,\, \bar t_{1,s_1} \mapsto \bar f_{1,s_1}\big]\ \cdots\\ 
	\big[p_\ell, \,\bar u_\ell \mapsto \bar d_\ell, \, \bar w_\ell \mapsto \bar e_\ell,\, \bar t_{\ell,1} \mapsto \bar f_{\ell,1},\, \ldots,\, \bar t_{\ell,s_\ell} \mapsto \bar f_{\ell,s_\ell}\big]\big)
\end{multline*}
whenever the following conditions are satisfied:
(a) $\big[p,\, \bar y \mapsto \bar a, \, \bar z \mapsto \bar b ,\, \bar t_1 \mapsto \bar c_1, \ldots, \bar t_s \mapsto \bar c_s\big] \in S(p)$;
(b) $\big[p_i,\, \bar u_i \mapsto \bar d_i, \, \bar w_i \mapsto \bar e_i,\, \bar t_{i,1} \mapsto \bar f_{i,1},\, \ldots,\, \bar t_{i,s_i} \mapsto \bar f_{i,s_i}\big] \in S(p_i)$ for each $i \in [\ell]$;
(c) $\bar t_i \mapsto \bar c_i$ is consistent with $\bar t_{j_1,j_2} \mapsto \bar f_{j_1,j_2}$ for every $i \in [s]$, $j_1 \in [\ell]$ and $j_2 \in [s_{j_1}]$; and 
(d) $\bar t_{j_1,j_2} \mapsto \bar f_{j_1,j_2}$ is consistent with $\bar t_{j_3,j_3} \mapsto \bar f_{j_3,j_4}$ for every $j_1 \in [\ell]$, $j_2 \in [s_{j_1}]$, $j_3 \in [\ell]$, $j_4 \in [s_{j_3}]$. On the other hand, if $p$ has no children in $T$, then the following tuple is included in $\Delta$
\begin{eqnarray*}
	\big(\big[p,\, \bar y \mapsto \bar a, \, \bar z \mapsto \bar b,\, \bar t_1 \mapsto \bar c_1,\, \ldots,\, \bar t_s \mapsto \bar c_s\big], \, \big[p,\, \bar z \mapsto \bar b\big],\, \lambda\big)
\end{eqnarray*}
whenever $\big[p,\, \bar y \mapsto \bar a, \, \bar z \mapsto \bar b,\, \bar t_1 \mapsto \bar c_1,\, \ldots,\, \bar t_s \mapsto \bar c_s\big] \in S(p)$. 

It is straightforward to see that there exists a polynomial-time algorithm that generates $S$, $S_0$, $\Sigma$ and $\Delta$ from $Q$, $D$ and the hypertree decomposition $\langle T, \chi, \xi\rangle$ for $Q$. In particular, we have that $|S(p)|$ is $O(\|D\|^k)$, where $|S(p)|$ is the number of elements in $S(p)$ and $\|D\|$ is the size of the database $D$, by definition of $S(p)$ and the fact that $\chi(p) \subseteq \bigcup_{R \in \xi(p)} \var(\bar t_R)$. Notice that this implies that each $S(p)$ is of polynomial size given that $k$ is fixed and each tuple in $S$ is of polynomial size in~$\|D\|$. Moreover, observe that as $n = |N|$, we can construct the (unary) input $0^n$ for the problem $\sta$ in polynomial time in the size of $Q$, given that the hypertree decomposition $\langle T, \chi, \xi\rangle$ is of polynomial size in the size of $Q$. 

Finally, we need to prove that $|Q(D)| = |\{ t \in \cL(\cT) \mid |t| = n\}|$. To see this, for every $\bar a \in Q(D)$, define a labeled tree $t_{\bar a}$ as follows. Tree $t_{\bar a}$ has the same structure as $T$, but every node $p \in N$ is assigned the following label in $\Sigma$. Assume that $\chi(p) \cap \var(\bar x) = \{z_1, \ldots, z_r\}$ and $\bar z = (z_1, \ldots, z_r)$. Moreover, assume that $z_i$ receives the value $a_i$ in $\bar a$ for every $i \in [r]$.  Then the label of $p$ in $t_{\bar a}$ is $[p,\, \bar z \mapsto \bar a]$, where $\bar a = (a_1, \ldots, a_r)$. By definition of $\cT$, we have that $\cL(\cT) = \{t_{\bar a} \mid \bar a \in Q(D)\}$. Therefore, given that  $t_{\bar a} \neq t_{\bar a'}$ for every $\bar a, \bar a' \in  Q(D)$ such that $\bar a \neq \bar a'$, we conclude that $|Q(D)| = |\{ t \in \cL(\cT) \mid |t| = n\}|$, as every tree accepted by $\cT$ has $n$ nodes. Moreover, given any $t_{\bar{a}} \in \cL(\cT)$, one can read off the labels of the nodes in the tree $t_{\bar{a}}$ and uniquely reconstruct the corresponding $\bar{a} \in Q(D)$ in polynomial time, which verifies the second claim that a sample from $\cL(\cT)$ yields in polynomial time a unique sample from $Q(D)$.


%% file: fpras2.tex

In this section, we provide an FPRAS for $\stta$. Thus, we obtain as well that $\sta$ admits an FPRAS, since by Lemma \ref{lem-tata} there exists a polynomial-time parsimonious reduction from $\sta$ to $\stta$.

Fix a tree automaton $\cT = (S, \Sigma, \Delta, \sinit)$ over binary trees and let $n \geq 1$ be a natural number given in unary. We assume that every state in $S$ is mentioned in $\Delta$, and that every symbol in $\Sigma$ is mentioned in $\Delta$ (if that is not the case, then the elements that are not mentioned in $\Delta$ can just be removed from the tree automaton).
Let $m$ be the size of the tree automaton $\cT$, defined as $m = \|\Delta\|$, where $\|\Delta\|$ is the size of the transition relation $\Delta$ (represented as a string over an appropriate alphabet). In the following, fix an error parameter $\eps > 0$. Since our algorithm will run in time $\poly(n, m, 1/\eps)$, we can assume $\epsilon < \frac{1}{(4nm)^{18}}$ without loss of generality. Note that if we are only interested in uniform sampling, we can just fix $\eps = 1/\poly(nm)$.  Finally, recall that $\cL_n(\cT) = \{ t \mid t \in \cL(\cT)$ and $|t| =n \}$.

\begin{remark}\label{rem:size}
	We can assume that $m,n = \omega(1)$, since if $n = O(1)$, then the number of unlabeled trees is constant, so the number of labeled trees is a polynomial in $m$, and we can check whether each such a tree is in $\mathcal{L}(\cT)$ to compute $|\cL_n(\cT)|$ in polynomial time. If $m = O(1)$, then we can transform $\mathcal{T}$ into a constant sized deterministic bottom-up tree automaton,\footnote{A tree automaton is bottom-up if it assigns states to a labeled tree $t$ starting from the leaves, and moving toward the root~\cite{tata2007}. In particular, if $t$ is a binary tree, then the transition function is of the form $\Delta : S \times S \times \Sigma \to S$, that is, a state is assigned to a node depending on the states of its two children and its label.} and then $|\cL_n(\cT)|$
	can be 
	computed in polynomial time by dynamic programming. Thus, for the remainder we can now assume that $n \geq 2$ and $m \geq 3$.  \qed
\end{remark}

\paragraph{Unfolding of the tree automaton $\cT$.} 
We begin by making a number of copies of the states in $\cT$ in order to ``unfold'' $\cT$ into $n$ \textit{levels}. For this, let the new set of states be $\oSt = \{\st^i \, \mid \, i \in [n], \st \in \St\}$. Intuitively, from $\st^i$ we only want to accept trees of size $i$. This will allow us to define a natural partition scheme for the sampling procedure.
To enforce this constraint, we build a new tree automaton $\ocT = (\oSt, \Sigma, \oDelta, \sinit^n)$ such that for every transition $(\st, a, \stq \cdot \str) \in \Delta$ and $i \in [2,n]$, we add the transition $(\st^i, a, \stq^j \cdot \str^{i-j-1})$ to $\oDelta$ for every $j \in [1,i-2]$. Also, for every transition $(\st, a, \es) \in \Delta$ we add $(\st^1, a, \es)$ to $\oDelta$. 
We say that $i$ is the \textit{level} of $\st^i$.
Note that one can construct the set $\oSt$ and the automaton $\ocT$ in polynomial time in the size of $\cT$~\cite{tata2007}.

Given the definition of $\ocT$, one can easily check that $\cL(\ocT[\st^i]) = \{t \in \cL(\cT[\st]) \, \mid \, |t| = i\}$ for every $\st^i \in \oSt$.
In particular, we have that $|\cL(\ocT)| = |\{t \in \cL(\cT) \mid |t| = n\}|$ and, thus, the goal becomes to estimate  $|\cL(\ocT)| = |\cL(\ocT[\sinit^n])|$.
For clarity of notation, we write $T(\st^i)$ for $\cL(\ocT[\st^i])$ and $N(\st^i)$ for $|T(\st^i)|$. Note that the goal becomes to estimate $N(\sinit^n)$.

\begin{remark}[Proviso on the sizes of trees]
Every binary tree has an odd number of nodes. Thus, we will have that $T(\st^{2i}) = \emptyset$ and $N(\st^{2i}) = 0$ for each $i \geq 1$. However, to make the notation simpler, we do not limit ourselves to the trees of odd sizes. On the contrary, the algorithms provided in this article are able to compute $N(\st^{2i}) = 0$, and also to realize that no sample has to be produced from~$T(\st^{2i})$.  \qed
\end{remark}

\paragraph{Two basic properties, and the estimation of $N(\st^i)$.} 
Our algorithm simultaneously computes estimates $\aN(\st^i)$ for the set sizes $N(\st^i)$, as well as \textit{sketches} $\aT(\st^i)$ of $T(\st^i)$ which consist of polynomially many uniform samples from $T(\st^i)$. Specifically, at each level $i$ and for every $\st \in \St$, our algorithm will store an estimate which satisfies $\aN(\st^i) = (1 \pm i \eps) N(\st^i)$. At step $i$, for each $j < i$, our algorithm will also store $i$ distinct independent uniformly sampled subsets sets $\aT_1(\st^j),\aT_2(\st^j),\dots,\aT_i(\st^j)$ of $T(s^j)$ which satisfy certain deterministic criteria that will result in the correctness of our sampling algorithm on states $\st^i$ (see Lemma \ref{lem:estimatepart}). Using these estimates $\aN(s^j)$ and sketches $\aT_i(s^j)$ for $j < i$ as input, we will construct a procedure that allows us to obtain fresh, independent samples from the sets $T(\st^i)$ for all $\st \in S$. Formally, the properties we need to inductively condition on are as follows:


\paragraph{Property 1:} For a fixed $i \in [n]$, we have $\aN(\st^i) = (1\pm i\eps)  N(\st^i)$ for all $\st \in \St$. \label{prop-one}

\paragraph{Property 2:} For a fixed $i \in [n]$, we have an oracle which returns uniform, independent samples $t \sim T(\st^j)$ for every $j \leq i$ and $\st \in \St$, and runs in $\poly(n,m,1/\eps,\gamma)$ time, for some fixed parameter $\gamma$ which we will later choose. The oracle is allowed to fail with probability at most $3/4$, in which case it outputs no sample.

\medskip
We remark that the parameter $\gamma$ will later be set to $\log(1/\delta) + n$, where $\delta$ is the failure probability. 
Fix an arbitrary $i \in [n]$, and suppose we have computed $\aN(\stq^j)$ and $\aT_{k}(\stq^j)$ for all $\stq \in \St$, $j < i$ and $k \in [i]$. Fix now a state $\st$. We first show how to compute the estimate $\aN(\st^i)$. 
\begin{proposition}\label{prop:prop1}
	Fix $\delta \in (0,1)$. If Property $1$ and $2$ hold for all levels $j < i$, then with probability $1-\delta$ and time $\poly(n,m,1/\eps,\log(1/\delta))$
	we can compute a value $\aN(\st^i)$ such that $\aN(\st^i) = (1 \pm i\eps)N(\st^i)$. In other words, Property $1$ holds for level $i$.
\end{proposition}

\begin{proof}
	If $i = 1$, we can compute $N(\st^i)$ exactly in time $O(m)$, and we make $\aN(\st^i) = N(\st^i)$. Thus, assume that $i \geq 2$. 
	For each transition $\tau = (\st^i, a, \stq^j \cdot \str^{i-j-1}) \in \oDelta$, recall the definition of the extension $\cT[\tau]$ (see Section~\ref{sec:preliminaries}),
	which recognizes trees where runs are forced to start with transition $\tau$. We now define $N(\tau) = |\cL(\cT[\tau])|$, and observe that $N(\tau) = |T(\stq^j) \times T(\str^{i-j-1})| = N(\stq^j) \cdot N(\str^{i-j-1})$. Thus, we obtain an estimate $\aN(\tau)$ of $N(\tau)$ via:
	\begin{equation*}
	\begin{split}
	\aN(\tau) &= \aN(\stq^j) \cdot \aN(\str^{i-j-1})  \\
	&= (1 \pm j \eps)( 1 \pm (i-j-1)\eps) \cdot N(\stq^j) \cdot N(\str^{i-j-1})\\
	& = (1 \pm( j \eps + (i-j-1) \eps) +j(i-j-1) \eps^2 )  \cdot N(\tau) \\
	& = (1 \pm((i-1) \eps) +j(i-j-1) \eps^2 )  \cdot N(\tau) \\
	&= \left(1 \pm \left(i-1 + \frac{1}{n}\right)\eps\right) \cdot N(\tau)
	\end{split}
	\end{equation*}
	Where in the last equation, we used our assumption that $\eps < 1/(4nm)^{18} < 1/n^3$ and then applied the fact that $j(i-j-1)\eps^2 \leq n^2 \eps^2 \leq \eps/n$. Also, notice that we are using the fact that Property 1 holds for all sizes $j<i$.
	Now let $\tau_1, \tau_2, \dots, \tau_\ell \in \oDelta$ be all the transitions of the form $\tau_j = (\st^i, a_j, \stq_j \cdot \str_j)$ with $\stq_j,\str_j \in \oSt$ and $a_j \in \Sigma$. Observe that $N(\st^i) = |\bigcup_{j=1}^\ell \cL(\cT[\tau_j])|$. 
	Now for each $j \in [\ell]$, let $p_j$ be the probability that a uniform sample $t \sim \cL(\cT[\tau_j])$ is not contained in $\cL(\cT[\tau_{j'}])$ for all $j' < j$.	Then $N(\st^i) \ = \ \sum_{j=1}^\ell N(\tau_j) p_j$, so in order to estimate $N(\st^i)$ it suffices to estimate the values $p_j$.  Since Property 2 holds for all levels less than $i$, by making calls to oracles $t_{\stq} \sim \cL(\cT[\stq_j])$ and $t_{\str} \sim \cL(\cT[\str_j])$ we can obtain an i.i.d. sample $a_j(t_{\stq}, t_{\str})$ from $\cL(\cT[\tau_j])$ (recall the notation for trees introduced in Section \ref{sec:preliminaries}). By repeating this process, we can obtain i.i.d. samples $t_1,t_2,\dots,t_h \sim \cL(\cT[\tau_j])$ uniformly at random, where $h = O(\log(4m/\delta) m^2/\eps^2)$. Now let $\tilde{p}_j$ be the fraction of the samples $t_{k}$ such that $t_{k} \not\in \cL(\cT[\tau_{j'}])$ for each $j'< j$. Note that checking if
	$t_{k} \not\in \cL(\cT[\tau_{j'}])$ can be done in $\poly(n,m)$ time via a membership query for tree automata. 	
	Thus if we let \[X_{k} =
	\begin{cases} 
	1 & \text{if } t_{k} \not\in \cL(\cT[\tau_{j'}]) \text{ for each } j'<j \\
	0 & \text{otherwise.} \end{cases}\]
	then we have $\tilde{p}_j =h^{-1}\sum_{k=1}^h X_k$. Then setting $p_j = \ex{X_k}$, by Hoeffding's inequality we have $|\tilde{p}_j - p_j| \leq \frac{\eps}{4m}$ with probability at least $1- \delta/(2m)$, 
	so we can union bound over all $j \in [\ell]$ and obtain $|\tilde{p}_j - p_j| \leq \frac{\eps}{4m}$ for all $j \in [\ell]$ with probability at least $1-\delta$.
	Putting all together, we can derive an estimate $\aN(\st^i)$ for $N(\st^i)$ by using the estimates $\aN(\tau_j)$ and $\tilde{p}_j$ of $N(\tau_j)$ and $p_j$, respectively, as follows:
	\begin{eqnarray*}
		\aN(\st^i) 
		&=& \sum_{j=1}^\ell \aN(\tau_j) \tilde{p}_j	\\
		&=& \left(1 \pm \left(i-1 + \frac{1}{n} \right) \eps\right)\sum_{j=1}^\ell N(\tau_j)  \tilde{p}_j	\\
		&=& \left(1 \pm \left(i-1 + \frac{1}{n} \right) \eps\right)\bigg(\sum_{j=1}^\ell N(\tau_j) p_j \pm\frac{\eps}{4m} \sum_{j=1}^\ell N(\tau_j)\bigg)	\\
		&=& \left(1 \pm \left(i-1 +\frac{1}{n} \right) \eps\right) \left( N(\st^i)\left(1 \pm \frac{\eps}{4} \right) \right) \\
		&=& \left(1 \pm i\eps\right)N(\st^i).
	\end{eqnarray*}
	Where we use that $\sum_{j=1}^\ell N(\tau_j) \leq \sum_{j=1}^\ell N(\st^i) = \ell N(\st^i) \leq m N(\st^i)$ in the second to last step, and the fact that $n \geq i \geq 2$ in the last step. For runtime, notice that the key result $\pr{|\tilde{p}_j - p_j| \leq \eps/(4m)}\geq 1 - \delta/(2m)$ is conditioned on the event that we were able to obtain $h$ samples $t_k$ using the oracle. Recall that the sampling oracle can fail with probability at most $3/4$. Then, the required number of calls $h'$ to the poly-time sampling oracle is at most $4h/3$ in expectation. For our purposes, $h'=O(h)$ will also be enough, as we now show. For $j\in[\ell]$ call $G_j$ the event that we obtain $h$ samples from $\cL(\cT[\tau_{j}])$ and $H_j$ the event that $|\tilde{p}_j - p_j| \leq \eps/(4m)$. Then, as we showed above,
	\begin{equation*}
		\pr{\aN(\st^i) = (1 \pm i\eps)N(\st^i)}
		\geq \bpr{\bigcap_{j=1}^\ell(H_j\cap G_j)}
		= 1-\bpr{\bigcup_{j=1}^\ell(\overline{H_j}\cup \overline{G_j})}
		\geq 1-m\pr{\overline{H_{j_0}}\cup \overline{G_{j_0}}},
	\end{equation*}
where the last inequality is due to a union bound obtained considering $j_0 = \operatorname{argmax}_{j \in[\ell]} \pr{\overline{H_{j}}\cup \overline{G_{j}}}$. Recall that we want $\pr{\aN(\st^i) = (1 \pm i\eps)N(\st^i)} \geq 1-\delta$, hence it suffices to show 
	\begin{equation}\label{bound_fpras}
		1-m\pr{\overline{H_{j_0}}\cup \overline{G_{j_0}}} \geq 1-\delta
		\iff \frac{\delta}{m} \geq \pr{\overline{H_{j_0}}\cup \overline{G_{j_0}}} 
		\iff \pr{H_{j_0}\cap G_{j_0}} \geq 1 - \frac{\delta}{m}
	\end{equation}
By Hoeffding's inequality, as we showed before, $\pr{H_{j_0} \mid G_{j_0}}\geq 1 - \delta/(2m)$. Suppose that we also have that $\pr{G_{j_0}}\geq 1 - \delta/(2m)$. Then,
	\begin{equation*}
		\pr{H_{j_0} \cap G_{j_0}}
		= \pr{H_{j_0} \>|\> G_{j_0}}\cdot \pr{G_{j_0}}
		\geq \left(1 - \frac{\delta}{2m}\right)^2
		\geq 1 - 2\cdot \frac{\delta}{2m}
		= 1 - \frac{\delta}{m}
	\end{equation*}
as required by equation \eqref{bound_fpras}. 
Thus, it suffices to show $\pr{G_{j_0}}\geq 1 - \delta/(2m)$. Letting $X_i$ be the random variable that indicates whether the $i$-th call to the sampling procedure was successful, then the total number of samples obtained is $X=\sum_{i=1}^{h'}X_i$, where $\ex{X}\geq h'/4$, so by a Chernoff bound we have
	\begin{equation*}
		\pr{G_{j_0}}
		= 1 - \pr{\overline{G_{j_0}}}
		= 1 - \pr{X<h}
		\geq 1 - \exp\left( -\frac{h'}{8}\Big(1-\frac{4h}{h'}\Big)^2 \right)
		\geq 1 - \exp\left( -\frac{h'}{8}\Big(1-\frac{4h}{h'}\Big) \right)
	\end{equation*}
assuming that $4h < h'$. 
Hence, 
	\begin{equation*}
		1 - \exp\bigg( -\frac{h'}{8}\Big(1-\frac{4h}{h'}\Big) \bigg)
		\geq 
		1 - \frac{\delta}{2m}
		\iff 
		\frac{\delta}{2m}
		\geq
		\exp\bigg( -\frac{h'}{8}\Big(1-\frac{4h}{h'}\Big) \bigg)
		\iff
		h'
		\geq
		4h+8\ln\Big(\frac{2m}{\delta}\Big).
	\end{equation*}
so by definition of $h$, it is sufficient to set $h'=5h$, which completes the proof.
\end{proof}

\paragraph{The notion of a partial tree.}
We need to demonstrate how to obtain uniform samples from $T(\st^i)$ to build the sets $\aT_j(\st^i)$. To do this, we will provide an algorithm that recursively samples a tree $t \in T(\st^i)$ from the top down. But before showing this procedure, we need to introduce the notion of a \textit{partial tree}. In the following, recall that $\Sigma$ is a finite alphabet and assume, without loss of generality, that $\Sigma \cap [n] =\emptyset$.  

\begin{definition}\label{def:partialtree}
A \textit{partial tree} is a binary labeled tree $t$ over $\Sigma \cup [n]$. A node $u$ labeled by $t(u) \in [n]$ is called a \textit{hole} of $t$, and we assume that holes can appear only at the leaves of $t$. 
The \textit{full} size of $t$, denoted by $\virtsize{t}$, is defined as $|\{u \mid t(u) \in \Sigma\}| + \sum_{u\,:\, t(u) \in [n]} t(u)$. Moreover, a partial tree $t$ is said to be \textit{complete} if $t$ contains no holes.
\end{definition}
Intuitively, in a partial tree $t$, a hole $u$ represents a placeholder where a subtree of size $t(u)$ is going to be hanged. 
That is, partial tree $t$ is representing all trees over $\Sigma$ that have the same trunk as $t$ and, for each hole $u$, the subtree rooted at $u$ is of size $t(u)$. 
Notice that all trees represented by $t$ will have the same size $|\{u \mid t(u) \in \Sigma\}| + \sum_{u\,:\, t(u) \in [n]} t(u)$ and, therefore, we define the full size of $t$ as this quantity. 
Finally, observe that if a partial tree $t$ is complete, then $t$ contains no holes and, hence, no extension is needed. For an example of a partial tree, see Figure \ref{fig:partial_tree}.


\begin{figure}[t]
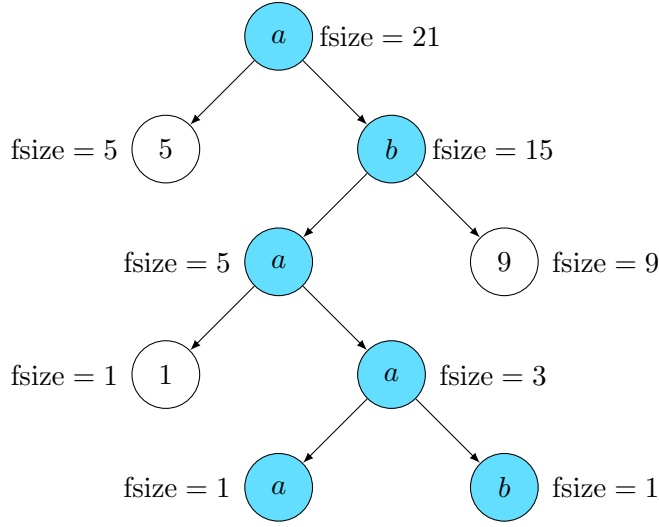

    \ctikzfig{partial_tree}
    \caption{An example of a partial tree. White nodes corresponds to holes, which are labeled by integers. Nodes that are not holes are labeled with symbols from $\Sigma=\{a,b\}$.}
    \label{fig:partial_tree}
\end{figure}

For every partial tree $t$ and node $x \in t$, write $t_x$ to denote the 
partial subtree of $t$ rooted at $x$. For each hole $u \in t$ with size $t(u) = i$, we say that $t'$ is an \textit{immediate extension} of $t$ over $u$ if $t' = t[u \rightarrow a(j, i-j-1)]$ for some $a\in \Sigma$ and $j \in [i-2]$. That is, $t$ is extended by replacing the label of $u$ with $a$ and hanging from $u$ two new holes whose sizes sum to $i-1$ (note that the resulting partial subtree $t_u'$ has full size $i$). 
In case that $i=1$, then it must hold that $t' = t[u \rightarrow a]$ for some $a\in \Sigma$.
We define the set of all immediate extensions of $t$ over $u$ as $\ext(t, u)$. Note that $|\ext(t, u)| = (i-2)|\Sigma|$.
Finally, given two partial trees $t$ and $t'$, we write $t \hookrightarrow_u t'$ if $t'$ is an immediate extension of $t$ over $u$, and $t \hookrightarrow t'$ if $t'$ is an immediate extension of $t$ over some hole $u \in t$. We then define the reflexive and transitive closure $\hookrightarrow^*$  of $\hookrightarrow$, and say that $t'$ is an \textit{extension} of $t$ if $t \hookrightarrow^* t'$. In other words, $t \hookrightarrow^* t'$ if either $t' = t$ or $t'$ can be obtained from $t$ via a non-empty sequence of immediate extensions $t \hookrightarrow t_1 \hookrightarrow t_2 \hookrightarrow \dots \hookrightarrow t'$. We say that $t'$ is a \textit{completion} of $t$ when $t \hookrightarrow^* t'$ and $t'$ is complete. 


\paragraph{Obtaining uniform samples from $T(\st^i)$.}
Given a partial tree $t$ with $\virtsize{t} = i$, consider now the set $T(\st^i, t)$ of all completions $t'$ of $t$ derivable with $\st^i$ as the state in the root node, namely, $T(\st^i, t) = \{t' \in T(\st^i) \, \mid\, \text{$t'$ is a completion of $t$}\}$.
Further, define $N(\st^i, t) = |T(\st^i, t)|$.
To obtain a uniform sample from $T(\st^i)$, we start with a partial tree $t = i$ (i.e. $t$ is a partial tree with one node, which is a hole of size $i$). 
At each step, we choose the hole $u \in t$ with the smallest size $t(u)$, and consider an immediate extension of $t$ over $u$.  Note that the set $T(\st^i, t)$ can be partitioned by the sets $\{T(\st^i, t')\}_{t \hookrightarrow_u t'}$ of such immediate extensions. The fact that $T(\st^i, t') \cap T(\st^i, t'') = \emptyset$, whenever $t \hookrightarrow_u t'$, $t \hookrightarrow_u t''$ and $t' \neq t''$,
follows immediately from the 
fact that 
$t'$ and $t''$ have different labels from $\Sigma$ in the place of $u$ or unequal sizes of the left and right subtrees of $u$. 
We then will sample each partition $T(\st^i, t')$ with probability approximately proportional to its size $N(\st^i, t')$, set $t' \leftarrow t$, and continue like that recursively. 
Formally, the procedure to sample a tree in $T(\st^i)$ is shown in Algorithm \ref{alg:sample}.

\begin{algorithm}[!h]
	\caption{\Sample$\big(\st^i,\,\{\aT_i(\str^j)\}_{\str \in S, j < i},\, \{\aN(\str^j)\}_{\str \in S, j \leq i},\, \eps,\, \delta\big)$ } \label{alg:sample}
	\smallskip
	Initialize a partial tree $t = i$, and set $\varphi = 1$\\ 
	\smallskip
	\While{$t$ is not complete}{
		\smallskip
		Let $u$ be the hole of $t$ with the minimum size $t(u)$. If more than one node reaches this minimum value, choose the first such a node according to a prespecified order on the holes of $t$. \label{algo:min-size-node}\\
		\smallskip
		Let $\ext(t, u) = \{t_1, \ldots, t_\ell\}$ be the set of immediate extensions of $t$ over $u$. \\
		\smallskip
		For each $k \in [\ell]$, call \EstimatePartitionSize$(t_k,\, \st^i,\, \{\aT_i(\str^j)\}_{\str \in S, j < i},\, \{\aN(\str^j)\}_{\str \in S, j \leq i},\, \eps,\, \delta)$ to obtain an estimate $\aN(\st^i, t_k)$ of $N(\st^i, t_k)$. 
		\hfill\tcp{Recall that $\cT = (S, \Sigma, \Delta, \sinit)$} \label{alg:sample:cps}
		\smallskip
		Sample partition $k \in [\ell]$ with probability $\frac{\aN(\st^i, t_k)}{\sum_{k'=1}^\ell \aN(\st^i, t_{k'})}$. \\
		\smallskip	
		Set $\varphi \leftarrow \varphi\cdot\frac{\aN(\st^i, t_k)}{\sum_{k'=1}^\ell \aN(\st^i, t_{k'})}$.\\ 
		\smallskip
		Set $t \leftarrow  t_k$. \medskip}  
		\smallskip
	\Return{$t$ with probability $\frac{1}{2\varphi \aN(\st^i)}$, otherwise output \FAIL}.
\end{algorithm}

Notice that \Sample$\big(\st^i,\,\{\aT_i(\str^j)\}_{\str \in S, j < i},\, \{\aN(\str^j)\}_{\str \in S, j \leq i},\, \eps,\, \delta \big)$ uses the precomputed values $\aT_i(\str^j)$ for every $\str \in S$ and $j \in [i-1]$, and the precomputed values $\aN(\str^j)$ for every $\str \in S$ and $j \in [i]$. This procedure 
first selects a hole $u$ with the minimum  size $t(u)$, and then calls a procedure $\EstimatePartitionSize$ to obtain an estimate $\aN(\st^i, t_k)$ of $N(\st^i, t_k)$ for every immediate extensions $t_k$ of $t$ over $u$. Thus, to prove our main theorem about the procedure \Sample, we first need  the following lemma about the correctness of the partition size estimates. The proof of Lemma~\ref{lem:estimatepart} is the main focus of Section \ref{sec:partition}. 

\begin{lemma}\label{lem:estimatepart}
	Let $\delta \in (0, 1/2)$, and fix independent and uniform samples sets $\aT_i(s^j)$ of $T(s^j)$ each of size $O(\log^2(\delta^{-1}) (nm)^{13} /\eps^5)$, for every $\st \in \St$ and $j < i$. Suppose further that we have values $\aN(s^j) = (1 \pm j\eps)N(s^j)$ for every $\st\in S$ and $j \leq i$ Then with probability $1-\delta^{nm}$, the following holds: for every state $\st \in \St$ and for every partial tree $t$ with $\virtsize{t} = i$, the procedure 
	\[\EstimatePartitionSize\left(t,\,\st^i,\left\{\aT_i(r^j)\right\}_{r \in S, j < i},\,\left\{\aN(r^j)\right\}_{r \in S, j \leq i},\, \eps,\, \delta\right)\] runs in $\poly(n,m,\eps^{-1},\log \delta^{-1})$-time and returns a value $\aN(\st^i, t)$ such that 
\[ 
\aN(\st^i, t ) \ = \ \left(1 \pm  (4nm)^{17}\eps\right) N\left(\st^i, t\right)	
\]
\end{lemma}

Notice that the guarantee of Lemma \ref{lem:estimatepart} is that  \EstimatePartitionSize always runs in polynomial time. Conditioned on the success of Lemma \ref{lem:estimatepart}, lines $1-9$ of the algorithm \Sample
always run in polynomial time (i.e, with probability $1$). Thus, conditioned on Lemma \ref{lem:estimatepart}, only on line $10$ of \Sample is it possible for the algorithm to output \FAIL. In the following Lemma, we demonstrate that, given the success of Lemma \ref{lem:estimatepart}, the $\Sample$ algorithm produces \textit{truly} uniform samples, and moreover the probability of outputting $\FAIL$ is at most a constant. Notice that this implies that, after running the inner loop of $\Sample$ a total of $O(\log \delta^{-1})$ times, we will obtain a sample with probability at least $1-\delta$. 


\begin{lemma}\label{lem:sampmain}
	 Given $\delta \in (0, 1/2)$, $\{\aT_i(\str^j)\}_{\str \in S, j < i}$ and $\{\aN(\str^j)\}_{\str \in S, j \leq i}$, suppose that the procedure 
	 \[\EstimatePartitionSize\left(t,\,\st^i,\,\left\{\aT_i(\str^j)\right\}_{\str \in S, j < i},\, \left\{\aN(\str^j)\right\}_{\str \in S, j \leq i},\, \eps,\, \delta\right)\]
	  produces an estimate $\aN(\st^i, t)$ with $\aN(\st^i, t) = (1 \pm (4nm)^{17}\eps) N(\st^i, t)$ for every partial tree $t$ of size $i$ and state $\st^i$.
	  
Further suppose that Property $1$ holds for all $j \leq i$ (see page \pageref{prop-one}), and $n \geq 2$.
	Then conditioned on not outputting \FAIL, each call to the procedure 
	\[\Sample\left(\st^i,\,\left\{\aT_i(\str^j)\right\}_{\str \in S, j < i},\, \left\{\aN(\str^j)\right\}_{\str \in S, j \leq i},\, \eps,\, \delta\right)\]
	produces an independent, uniform sample $t \sim T(\st^i)$. Moreover, the probability that a given call 
	outputs \FAIL\ is at most $3/4$, and the 
	number of times \EstimatePartitionSize\ is called in each iteration of the loop is at most $nm$.
\end{lemma}
\begin{proof}
	Fix a tree $t \in T(\st^i)$. Then there is a unique sequence of partial trees $i = t_0 ,t_1,t_2,\dots,t_i = t$ such that $T(\st^i) = T(\st^i, t_0) \supseteq T(\st^i, t_1) \supseteq T(\st^i, t_2) \supseteq \dots \supseteq T(\st^i, t_i) = \{t\}$, which gives a sequence of nested partitions which could have been considered in the call \Sample$(\st^i,\,\{\aT_i(\str^j)\}_{\str \in S, j < i},$ $\{\aN(\str^j)\}_{\str \in S, j \leq i},\, \eps,\, \delta)$. 
	For $j \in [i]$, let $p_j$ be the true ratio of $\frac{N(\st^i, t_{j})}{N(\st^i, t_{j-1})}$, which is the probability that we should have chosen partition $T(\st^i, t_j)$ conditioned on being in partition $T(\st^i, t_{j-1})$. Note that $\prod_{j=1}^i p_j = \frac{1}{|T(\st^i)|} = \frac{1}{N(\st^i)}$. Now assuming \EstimatePartitionSize\  always returns an estimate with at most $(1 \pm (4nm)^{17}\eps)$-relative error, it follows that conditioned on being in partition $T(\st^i, t_{j-1})$, we chose the partition $T(\st^i, t_{j})$ with probability $\tilde{p}_j = (1 \pm (4nm)^{17} \eps )p_j$. Thus the probability that we choose $t$ at the end of the loop in step $2$ of the \Sample\ procedure is:
	\begin{multline*}
	\varphi \ = \ \prod_{j=1}^i \tilde{p}_j
	\  = \ (1 \pm (4nm)^{17}\eps)^i \prod_{j=1}^i p_j
	\ = \ (1\pm 1/n) \prod_{j=1}^i p_j
	\ = \ (1\pm 1/n) \frac{1}{N(\st^i)} 
	\ = \\ (1\pm 1/n) (1 \pm i \eps) \frac{1}{\aN(\st^i)}
	\ = \ (1\pm 1/n) (1 \pm 1/(4n)) \frac{1}{\aN(\st^i)} 
	\ = \ (1 \pm 2/n) \frac{1}{\aN(\st^i)}.
	\end{multline*}
	Notice that we use the fact that $\eps < (4nm)^{-18}$ and that Property 1 holds for all $j \leq i$. 
	The probability that we do not output \FAIL\ can be bounded by
	 \begin{eqnarray*}
	 \frac{1}{2\varphi \aN(\st^i)} \ =\ \frac{1}{2 \aN(\st^i)}\prod_{j=1}^i \frac{1}{\tilde{p}_j} \ \geq\  \frac{1}{2 \aN(\st^i)} \frac{\aN(\st^i)}{(1 + 2/n)} \ \geq \ \frac{1}{2(1 + 2/n)} \ \geq\  1/4
	 \end{eqnarray*}
	 since $n \geq 2$, 
	 which completes the proof that the probability that the call 
	 \Sample$(\st^i,\,\{\aT_i(\str^j)\}_{\str \in S, j < i},$ $\{\aN(\str^j)\}_{\str \in S, j \leq i},\, \eps,\, \delta)$ 
	 outputs \FAIL\ is at most $3/4$. For the uniformity claim, note that we accept $t$ at the end with probability $\varphi \cdot \frac{1}{2\varphi \aN(\st^i)} = \frac{1}{2 \aN(\st^i)}$, which is indeed uniform conditioned on not outputting \FAIL, as it does not depend on $t$.
	Finally, notice that \EstimatePartitionSize\ is called at most $(t(u) - 2)\cdot |\Sigma| \leq nm$ times in each iteration of the loop.
\end{proof}


\paragraph{An FPRAS for $\stta$ and $\sta$.}
We show in Algorithm \ref{alg:fpras} a fully polynomial-time approximation schema for $\stta$, which puts together the different components mentioned in this section. The correctness of this algorithm is shown in the following theorem. 

\phantom{please do not remove, this is solve a problem with the footnote of the algorithm}

\begin{algorithm}[!h]
	\SetKwInOut{Input}{Input}
	\caption{\astta$(\cT, 0^n, \eps, \delta)$ } \label{alg:fpras}
	Set $m \leftarrow |\cT|$\\
	\smallskip
	\If{$n < 2$ or $m<3$} {
	\smallskip
Edge case,  $|\cL_n(\cT)|$ can be exactly computed  (Remark \ref{rem:size})
	}
	Construct the tree automaton $\ocT$\\
	\smallskip
	Set $\eps \leftarrow  \min\{\eps, 1/(4mn)^{18}-1\}$ \\
	\smallskip
	Set $\gamma = \log(1/\delta) + 2n$  \label{alg:fpras:gamma}\\
	\smallskip
	Set $\alpha \leftarrow O(\log^2(1/\delta) (nm)^{13} /\eps^5)$,\ $\textit{total} \leftarrow O(\alpha)$ \\
	\smallskip
	For each $\st\in \St$, compute $N(\st^1)$ exactly and set $\aN(s^1) \leftarrow N(s^1)$ \label{alg:fpras:ns1} \\
	\smallskip
	For each $\st\in \St$, create set $\aT_2(\st^1)$ with $\alpha$ uniform, independent samples from $T(\st^1)$ \label{alg:fpras:ts1} \\
	\smallskip
	\For{$i=2, \dots, n$}{
		\smallskip
		For each $\st\in\St$, compute $\aN(\st^i)$ such that $\pr{\aN(\st^i)=(1 \pm i\eps)N(\st^i)}\geq 1-\exp(-\gamma n^{20})$ \\
		\smallskip
		\If{$i < n$}{
		\For{each $\st\in\St$ and $j=1, \dots, i$}{
			\smallskip
			Set $\aT_{i+1}(\st^j) \leftarrow \emptyset$,  $\textit{counter} \leftarrow 1$\\
			\smallskip
			\While{$|\aT_{i+1}(\st^j)|<\alpha$ and $\textit{counter} \leq \textit{total}$}{
				\smallskip
				Call 
				the procedure 
				\Sample$\big(\st^j,\{\aT_i(\str^k)\}_{\str \in S, k < j}, \{\aN(\str^k)\}_{\str \in S, k \leq j},\, \eps,\, 2^{-2n}\delta\big)$ \label{alg:fpras:sample}\\
				\smallskip
				If this procedure returns a tree $t$, then 
				set $\aT_{i+1}(\st^j)\leftarrow \aT_{i+1}(\st^j)\cup\{t\}$\smallskip\\
			Set $\textit{counter} \leftarrow \textit{counter} +1$\smallskip
			}\smallskip
			\If{$|\aT_{i+1}(\st^j)|<\alpha$}{\smallskip
			\Return{\FAIL} \label{alg:fpras:fail}\smallskip
			}\smallskip
		}\smallskip
		}\smallskip
	} 
	\smallskip
	\Return{$\aN(\sinit^n)$}.
\end{algorithm}

\phantom{please do not remove, this is solve a problem with the footnote of the algorithm}

\begin{theorem}\label{theo:fpras-bta}
	Let $\eps,\delta \in (0,1/2)$, $n \geq 1$, $\cT = (\St, \Sigma, \Delta, \sinit)$ be a tree automaton, and $m = \|\Delta\|$ be the size of $\cT$. Then the call \astta$(\cT, 0^n, \eps, \delta)$\footnote{Here we write $0^n$ as the unary representation of $n$. Since our algorithms are polynomial in $n$, the algorithm is polynomial in the size of the input. } returns, 
	with probability at least $1-\delta$,
	a value $\aN$ such that 
	$\aN = (1 \pm \eps)|\cL_n(\cT)|$.
	Moreover, the runtime of the algorithm \astta\ is~$\poly(n,m,1/\eps,\log(1/\delta))$.
\end{theorem}



\begin{proof}
	Set $\alpha = O(\log^2(1/\delta) (nm)^{13} / \eps^5)$. 
	For every $j \in [n]$, let $\mathcal{E}_j^1$ denote the event that Property $1$ holds for level $j$, and similarly define $\mathcal{E}_j^2$ for Property $2$. Set $\gamma = \log(1/\delta)$. We prove inductively that 
	\begin{eqnarray*}
	\bpr{\bigwedge_{j \leq i}\left( \mathcal{E}_j^1 \wedge \mathcal{E}_j^2\right)} & \geq & 1- 2^{-\gamma +2i}
	\end{eqnarray*}
	 for each $i \in[n]$. Since $N(s^1)$ is computed exactly in step \ref{alg:fpras:ns1} of \astta$(\cT, 0^n, \eps, \delta)$ and the size of each tree in $T(s^1)$ is 1, 
	  the base case $i = 1$ trivially holds. Now at an arbitrary step $i \geq 2$, suppose $\mathcal{E}_j^1 \wedge \mathcal{E}_j^2$ holds for all $j < i$. By considering $\exp(-\gamma n^{20}  )$ as the value for the parameter $\delta$ in Proposition \ref{prop:prop1},
	  it follows that $\mathcal{E}_i^1$ holds with probability at least $1-\exp(-\gamma n^{20})$, and the runtime to obtain Property 1 is $\poly(n,m, \frac{1}{\eps}, \gamma)$.  Thus,
	   \begin{eqnarray*}
	   \bpr{\mathcal{E}_i^1 \; \big| \; \bigwedge_{j < i}\left( \mathcal{E}_j^1 \wedge \mathcal{E}_j^2\right)}&  \geq & 1-\exp(- \gamma n^{20}). 
	\end{eqnarray*}
We now must show that $\mathcal{E}^2_i$ holds -- namely, that we can obtain uniform samples from all sets $T(s^i)$. By Lemma \ref{lem:estimatepart}, if we can obtain fresh uniform sample sets $\aT_i(s^j)$ of $T(s^j)$ for each $\st \in S$ and $j < i$, each of size $\alpha$, then with probability at least $1-2^{-\gamma nm}$, we have that for every partial tree $t'$ of size $i$ (that is, $\virtsize{t'} = i$) and state $\st \in S$, the procedure \EstimatePartitionSize$(t',\, s^i,\, \{\aT_i(s^j)\}_{\st \in S, j < i},$ $\{\aN(s^j)\}_{\st \in S, j \leq i},\, \eps,\, \delta )$  produces an estimate $\aN(s^i, t')$ such that $\aN(s^i, t') = (1 \pm (4nm)^{17}\eps)N(s^i,t')$. Since we have to call \EstimatePartitionSize\ at most $inm$ times (see Lemma \ref{lem:sampmain}), after a union bound we get that the conditions of Theorem \ref{lem:sampmain} are satisfied with probability at least $1-2^{-\gamma}$, and it follows that we can sample uniformly from the set $T(s^i)$ for each $\st \in S$ in polynomial time.
	
	It remains to show that we can obtain these fresh sample sets $\aT_i(\st^j)$ of $T(\st^j)$ for each $\st \in S$ and $j < i$ in order to condition on the above. But the event $\mathcal{E}_{i-1}^2$ states precisely that can indeed obtain such samples in $\poly(n,m,\frac{1}{\eps},\gamma)$ time per sample. Thus the conditions of the above paragraph are satisfied, so we have 
	\begin{eqnarray*}
	\bpr{\mathcal{E}_i^2 \; \big| \; \mathcal{E}_i^1 \wedge  \bigwedge_{j < i}\left( \mathcal{E}_j^1 \wedge \mathcal{E}_j^2\right)} & \geq & 1-2^{- \gamma}. 
	\end{eqnarray*}
	Therefore, we conclude that
	\begin{eqnarray*}
	\bpr{\mathcal{E}_i^1 \wedge \mathcal{E}_i^2 \; \big| \; \bigwedge_{j < i}\left( \mathcal{E}_j^1 \wedge \mathcal{E}_j^2\right)} \ \geq \ 1-2^{- \gamma} - \exp(-\gamma n^{20}) + 2^{- \gamma}\exp(-\gamma n^{20}) \ \geq \ 1 - 2^{-\gamma + 1}
	\end{eqnarray*}
	Hence, by induction hypothesis:
	\begin{eqnarray*}
	\bpr{\bigwedge_{j \leq i}\left( \mathcal{E}_j^1 \wedge \mathcal{E}_j^2\right)} &=& \bpr{\mathcal{E}_i^1 \wedge \mathcal{E}_i^2 \; \big| \; \bigwedge_{j < i}\left( \mathcal{E}_j^1 \wedge \mathcal{E}_j^2\right)} \cdot \bpr{\bigwedge_{j <i}\left( \mathcal{E}_j^1 \wedge \mathcal{E}_j^2\right)} \\
	&\geq& (1-2^{-\gamma+1})(1-2^{-\gamma + 2(i-1) })\\
	&=& 1-2^{-\gamma+1}   -  2^{-\gamma + 2i -2} + 2^{-2 \gamma + 2i -1}\\
	&\geq& 1-2^{-\gamma+1}   -  2^{-\gamma + 2i -1}\\
	&\geq& 1-2^{-\gamma+2i-1}   -  2^{-\gamma + 2i -1}\\
	&=&  1-2^{-\gamma+2i}
	\end{eqnarray*}
	which completes the inductive proof.
	Redefining $\gamma = \log(1/\delta) + 2n$ (see Line \ref{alg:fpras:gamma} of Algorithm \ref{alg:fpras:sample}) and considering $2^{-2n} \delta$ when using Lemma \ref{lem:estimatepart} (see Line \ref{alg:fpras:sample} of Algorithm \ref{alg:fpras:sample}), we obtain that the success probability of the overall algorithm is $1-\delta$ as needed. 
	
	For runtime, note that by Lemma \ref{lem:sampmain}, the expected number of trials to obtain $\alpha$ samples $\aT_i(s^j)$ for each $s^j \in \overline{S}$ and $i \in [n]$ is $O(\alpha)$, and thus is $O(\alpha)$ with probability $1-2^{-\alpha} > 1-2^{-mn\gamma}$ by Chernoff bounds.
	That is, with $O(\alpha)$ trials, we have probability at least $1-2^{-mn\gamma}$ of not failing in step \ref{alg:fpras:fail} of Algorithm \ref{alg:fpras}. Since we go through that step at most $O(n^2m)$ times during the whole run of the algorithm, that means that the overall probability of returning \FAIL$\>$ can be bounded by $1-2^{\gamma}=1-\delta$, which is a loose bound but enough for our purposes. Moreover, by Lemma \ref{lem:sampmain}, the runtime of each sampling trial in step \ref{alg:fpras:sample} of Algorithm \ref{alg:fpras} is polynomial in $n$, $m$, $1/\eps$ and $\log(1/(2^{-2n}\delta)) = \gamma$. It follows that the entire algorithm runs in $\poly(n,m,1/\eps,\log(1/\delta))$ time, which completes the proof.
	\end{proof}

We now provide our main theorem for uniformly sampling from tree automata. The notion of sampling we get is in fact stronger than the definition of a FPAUS as defined in Section~\ref{sec:preliminaries}. Specif
\begin{theorem}\label{thm:samplemain}
	Let $\delta \in (0,1/2)$, $n \geq 1$, $\cT = (\St, \Sigma, \Delta, \sinit)$ be a tree automaton, and $m = \|\Delta\|$ be the size of $\cT$. Then there is a sampling algorithm $\mathcal{A}$ and a pre-processing step with the following property. The preprocessing step runs in $\poly(n,m,\log \delta^{-1})$ time, and with probability $1-\delta$ over the randomness used in this pre-processing step,\footnote{Note that we cannot detect if the event within the preprocessing step that we condition on here fails, which occurs with probability $\delta$.} each subsequent call to the algorithm $\mathcal{A}$ runs in time $\poly(n,m,\log \delta^{-1})$ time, and returns either a uniform sample $t \sim \cL_n(\cT)$ or \FAIL. Moreover,  if $\cL_n(\cT) \neq \emptyset$, the probability that the sampler returns \FAIL\ is at most $1/2$. 
	
Additionally, this implies that there is an FPAUS for $\cL_n(\cT)$ as defined in Section~\ref{sec:preliminaries}.

\end{theorem}
\begin{proof}
The preprocessing step here is just the computation of the estimates $\wt{N}(s^i)$ and sketches $\wt{T}(s^i)$ for all $i \leq n$, which are obtained by a single to the FPRAS of Theorem \ref{theo:fpras-bta} using a fixed $\eps = (nm)^{-C}$ for a sufficiently large constant $C$. The sampling algorithm is then just a call to 	\[\Sample\left(\st^n_{\text{init}},\,\left\{\aT_i(\str^j)\right\}_{\str \in S, j < n},\, \left\{\aN(\str^j)\right\}_{\str \in S, j \leq n},\, \eps,\, \delta\right)\]

 Then the first result follows from Lemmas \ref{lem:sampmain} and \ref{lem:estimatepart}, as well as Theorem \ref{theo:fpras-bta}. In particular, if we condition on the success of Theorem \ref{theo:fpras-bta}, which hold with probability $1-\delta$, then by the definition of Property 2, and the fact that Property $2$ holds for the size $n$ conditioned on  Theorem \ref{theo:fpras-bta}, this is sufficient to guarantee that the samples produced by our sampling procedure are uniform. Furthermore, by Lemma \ref{lem:sampmain}, the probability that a call to $\Sample$ outputs $\FAIL$ is at most $3/4$, so repeating the algorithm three times, the probability that a sample is not output is at most $(3/4)^3 < 1/2$. Finally, by Lemma \ref{lem:estimatepart} the runtime of each call to $\EstimatePartitionSize$ is at most $\poly(n,m,\log \delta^{-1})$, and by Lemma \ref{lem:sampmain}, the procedure $\EstimatePartitionSize$ is called at most $nm$ times per call to $\Sample$, which completes the proof of the runtime. 

We now verify that the above implies that $\cL_n(\cT)$ admits an FPAUS. First,  if $\cL_n(\cT) \neq \emptyset$, note that the probability $\delta$ of failure of the pre-processing algorithm induces an additive $\delta$ difference in total variational distance from the uniform sampler. Moreover, by testing deterministically whether the tree $t$ obtained by the algorithm is contained in $\cL_n(\cT)$, we can ensure that, conditioned on not outputting $\FAIL$, the output of the algorithm is supported on $\cL_n(\cT)$.  We can then run the algorithm with $\delta_0 = \delta|\cL_n(\cT)|^{-1} = \delta\exp(-\poly(n,m))$, which does not affect the stated polynomial runtime.  This results~in
  \begin{equation}\label{eqn:additiverror}
  \begin{split}
  	\mathcal{D}(t)& =  \frac{1}{|\cL_n(\cT)|} \pm \delta_0 \\
	& =  \frac{1}{|\cL_n(\cT)|} \pm \delta|\cL_n(\cT)|^{-1}\\
  	&=\frac{(1 \pm \delta)}{|\cL_n(\cT)|} \\
  \end{split}
  \end{equation}
  for every $t \in \cL_n(\cT)$, as desired.  To deal with the $1/2$ probability that the output of our sampler is \FAIL, we can run the sampler for a total of $\Theta(\log(\delta^{-1}|\cL_n(\cT)|)) = \poly(n,m ,\log \delta^{-1})$ trials, and return the first sample obtained from an instance that did not return $\FAIL$. If all trails output $\FAIL$, then we can also output $\bot$ as per the specification of an FPAUS. By doing so, this causes another additive $\delta|\cL_n(\cT)|^{-1}$ error in the sampler, which is dealt with in the same way as shown in Equation \eqref{eqn:additiverror} above. Finally, if  $\cL_n(\cT)= \emptyset$, the algorithm must always output $\bot$, since given any potential output $t \neq \bot$, we can always test if $t \in \cL_n(\cT)$ in polynomial time, which completes the proof that the algorithm yields an FPAUS. 
  
\end{proof}

We conclude this section by pointing out that from Theorem \ref{theo:fpras-bta} and the existence of a polynomial-time parsimonious reduction from $\sta$ to $\stta$, and the fact that given a binary tree $t' \in \cL_n(\cT')$ after the reduction from a tree automata $\cT$ to a binary tree automata $\cT'$, the corresponding original tree $t \in \cL_n(\cT)$  can be reconstructed in polynomial time, we obtain the following corollary:
\begin{corollary}\label{cor:fpras-ta-bta}
Both $\stta$ and $\sta$ admit an FPRAUS and an FPAUS. 
\end{corollary}


%% file: partition-size2.tex

The goal of this section is to prove Lemma \ref{lem:estimatepart}, namely, to show how to implement the procedure \EstimatePartitionSize. 
To this end, we first show  how \EstimatePartitionSize\ can be implemented by reducing it 
to the problem of counting words accepted by a succinct NFA, which we introduced in Section \ref{sec:tech} and formally define here. Next, we demonstrate an FPRAS for counting words accepted by a succinct NFA, which will complete the proof of Lemma \ref{lem:estimatepart}.

\paragraph{Succinct NFAs.} 
Let $\Gamma$ be a finite set of labels. A {\em succinct NFA}
over $\Gamma$ is a $5$-tuple \linebreak $\mathcal{N} = (S, \Gamma, \Delta, \sinit, \sfin)$ where $S$ is the set of states and each transition is labeled by a subset of $\Gamma$, namely $\Delta \subseteq S \times 2^{\Gamma} \times S$. Thus each transition is of the form $(s,  A, s')$, where $A \subseteq \Gamma$.
For each transition $(s, A, s') \in \Delta$, the set $A \subseteq \Gamma$ is given in some 
representation (e.g. a tree automaton, a DNF formula, or an explicit list of elements), and we write $\rsize{A}$ to denote the size of the representation. Note that while the whole set $A$ is a valid representation of itself, generally the number of elements of $A$, denoted by $|A|$, will be exponential in the size of the representation $\rsize{A}$.
We define the size of the succinct NFA $\mathcal{N}$ as $|\mathcal{N}| = |S|+|\Delta|+\sum_{(s,A,s') \in \Delta} \rsize{A}$. For notational simplicity, we will sometimes write $r = |\mathcal{N}|$.

Given a succinct NFA $\mathcal{N}$ as defined above and elements $w_1, \ldots, w_n \in \Gamma$, we say that $\mathcal{N}$ accepts the word $w_1 w_2 \ldots w_n$ if there exist states $s_0, s_1, \ldots, s_n \in S$ and sets $A_1, \ldots, A_n \subseteq \Gamma$ such that:
\begin{itemize}
	\item $s_0 = \sinit$ and $s_n = \sfin$
	\item $w_i \in A_i$ for all $i=1 \ldots n$
	\item $(s_{i-1}, A_i, s_i) \in \Delta$ for all $i=1 \ldots n$
\end{itemize}

We denote by $\cL_k(\mathcal{N})$ the set of all words of length $k$ accepted by $\mathcal{N}$. We consider the following general counting problem:
\begin{center}
	\framebox{
		\begin{tabular}{ll}
			\textbf{Problem:} &  $\slp$\\
			\textbf{Input:} & $k \geq 1$ given in unary and a succinct NFA $\mathcal{N}$\\
			\textbf{Output:} & $|\cL_k(\mathcal{N})|$
		\end{tabular}
	}
\end{center}
\paragraph{Reduction to Unrolled Succinct NFAs}
Our algorithm for approximating $|\cL_k(\mathcal{N})|$ first involves unrolling $k$ times the NFA $\mathcal{N} = (S, \Gamma, \Delta, \sinit, \sfin)$, to generate an unrolled NFA $\nun^k$.
Specifically, for every state $p \in S$ create $k-1$ copies $p^1, p^2,\ldots,p^{k-1}$ of $p$, and include them as states of the unrolled NFA $\nun^k$. Moreover, for every transition $(p,A,q)$ in $\Delta$,
create the edge $(p^\alpha, A, q^{\alpha+1})$ in $\nun^k$, for every $\alpha \in \{1, \ldots, k-2\}$. Finally, if $(\sinit,A,q)$ is a transition in $\Delta$, then $(\sinit,A,q^1)$ is a transition in $\nun^k$, while if $(p,A,\sfin)$ is a transition in $\Delta$, then $(p^{k-1},A,\sfin)$ is a transition in $\nun^k$. In this way, we keep $\sinit$ and $\sfin$ as the initial and final states of $\nun^k$, respectively.
Since $k$ is given in unary, it is easy to see that $\nun^k$ can be constructed in polynomial time from $\mathcal{N}$. Thus, for the remainder of the section, we will assume that the input succinct NFA $\mathcal{N}$ has been unrolled according to the value $k$. Thus, we consider the following problem.
\begin{center}
	\framebox{
		\begin{tabular}{ll}
			\textbf{Problem:} &  $\salp$\\
			\textbf{Input:} & $k \geq 1$ given in unary and an unrolled succinct NFA $\nun^k$\\
			\textbf{Output:} & $|\cL_k(\nun^k)|$
		\end{tabular}
	}
\end{center}
Clearly in the general case, without any assumptions on our representation $\|A\|$ of $|A|$, it will be impossible to obtain polynomial in $|\mathcal{N}|$ time algorithms for the problem above. In order to obtain polynomial time algorithms, we require the following four properties of the label sets $A$ to be satisfied.
The properties state that the sizes $|A|$ are at most singly exponential in $|\mathcal{N}|$, we can efficiently test whether an element $a \in \Gamma$ is a member of $A$, we can obtain approximations of $|A|$, and that we can generate almost uniform samples from $A$.
%
%
\begin{definition}[Required properties for a succinct NFA]\label{def:prop} Fix $\eps_0 > 0$.  Then for every label set $A$ present in $\Delta$, we have:
	
	\begin{enumerate}
		\item \textbf{Size bound:} There is a polynomial $g(x)$ such that $|A| \leq 2^{g(|\mathcal{N}|)}$.
		\item \textbf{Membership:} There is an algorithm that given any $a \in \Gamma$, verifies in time  $T = \poly(|\mathcal{N}|)$ whether 
		$a \in A$.
		\item \textbf{Size approximations:} We have an estimate $\aN(A) = (1\pm\epsilon_0)|A|$.
		\item \textbf{Almost uniform samples:} We have an oracle which returns independent samples $a \sim A$ from a distribution $\mathcal{D}$ over $A$, such that for every $a \in A$:
\begin{eqnarray*}
\DD(a) &=& (1 \pm \eps_0)\frac{1}{|A|}
\end{eqnarray*} \qed
	\end{enumerate}
\end{definition}
The reason for the first condition is that our algorithms will be polynomial in $\log(N)$, where $N$ is an upper bound on the size of $|\cL_k(\mathcal{N})|$. 
We remark that for the purpose of our main algorithm, we actually have truly uniform samples from each set $A$ that is a label in a transition. However, our results may be applicable in other settings where this is not the case. In fact, along with the first two conditions from Definition \ref{def:prop}, a sufficient condition for our algorithm to work is that the representations of each set $A$ allows for an FPRAS and a polynomial time almost uniform sampler. 

\paragraph{The Main Path of a Partial Tree}
Next we show that if we can approximate the number of words of a given length accepted by a succinct NFA, then we can 
implement the procedure {\EstimatePartitionSize}. 
But first we need to introduce the notion of 
main path of a partial tree. %
Let $t$ be a partial tree constructed via the partitioning procedure of Algorithm~\ref{alg:sample} (see page \pageref{alg:sample}). Given that we always choose the hole with the minimal size in Line~\ref{algo:min-size-node}, one can order the holes of $t$ as $u_1,u_2,\dots,u_k$, such that for each $i$, $\parent{u_i}$ is an ancestor of $u_{i+1}$, namely, $u_{i+1},\ldots,u_{k}$ are contained in the subtree rooted at the parent of $u_i$. Note that by definition of the loop of Algorithm~\ref{alg:sample}, it could be the case that two holes $u$ and $v$ share the same parent (e.g. the last step produced a subtree of the form $a(i, j)$). If this is the case, we order $u$ and $v$ arbitrarily.
Then we define the \textit{main path} $\pi$ of $t$ considering two cases. If no two holes share the same parent, then $\pi$ is the path $\parent{u_1}, \parent{u_2}, \ldots, \parent{u_k}$ (from the most shallow node $u_1$ 
to the deepest node $u_k$). On the other hand, if two holes share the same parent, then by definitions of Algorithm~\ref{alg:sample} and sequence $u_1$, $\ldots$, $u_k$, these two nodes must be $u_{k-1}$ and $u_k$. In this case, we define $\pi$ as the path $\parent{u_1}, \parent{u_2}, \ldots, \parent{u_{k-1}},u_k$, (again, from the most shallow node $u_1$ to the deepest node $u_k$).\footnote{Strictly speaking, $\parent{u_1}, \ldots, \parent{u_{k}}$ is a sequence and is not necessarily a path in the tree, because there could be missing nodes between the elements of the sequence. However, for the purpose of the proof the missing nodes do not play any role and will be omitted.} 
Observe that by definition of Algorithm~\ref{alg:sample}, every hole 
 is a child of some node in $\pi$.
 We illustrate the notion of main path in Figure~\ref{fig:fundamental_path}. For the partial tree in the left-hand side, we have that the main path is  $\parent{H_1}, \parent{H_2}, \parent{H_3}, \parent{H_4}$ as no two holes share the same parent. On the other hand, the main path for the partial tree in the right-hand side is $\parent{H'_1}, \parent{H'_2}, \parent{H'_3}, H'_4$, as in this case holes $H'_3$ and $H'_4$ share the same parent.

\begin{figure}
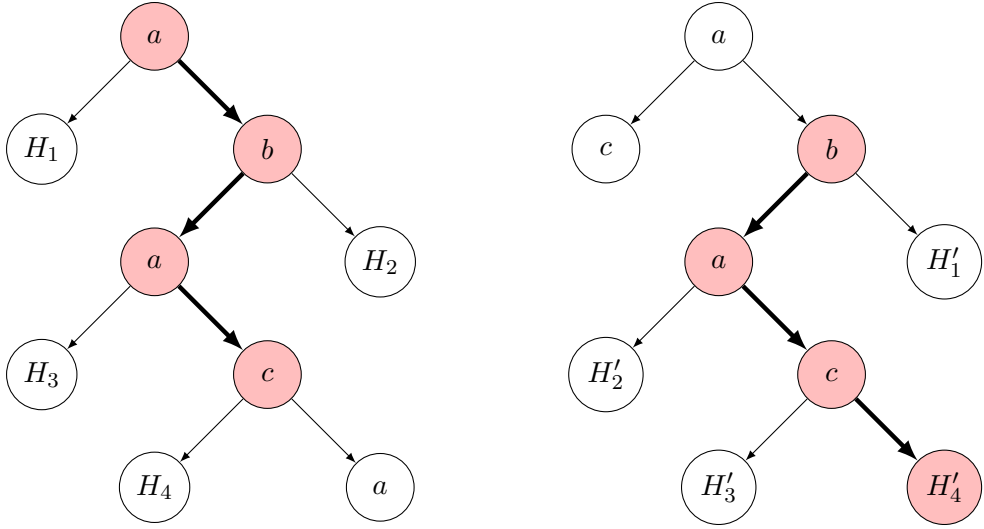

    \ctikzfig{fund_path}
    \caption{Two examples of a partial tree. The holes are indicated by the letters $H$ and $H'$, while the nodes that are not holes have labels $a, b,c \in \Sigma$. 
    Non-white nodes and thick arcs are used to highlight the main paths.}
    \label{fig:fundamental_path}
\end{figure}


	
\begin{lemma}\label{lem:treetograph}
	There exists a polynomial-time algorithm that, given a tree automaton $\cT$, a partial tree $t$ with $k$ holes
	constructed via the partitioning procedure of Algorithm \ref{alg:sample} with $i = \virtsize{t}$ and state $s$ of $\cT$, returns a succinct NFA $\mathcal{N}$ such that 
	\begin{eqnarray*}
	|T(\st^i, t)| & = & |\cL_k(\mathcal{N})|.
	\end{eqnarray*}
	Moreover, $|\mathcal{N}| \leq 3(i m)^4$, where $m$ is the size of $\cT$. 
\end{lemma}
\begin{proof}
Let $u_1, \ldots, u_k$ be the holes of $t$. Assume first that no two holes of $t$ share the same parent, so that $\pi = p_1, p_2, \ldots, p_k$ is the main path of $t$ with $p_i = \parent{u_i}$.
Counting the number of elements of $T(\st^i, t)$ is the same as counting all sequences of trees $t_1, \ldots, t_k$ over $\Sigma$ such that there exists a run $\rho$ of $\ocT$ over the tree $t[u_1 \rightarrow t_1] \cdots [u_k \rightarrow t_k]$ with $\rho(\lambda) = \st^i$. In other words, we hang $t_1$ on $u_1$, \ldots, $t_k$ on $u_k$ to form a tree that is accepted by $\ocT$ when $\st^i$ is the initial state. 
Then the plan of the reduction is to produce a succinct NFA $\mathcal{N}$ such that all words accepted by $\mathcal{N}$ are of the form $t_1 \cdots t_k$ with $t[u_1 \rightarrow t_1] \cdots [u_k \rightarrow t_k] \in T(\st^i, t)$.

For the construction of $\mathcal{N}$ it will be useful to consider the following extension of $\ocT$ over partial trees. Recall the definition of $\ocT = (\oSt, \Sigma, \oDelta, \sinit)$ from Section~\ref{sec:fpras}, but assuming here that the unfolding is done for $i$ levels. Then define $\ocT^* = (\oSt, \Sigma \cup [i], \oDelta^*, \sinit^i)$ such that $\oDelta^* = \oDelta \cup \{(s^j, j, \lambda) \mid s^j \in \oSt\}$, namely we add to $\ocT$ special transitions over holes when the level $j$ of $s^j$ coincides with the value of the hole. 
Intuitively, if we have a run $\rho$ of $\ocT^*$ over $t$ with $\rho(\lambda) = s^i$ and $T(\rho(u_\ell)) \neq \emptyset$ for every $\ell \in [k]$, then $t$ can be completed with trees $t_1 \in T(\rho(u_1)), \ldots, t_k \in T(\rho(u_k))$ such that $t[u_1 \rightarrow t_1] \cdots [u_k \rightarrow t_k] \in T(\st^i, t)$.

Let $i_1, \ldots, i_k$ be the sizes $t(u_1), \ldots, t(u_k)$ on the holes $u_1, \ldots, u_k$, respectively. Furthermore, let $j_1, \ldots, j_k$ be the final sizes of the subtrees of $t$ hanging from nodes $p_1, \ldots, p_k$, respectively. That is, if $t_1$ is the subtree hanging from $p_1$ in $t$, then $j_1 = \virtsize{t_1}$, and so on. 
Note that by the definition of the main path $\pi$, we have that $j_1 > j_2 > \ldots > j_k$ (since each $p_i$ is the parent of $p_{i+1}$). 
We now have the ingredients to define the
succinct NFA $\mathcal{N} = (S_\mathcal{N}, \Gamma, \Delta_\mathcal{N}, s_0, s_e)$. The set $S_\mathcal{N}$ of states will be a subset of the states of $\oSt$, plus two additional states $s_0$ and $s_e$, formally, $S_\mathcal{N} = \bigcup_{\ell=1}^k \{\stq^{j_{\ell}} \in \oSt \mid \stq \in \St\} \cup \{s_0, s_e\}$.
The set $\Delta_\mathcal{N}$ of transitions is defined as follows: for every states $\stq_1, \stq_2, \str \in S$ and $\ell \in [k-1]$, we add a transition $(\stq_1^{j_{\ell}}, T(\str^{i_{\ell}}), \stq_2^{j_{\ell+1}}) \in \Delta_\mathcal{N}$  if there exists a run $\rho$ of $\ocT^*$ over $t$ such that $\rho(p_{\ell}) = \stq_1^{j_{\ell}}$, $\rho(p_{\ell+1}) = \stq_2^{j_{\ell}+1}$, and $\rho(u_{\ell}) = \str^{i_{\ell}}$.
Moreover, assuming that $\&$ is a fresh symbol, we add transition $(u_0, \{\&\}, \stq^{j_{1}})$ to $\Delta_\mathcal{N}$ if there exists a run $\rho$ of $\ocT^*$ over $t$ such that $\rho(p_{1}) = \stq^{j_{1}}$, and $\rho(\lambda) = \st^{i}$. Finally, we add transition $(\stq^{j_{k}}, T(\str^{i_{k}}), u_e)$ to $\Delta_\mathcal{N}$.

Note that all transitions in the succinct NFA are directed from level $j_{\ell-1}$ to level $j_{\ell}$ with $j_{\ell-1} > j_{\ell}$, for some $\ell \in \{1,2,\dots,k-1\}$, which implies that $\mathcal{N}$ is unrolled.
Here, the level $j_\ell$ is defined as the set of states $\{q^{j_\ell} \in \overline{S} \; | \; q \in S	\}$. Furthermore, 
note that transitions are labeled by sets $T(\str^{i_{\ell}})$ where $i_{\ell} < i$, which are represented by tree automaton $\ocT[\str^{i_{\ell}}]$ for $i_{\ell} < i$. Thus, the conditions required by Definition~\ref{def:prop} are satisfied since for each transition label $T(\str^{i_{\ell}})$, it holds that $|T(\str^{i_{\ell}})|$ is at most exponential in the size of $\ocT[\str^{i_{\ell}}]$, and by Algorithm~\ref{alg:sample}, we have already precomputed values such that we can check membership, approximate its size, and obtain an almost uniform sample from $T(\str^{i_{\ell}})$. 
Finally, the existence of the run $\rho$ for the definition of each transition in $\Delta_\mathcal{N}$ can be checked in polynomial time in the size of $t$~\cite{tata2007} and, thus, $\mathcal{N}$ can be constructed from $t$ and $\cT$ in polynomial time. 

It's only left to show that $|\cL_k(\mathcal{N})| = |T(\st^i, t)|$. For this, note that every word accepted by $\mathcal{N}$ is of length $k+1$ and of the form $\& t_1 t_2 \cdots t_k$. Then consider the function that maps words $\& t_1 t_2 \ldots t_k$ to the tree $t[u_1 \rightarrow t_1] \cdots [u_k \rightarrow t_k]$. 
One can show that each such a tree is in $T(\st^i, t)$, and then the function goes from $\cL_k(\mathcal{N})$ to $T(\st^i, t)$.
Furthermore, the function is a bijection. Clearly, if we take two different words, we will produce different trees in $T(\st^i, t)$, and then the function is injective. 
To show that the function is surjective, from a tree $t' \in T(\st^i, t)$ and a run $\rho$ of $\ocT$ over $t'$, we can build the word $\& t_1 t_2 \ldots t_k$ where each $t_i$ is the subtree hanging from the node $u_i$ in $t'$. Also, this word is realized by the following sequence of transitions in $\mathcal{N}$:
$$
(s_0, \{\&\}, \rho(p_1)), (\rho(p_1),\, T(\rho(u_1)), \rho(p_2)),\, \ldots,\,  (\rho(p_{k-1}), T(\rho(u_{k-1})), \rho(p_k)),\, (\rho(p_k), T(\rho(u_k)), s_e).
$$
Thus, the function is surjective. Hence, from the existence of a bijection from $\cL_k(\mathcal{N})$ to $T(\st^i, t)$, we conclude that  $|\cL_k(\mathcal{N})| = |T(\st^i, t)|$.  

Recall that the size of succinct NFA $\mathcal{N}$ is defined as $|\mathcal{N}| = |S_\mathcal{N}|+|\Delta_\mathcal{N}|+\sum_{(s,A,s') \in \Delta_\mathcal{N}} \rsize{A}$. 
Thus, given that $|S_\mathcal{N}| = im$, each set label $A = T(\rho(u_i))$ is represented by the tree automaton $\ocT[\rho(u_i)]$ and the size of $\ocT[\rho(u_i)]$ is bounded by $(im)^2$, we conclude that $|\Delta_\mathcal{N}| \leq (im)^3$ and $\sum_{(s,A,s') \in \Delta_\mathcal{N}} \rsize{A} \leq (im)^4$. Putting everything together, we conclude that $|\mathcal{N}| \leq 3(im)^4$, which was to be shown.

To finish with the proof, we need to consider the sequence $u_1, \ldots, u_k$ of holes of $t$, and assume that two holes of $t$ share the same parent, so that $\pi = p_1, p_2, \ldots, p_k$ is the main path of $t$, with $p_i = \parent{u_i}$ for each $i \in [k-1]$ and $p_k = u_k$. The proof for this case can be done in a completely analogous way.
\end{proof}

\noindent
In the following theorem, we show how to estimate $|\cL_k(\mathcal{N})|$ for a given 
unrolled succinct NFA $\mathcal{N}$ and integer $k\geq 1$ given in unary (recall the definition of unrolled succinct NFA from the beginning of this section).

\begin{theorem}\label{thm:progmain}
Let $\mathcal{N}$ be an unrolled succinct NFA, $k\geq 1$
$\epsilon \in (100 |\mathcal{N}|^4 \eps_0,1) $, where $\eps_0$ is as in Definition \ref{def:prop}. Moreover, fix $\delta \in (0,1/2)$ and assume that $N$ satisfies that $|\cL_k(\mathcal{N})| \leq N$. Then there exists an algorithm that with probability at least $1-\delta$ outputs a value $\aN$ such that $\aN = (1 \pm \eps)|\cL_k(\mathcal{N})|$. The algorithm runs in time 
\begin{eqnarray*}
O\bigg(T \cdot \frac{\log(N/\eps)\log^2(1/\delta) |\mathcal{N}|^{18}}{\eps^4}\bigg),
\end{eqnarray*}
where $T$ is as in Definition \ref{def:prop}, and makes at most 
\begin{eqnarray*}
O\bigg(\frac{\log^2(1/\delta) \log(N/\eps) |\mathcal{N}|^{18}} { \eps^4}\bigg)
\end{eqnarray*}
queries to the sampling oracle. 	
	Furthermore, there is an almost uniform sampler which returns elements of $\cL_k(\mathcal{N})$ such that 
	\begin{eqnarray*}
	\pr{\text{outputs  }\pi}  & = & (1 \pm \eps)\frac{1}{| \cL_k(\mathcal{N})|}
	\end{eqnarray*}
	for every $\pi \in \cL_k(\mathcal{N})$, and has the same runtime and oracle complexity as above.
\end{theorem}
As a corollary based on the reduction described earlier, we obtain the following. 
\begin{corollary}\label{cor:progmain}
	Let $\mathcal{N}$ be a succinct NFA, $k \geq 1$, $u,v \in V$, $\epsilon \in (100 ((k+1)|G|)^4 \eps_0,1) $, where $\eps_0$ is as in Definition \ref{def:prop}. Moreover, fix $\delta \in (0,1/2)$ and assume that $N$ satisfies that 
	$|\cL_k(\mathcal{N})| \leq N$. Then there exists an algorithm that, with probability at least $1-\delta$ outputs a value $\aN$ such that $\aN = (1 \pm \eps)|\cL_k(\mathcal{N})| \leq N$. The algorithm runs in time 
	\begin{eqnarray*}
		O\bigg(T \cdot \frac{\log(N/\eps)\log^2(1/\delta) (k|\mathcal{N}|)^{18}}{\eps^4}\bigg),
	\end{eqnarray*}
	where $T$ is as in Definition \ref{def:prop}, and makes at most 
	\begin{eqnarray*}
		O\bigg(\frac{\log^2(1/\delta) \log(N/\eps) (k|\mathcal{N}|)^{18}} { \eps^4}\bigg)
	\end{eqnarray*}
	queries to the sampling oracle. 	
	Furthermore, there is an almost uniform sampler which returns elements of $\cL_k(\mathcal{N})$
	such that 
	\begin{eqnarray*}
		\pr{\text{outputs  }\pi}  & = & (1 \pm \eps)\frac{1}{|\cL_k(\mathcal{N})|}
	\end{eqnarray*}
	for every $\pi \in \cL_k(\mathcal{N})$, and has the same runtime and oracle complexity as above.
\end{corollary}
\begin{proof}
	The result follows from Theorem \ref{thm:progmain}, as well as the reduction described earlier from arbitrary succinct NFAs to unrolled succinct NFAs.
	Notice that in this reduction the size of $\mathcal{N}$ increases by a factor of $O(k)$, which completes the proof.
\end{proof}

Using  Theorem \ref{thm:progmain}, we can prove Lemma \ref{lem:estimatepart}.
\begin{proof}[Proof of Lemma \ref{lem:estimatepart}]
Recall that we assume given a tree automaton $\cT = (S, \Sigma , \Delta,\sinit)$ over binary trees, a natural number $n \geq 1$ given in unary, the relative error $\eps \in (0,1)$ and a value $\delta \in (0,1/2)$. Moreover, we assume that $m \geq 3$ is the size of $\cT$, which we define as $m = \|\Delta\|$.

Let $s \in S$ and $t$ be a partial tree such that $\virtsize{t} = i$. By using Lemma \ref{lem:treetograph}, we can construct in polynomial-time a succinct NFA $\mathcal{N}$ such that $N(s^i,t) = |T(s^i,t)| = |\cL_k(\mathcal{N})|$. 
Moreover, by the reduction of Lemma \ref{lem:treetograph}, each label set $A$ of a transition in $\mathcal{N}$ is of the form $A = T(s^j)$ for some state $s$ and some $j < i$ in the graph.  
Thus, if we assume $\eps_0 = \eps 4nm$, then this gives us $\aN(s^j) = (1 \pm \eps_0)N(s^j)$ as $\mathcal{N}$ is required to satisfy the properties of Definition \ref{def:prop}. Moreover, 
define $\eps_1  = \eps_0 (4n m)^{16}$, where $\eps_1$ is the precision parameter from Theorem \ref{thm:progmain}. 
Then we have that $\eps_1 >  300 (n m)^{16}\eps_0 \geq 100 |\mathcal{N}|^4 \eps_0$ as required by Theorem \ref{thm:progmain}, where here we used the fact that $|\mathcal{N}| \leq 3(nm)^4$ in the reduction of Lemma \ref{lem:treetograph}. Moreover, $\eps_1 = (4n m)^{17} \eps < 1$,  as also required by Theorem \ref{thm:progmain}, since we assume that $\eps < 1/(4nm)^{18}$. Finally, we also set $\delta_0 = \delta^{(nm)^3}$ to be the failure probability as in Theorem \ref{thm:progmain}. Thus, by Theorem \ref{thm:progmain},	
	 using at most $O(\log^2(1/\delta_0) \log(N/\eps_1) (nm)^{4\cdot 18} /\eps_1^4) = O(\log^2(1/\delta) \log(N/\eps) (nm)^{10} /\eps^4)$	 
	  samples, we obtain a $(1 \pm \eps_1)$-estimate of the size of the number of labeled paths with probability $1-\delta_0 = 1-\delta^{(nm)^3}$.
	   By Lemma \ref{lem:treetograph}, we therefore obtain the same estimate of the partition $t$, for a given partial tree $t$.	  
	   
	   We now bound the number of ordered, rooted, labeled trees of size $n$. 
	  By Cayley's formula, we can bound the number of unlabeled, unordered, undirected trees by $n^{n-2}$. The number of rooted, unordered, undirected, unlabeled trees can then be bounded by $n^{n-1}$. For each tree, every vertex has $|\Sigma| \leq m$ choices of a labeling, thus there are $m^n  n^{n-1} < (nm)^{nm}$ labeled unordered, undirected rooted trees (recall that $m \geq 3$). Finally, for each such a tree, we can bound the number of ways to transform it into an ordered and directed tree by $(2n)!$, which gives a bound of $(2n)! \cdot (nm)^{nm} \leq (nm)^{(nm)^2}$ (recall again that $m \geq 3$). Note that this also implies that $N \leq  (nm)^{(nm)^2}$, 
	  which gives a total sample complexity bound of $O(\log^2(1/\delta) (nm)^{13} /\eps^5)$. 
	  Observe that the number of partial trees $t$ such that $\virtsize{t} = i$ is bounded by $n (nm)^{(nm)^2}$. In particular, this bound is obtained by considering that $m \geq 3$ and the fact that the number of labels for partial trees is at most $|\Sigma| + n \leq m + n$.
	  Now by a union bound and considering that $\delta < 1/2$, with probability
	  \begin{eqnarray*}
	  1- (nm)(nm)^{(nm)^2} \delta^{(nm)^3} \ \geq \ 1 - \delta^{(nm)^3 - ((nm)^2+1) \log(nm)} \ \geq \ 1 - \delta^{nm},
	  \end{eqnarray*}
	  \EstimatePartitionSize$(t,\,\st^i,\,\{\aT_i(s^j)\}_{s \in S, j < i}, \allowbreak \{\aN(s^j)\}_{s \in S, j \leq i},\, \eps,\, \delta)$ returns a $(1 \pm \eps_1) = (1 \pm (4nm)^{17}\eps)$ estimate for all trees $t$ such that $\virtsize{t} = i$ and for all states $s^i$ such that $s \in S$.  
	  Finally, note that the runtime is $\poly(n,m,1/\eps,\log(1/\delta))$ since 
	  it is bounded by a polynomial in the sample complexity, which is polynomial in $n$, $m$, $1/\eps$ and $\log(1/\delta)$
	  by Lemma \ref{lem:beta} and Theorem \ref{thm:progmain}.
\end{proof}

\subsection{Approximate Counting of Accepted Words in Succinct NFAs}
The goal of this section is to prove Theorem \ref{thm:progmain}. 
In what follows, fix a succinct NFA $\mathcal{N} =(S,\Gamma,\Delta,\sinit,\sfin)$ over a finite set of labels $\Gamma$ and recall that the label sets of $\Delta$ have to satisfy the conditions of Definition \ref{def:prop}.
Besides, assume that $\mathcal{N}$ is unrolled, and recall the definition of unrolled succinct NFA from the beginning of this section.
Without loss of generality, assume that $S$ only contains states which lie on a path from $\sinit$ to $\sfin$. Furthermore, let $s_0, \ldots, s_n$ be a topological order of the states in $S$ such that $s_0 = \sinit$ and $s_n = \sfin$, where $|S|=n+1$. In other words, every path from $s_0$ to $s_n$ can be written in the form $s_0,s_{i_1},s_{i_2},\dots,s_n$, where $1 \leq i_1 \leq i_2 \leq \dots \leq n$. Finally, for brevity, we write $r = |\mathcal{N}|$.

For the sake of presentation, for every state $s_i$, let $W(s_i) = \cL(\mathcal{N}_{s_i})$ and $N(s_i) = |\cL(\mathcal{N}_{s_i})|$, where $\mathcal{N}_{s_i}$ is an exact copy of $\mathcal{N}$ only with the final state changed to $s_i$. Then our goal is to estimate $N(s_n)$.
Similar than for the previous section, we will simultaneously compute estimates $\aN(s_i)$ of the set sizes $N(s_i)$, as well as multi-set sketches $\aW(s_i)$ filled with i.i.d. nearly-uniform samples from $W(s_i)$. We do this iteratively for $i=0,1,2,\dots,n$.  For the remainder of the section, fix $\eps,\eps_0,\delta$ as in Theorem \ref{thm:progmain}, and assume that $\mathcal{N}$ satisfies the conditions of Definition \ref{def:prop}, from which we know that $|W(s_i)| \leq N$ for each node $s_i$, for some $N \leq 2^{\poly(r)}$. 
 Set $\gamma = \log(1/\delta)$. 
Finally, since the FPRAS must run in time $\poly(r,1/\eps)$, we can assume $\eps < 1/(300r)$ without loss of generality.

We first observe that membership in $W(s_i)$ is polynomial-time testable given polynomial-time membership tests for each label $A$.

\begin{proposition}\label{prop:membertest}
	Suppose that given a sequence $a_1 \ldots a_t \in \Gamma^t$, we can test in time $T$ whether $a_i \in A$ for each transition label $A$  (T is the membership time in Definition \ref{def:prop}). Then given any state $s_j$, we can test whether $a_1 \dots a_t \in W(s_j)$ in time $O(|\Delta| T)$ 
\end{proposition}
\begin{proof}
	We can first remove all transitions not contained in a run of length exactly $t$ from $s$ to $s_j$ in time $O(|\Delta|)$ by a BFS.
	Then for each transition $e = (s',A,s'')$ remaining which is on the $i$-th step from $s$ to $s_j$, with $i \leq t$, we keep $e$ if and only if $a_i \in A$. Note that $i$ is unique for $e$
	as $\mathcal{N}$ is unrolled. 
	It is now straightforward to check that $s_j$ is reachable from $s$ with the remaining transitions if and only if $a_1\ldots a_t \in W(s_j)$. It is easy to check that the time needed by the entire procedure is~$O(|\Delta| T)$.
\end{proof}

\noindent
Now, analogous to the prior section, we define the following properties for each state $s_i$. Recall that we use $r = |\mathcal{N}|$
to denote the size of $\mathcal{N}$ for brevity.


\paragraph{Property 3:} We say that $s_i$ satisfies Property 3 if $\aN(s_i) = (1 \pm i  \eps/r) N(s_i)$. 

\paragraph{Property 4:}  We say that $s_i$ satisfies Property 4 if for every subset $L \subseteq \{0,1,2,\dots,i-1\}$ , we have that
\begin{eqnarray*}
	\bigg|\frac{\big|\aW(s_i) \setminus \big(\bigcup_{j \in L} W(u_j)\big)\big|}{|\aW(s_i)|} - \frac{\big|W(s_i) \setminus \big(\bigcup_{j \in L} W(u_j)\big)\big|}{|W(s_i)|}\bigg|  & \leq & \frac{\eps}{r}
\end{eqnarray*}
Moreover, we have that the subsets $\aW(s_i)$ are of size $|\aW(s_i)| = O(\frac{r^3 \gamma}{\eps^2})$.
\paragraph{Property 5:}  We say that $s_i$ satisfies Property 5 if we have a polynomial-time algorithm which returns independent samples from $W(s_i)$, such that for all $w \in W(s_i)$:
\begin{eqnarray*}
	\pr{\, \text{outputs $w$} \mid \neg \textbf{FAIL} \, } &=& \bigg(1 \pm \frac{\eps}{3r^2}\bigg) \frac{1}{N(s_i)}
\end{eqnarray*} 
The algorithm is allowed to fail with probability at most $1/4$, in which case it returns \textbf{FAIL} (and returns no element). Finally, each run of the algorithm is allowed to use at most $O(\frac{\log(N/\eps)\gamma r^{11}}{\eps^2})$ oracle calls to the sampling oracle of Definition \ref{def:prop}.

\begin{lemma}\label{lem:alpha}
	Suppose that Properties  $3$, $4$ and $5$ hold for all $s_j$ with $j < i$. Then with probability at least $1-2^{-\gamma r}$ we can return an estimate $\aN(s_i) = (1 \pm i  \eps/r) N(s_i)$. In other words, under these assumptions it follows that Property $3$ holds for $s_i$. Moreover, the total number of calls to the sampling oracle of Definition \ref{def:prop} is $O(\frac{\log(N/\eps)\gamma^2 r^{17}}{\eps^4})$, and the total runtime can be bounded by $O(T\frac{\log(N/\eps)\gamma^2 r^{17}}{\eps^4})$, where $T$ is the membership test time in Definition \ref{def:prop}.
\end{lemma} 

\begin{proof}	
	First note that for $i=0$, Property 3 trivially holds since $W(s_0) = \emptyset$. Otherwise, let $i \geq 1$ and $(v_1, A_1, s_i),\ldots, (v_k, A_k, s_i) \in \Delta$ be the set of all transitions going into $s_i$ (recall that we sorted the $\{s_i\}_{i\in [0,n]}$ by a topological ordering, so $v_1,\dots,v_k \in \{s_0,\dots,s_{i-1}\}$). Observe $W(s_i) = \bigcup_{j=1}^k \left(W(v_j) \cdot A_j\right)$.
	Fix a transition $(v_j, A_j, s_i)$ and assume that $\eps_0 = \eps/(100r^4)$ in Definition \ref{def:prop}, so that we are given estimates $\aN(A_j) = (1 \pm \epsilon_0)|A_j|$ since $\mathcal{N}$ satisfies the conditions in this definition. 
	Then the number of words reaching $s_i$ through $(v_j, A_j, s_i)$ is given by $N(v_j)|A_j|$, and it can be estimated as follows assuming that $v_j = s_k$ with $k < i$:
	\begin{eqnarray*}
		\aN(v_j)\cdot  \aN(A_j) & = & (1 \pm k \eps/r)(1 \pm \eps_0)N(v_j) |A_j|\\
		& = & (1 \pm (i-1) \eps/ r)(1 \pm \eps/(100r^4))N(v_j) |A_j|\\
		& = & (1 \pm (i-1+1/r^{2})(\eps/r))N(v_j) |A_j|.
	\end{eqnarray*}
	Notice that in this deduction we use the fact that $i - 1 \leq r$. 
	Let $p_j$ denote be the probability that a uniformly drawn $s \sim \left(W(v_j) \cdot A_j\right)$ is not contained in $ \bigcup_{j' < j} \left(W(v_{j'}) \cdot A_{j'}\right)$. Then we can write $W(s_i) =   \sum_{j=1}^k N(v_j) |A_j|  p_j$.	We now estimate $p_j$ via $\tilde{p}_j$. By Property 5 and the assumptions from Definition \ref{def:prop}, we can obtain nearly uniform samples $w \sim W(v_j)$ and  $a \in A_j$ in polynomial time, 
	such that the probability of sampling a given $w$ and $a$ are $(1 \pm \eps/(3r^2)) N(v_j)^{-1}$	
	and $(1 \pm \eps_0)|A_j|^{-1}$, respectively.
	Moreover, $w \cdot a$ is a sample from $W(v_j) \cdot A_j$, such that for any $w' \cdot a' \in W(v_j) \cdot A_j$:
	\begin{equation}\label{eq:prob1}
	\pr{w \cdot a = w' \cdot a' } = \left(1 \pm \frac{\eps}{3r^2}\right) \frac{1 \pm \eps_0}{N(v_j)|A_j|} 
	\end{equation}
	Note that the relative error $(1 \pm \eps/(3r^2))(1 \pm \eps_0)$ can be bounded in the range $(1 \pm 2\eps/(5r^2))$ using the fact that $\eps_0 = \eps/(100r^4)$. 
	We repeat this sampling process $d = O(\gamma r^{5}/ \eps^2)$ times, obtaining samples $w_1a_1,\ldots,w_da_d \sim W(v_j) \cdot A_j$ and set $\tilde{p}_j$ to be the fraction of these samples not contained in $ \bigcup_{j' < j} \left(W(v_{j'}) \cdot A_{j'}\right)$. Then by Hoeffding's inequality, with probability  $1-2^{-\gamma r}$ we have that $\tilde{p}_j = p_j \pm 2\eps/(3r^2)$ (here we use \eqref{eq:prob1}, which tell us that the expectation of $\tilde{p}_j$ is at most $2\eps/(5r^2)$ far from the correct expectation $p_j$). We then set:
	\begin{eqnarray*}
		\aN(s_i) &=& \sum_{j=1}^k \aN(v_j)\cdot  \aN(A_j) \cdot  \tilde{p}_j \\
		& = & (1 \pm (i-1+1/r^{2})(\eps/r))\sum_{j=1}^k N(v_j) \cdot |A_j| \cdot  \left(p_j \pm \frac{2\eps}{3r^2}\right) \\
		& = & (1 \pm (i-1+1/r^{2})(\eps/r))\left[ \sum_{j=1}^k N(v_j)  |A_j|  p_j \pm \frac{2\eps}{3r^2}   \sum_{j=1}^k N(v_j)  |A_j| \right] \\
		& = &(1 \pm (i-1+1/r^{2})(\eps/r))\left[ N(s_i) \pm \frac{2\eps}{3r^2}  \sum_{j=1}^k N(s_i) \right] \\
		& = & (1 \pm (i-1+ 1/r^{2})(\eps/r))\left[ N(s_i) \pm 2\eps  /(3r) N(s_i) \right] \\
		& = &(1 \pm i \eps/ r) N(s_i)
	\end{eqnarray*}
	as desired. Note that we need only compute $\tilde{p}_j$ for at most $r$ values of $v_j$, thus the total number of samples required is $O(\gamma r^{6}/ \eps^2)$. By Property $5$, each sample required $O(\frac{\log(N/\eps)\gamma r^{11}}{\eps^2})$ oracle calls, thus the total oracle complexity is $O(\frac{\log(N/\eps)\gamma^2 r^{17}}{\eps^4})$ as needed. By Proposition \ref{prop:membertest}, each membership test required while computing the probabilities $\tilde{p}_j$ required $O(Tr)$ time, thus the total runtime can be bounded by $O(T\frac{\log(N/\eps)\gamma^2 r^{17}}{\eps^4})$, which was to be shown. 
	
	
	\end{proof}
	\begin{algorithm}[!h]
		\SetKwInOut{Input}{Input}
		\caption{\SampleFromVertex$(s_i, \aN(s_i))$\label{fig:sampPath}}
		\smallskip
		\If{$\aN(s_i) = 0$}{
		\smallskip
		\Return{$\bot$ \hfill \tcp{$\bot$ indicates that $W(s_i)$ is empty} \label{alg:empty_bot}}
		}\smallskip
		Initialize $w \leftarrow \lambda,\ q \leftarrow 1$.  \hfill \tcp{$\lambda$  is the empty string} 
		\smallskip
		\For{$\beta = 1,2,\dots,d(s_0,s_i)$}{	 \label{line:outerfor}
				\smallskip
			\While{$|w| < \beta$}{ \label{Line:forstart}
					\smallskip
				Let $\FF(s_i,w) = \{(x_1, A_1,y_1),\ldots,(x_k, A_k, y_k)\}$. \\
					\smallskip
				Let $Z_j = \aN(x_j) \aN(A_j)$ for each $j \in [k]$, and $Z = \sum_{j=1}^k Z_j$.\\
					\smallskip
				Order the $Z_i$'s so that $Z_1 \geq Z_2 \geq \cdots \geq Z_k$. \\
					\smallskip
				Sample $j \sim [k]$ with probability $\frac{Z_j}{Z}$. \\
					\smallskip
			Obtain an almost uniform sample $a \sim A_j$.  \hfill	 \tcp{via Definition \ref{def:prop} }\label{line:sample}
				\smallskip
				Let $\BB(a) = \{j'  \in [k] \mid a \in A_{j'}\}$ and accept $a$ with probability:
				\begin{eqnarray*}
					q_{a,j} &=& \frac{\big|\aW(x_j) \setminus \big(\bigcup_{j' \in \BB(a) \,:\, j' < j} W(x_{j'}) \big) \big|}{\big|\aW(x_j)\big|}
				\end{eqnarray*}\label{line:acceptprob}
				
				\If{$a$ is accepted}{\smallskip
					$\rho \leftarrow 0$, $M \leftarrow \Theta(\frac{\log(N/\eps)\gamma r^{10}}{\eps^2}$). \\
				        \tcp{$\rho$ approximates the probability that a trial fails to accept some $a$} 
					\smallskip
					\For{$h=1,2,\dots,M$}{		\smallskip 	\label{line:estimatefail}		
						Sample $j \sim [k]$ with probability $\frac{Z_j}{Z}$, then sample $a_h \sim A_j$. \\	\smallskip
						With probability $1-q_{a_h,j}$ increment $\rho \leftarrow \rho+1$.\\	\smallskip
					}	\smallskip
					$\rho \leftarrow \frac{\rho}{M}$, and update: 
					\begin{eqnarray*}
						q &\leftarrow &  q \cdot \left( \frac{\sum_{j' \in B(a)} \frac{\aN(x_{j'})}{Z} q_{a,j'}}{1-\rho}\right)
					\end{eqnarray*}
					
					$w \leftarrow a w$ \label{alg:append_a_w}\\	\smallskip
				}	\smallskip
			}\label{Line:forend}	\smallskip
		
		}	\tcp{$q$ approximates the probability that $w$ was sampled up to this point} 
		
		With probability $\frac{1}{2q \aN(s_i)}$ \Return{$w$}, otherwise \Return{\FAIL}.  \label{Line:final}
	\end{algorithm}
We now describe our sampling procedure.  To do so, we will first develop some notation.
We extend our previous notation and use $\mathcal{N}_{s,s'}$ to denote an exact copy of $\mathcal{N}$ but with $s$ as the initial state and $s'$ as the final state. Let $s_i$ be a vertex and $w \in \Gamma^*$ a sequence of symbols. If $w$ contains at least one symbol, then let $\FF(s_i, w) = \{(s_{a}, A, s_{b}) \in \Delta \mid  w \in \cL_k(\mathcal{N}_{s_{b}, s_i})$, namely, $\FF(s_i, w)$ is the set of all transitions $(s_{a}, A, s_{b})$ incident to a state $s_b$ from which we can reach $s_i$ by a path labeled 
by $w$. Otherwise, we have that $w = \lambda$, where $\lambda$ is the empty string, and $\FF(s_i, \lambda)$ is defined as $\{(s_{a}, A, s_{b}) \in \Delta \mid  s_{b} = s_i\}$.
Moreover, let $|w|$ be the length of $w$
and $d: S \times S \to \mathbb{N}$ be the distance metric between states of $\mathcal{N}$ when considered as a graph, i.e. $d(s, s')$ is the number of transitions that we need to make to get from $s$ to $s'$. Without loss of generality (by unrolling the succinct NFA if needed), we can make sure that $d$ is well defined.

We now present our main sampling algorithm of this section: Algorithm \ref{fig:sampPath}. For ease of presentation, Algorithm \ref{fig:sampPath} is written as a Las Vegas randomized algorithm, which could potentially have unbounded runtime. However, by simply terminating the execution of the algorithm after a fixed polynomial runtime and outputting an arbitrary string of bits, the desired correctness properties of the sampler will hold. The analysis of Algorithm \ref{fig:sampPath}, along with the finite-time termination procedure, is carried out in the proof of Lemma \ref{lem:beta} below.

\begin{lemma}\label{lem:beta}
	Suppose Property 3 holds for all levels $j \leq i$, and Property 4 holds for all levels $j < i$. 
	If $W(s_i) = \emptyset$, then \SampleFromVertex$(s_i, \aN(s_i))$ return $\bot$ with probability 1. Otherwise, conditioned on not outputting \FAIL, \SampleFromVertex$(s_i, \aN(s_i))$ returns $w \sim W(s_i)$ from a distribution $\DD$ over $W(s_i)$ such that 
	\begin{eqnarray*}
	\DD(w) & = & \bigg(1 \pm \frac{\eps}{3r^2}\bigg) \frac{1}{|W(s_i)|}
	\end{eqnarray*}
	for all $w \in W(s_i)$. Moreover, the algorithm uses at most $O(\frac{\log(N/\eps)\gamma r^{11}}{\eps^2})$ calls to the uniform sampling oracle of Definition \ref{def:prop}, runs in time $O(T\frac{\log(N/\eps)\gamma^2 r^{15}}{\eps^4})$, and outputs \FAIL\ with probability at most $3/4$. 
\end{lemma}


\begin{proof}
Assume first $W(s_i) = \emptyset$, so that $N(s_i) = 0$. Then given that $\aN(s_i) = (1 \pm \eps)N(s_i)$ by Property 3, we conclude that $\aN(s_i) = 0$ and the algorithm returns $\bot$ in line \ref{alg:empty_bot}. Notice that if $W(s_i) \neq \emptyset$, then $N(s_i) > 0$ and, therefore, $\aN(s_i) > 0$ by Property 3. Thus, if $W(s_i) \neq \emptyset$, then the algorithm does not return $\bot$.

Assume that $W(s_i) \neq \emptyset$, and notice that this implies $r \geq 4$.
	Consider an element $w$ sampled so far at any intermediate state of the execution of \SampleFromVertex$(s_i, \aN(s_i))$.
	Let $W(s_i,w) = \{t \in W(s_i) \mid t = w' \cdot w\}$. In other words, $W(s_i,w) \subseteq W(s_i)$ is the subset of words with suffix equal to $w$. We now want to sample the next symbol $a \in \Gamma$ conditioned on having sampled the suffix $w$ of a path so far. In other words, we want to sample $a$ with probability proportional to the number of words in $W(s_i,w)$ which have the suffix $aw$, meaning we want to choose $a$ with probability:
	\[	\frac{|\eval(s_i,aw)|}{|\eval(s_i,w)|}.	\]
	However, we do not know these sizes exactly, so we must approximately sample from this distribution. Let us consider the probability that our algorithm samples $a$ on this step (given $w$).
	For the algorithm to sample $a$, it must first choose to sample a transition $(x_j,A_j,y_j)$ from the set $\FF(s_i,w) = \{(x_1,A_1,y_1),(x_2,A_2,y_2),\ldots,(x_k,A_k,y_k)\}$
	such that $a \in A_i$, which occurs with probability $Z_j/Z$ with $Z = \sum_{j'=1}^k Z_{j'}$ and $Z_{j'} = \aN (x_{j'}) \aN(A_{j'})$. Then, on the call to the oracle on line \ref{line:sample}, it must obtain $a \sim A_j$ as the almost uniform sample, which occurs with probability $(1 \pm \eps_0)|A_j|^{-1}$ by Definition \ref{def:prop}. Finally, it must choose to keep $a$ on line \ref{line:acceptprob}, which occurs with probability $\frac{\big|\aW(x_j) \setminus \big(\bigcup_{j' \in \BB(a) \,:\, j' < j} W(x_{j'}) \big) \big|}{\big|\aW(x_j)\big|}$, where  $\BB(a) = \{j' \in [k]  \mid a  \in A_{j'} \}$. Thus, altogether, the probability that we choose $a \in \Gamma$ on this step is 
	\begin{align}
	&\nonumber \sum_{j \in \BB(a)} \frac{\aN (x_j) \aN(A_j) }{Z} \cdot \frac{1 \pm \eps_0}{|A_j|} \cdot \frac{\big|\aW(x_j) \setminus \big(\bigcup_{j' \in \BB(a) \,:\, j' < j} W(x_{j'}) \big) \big|}{\big|\aW(x_j)\big|} \\
& \hspace{40pt} =\ (1 \pm 3 \eps_0)\sum_{j \in \BB(a)} \frac{\aN (x_j)  }{Z}  \cdot \frac{\big|\aW(x_j) \setminus \big(\bigcup_{j' \in \BB(a) \,:\, j' < j} W(x_{j'}) \big) \big|}{\big|\aW(x_j)\big|}\label{eq:5} \\	
&\nonumber\hspace{40pt} =\ (1 \pm 3 \eps_0)\sum_{j \in \BB(a)} \frac{(1 \pm \eps)|\eval(x_j)|}{Z}\cdot \frac{\big|\aW(x_j) \setminus \big(\bigcup_{j' \in \BB(a) \,:\, j' < j} W(x_{j'}) \big) \big|}{\big|\aW(x_j)\big|} \\	
&\nonumber\hspace{40pt} =\ (1 \pm 2 \eps)\sum_{j \in \BB(a)} \frac{|\eval(x_j)|}{Z}\cdot \frac{\big|\aW(x_j) \setminus \big(\bigcup_{j' \in \BB(a) \,:\, j' < j} W(x_{j'}) \big) \big|}{\big|\aW(x_j)\big|} \\	
&\nonumber\hspace{40pt} =\ (1 \pm 2 \eps)\sum_{j \in \BB(a)} \frac{|\eval(x_j)|}{Z} \bigg[ \frac{\big| W(x_j) \setminus \big(\bigcup_{j' \in \BB(a) \,:\, j' < j} W(x_{j'}) \big) \big|}{\big| W(x_j)\big|} \pm \frac{\eps}{r}\bigg] \\
&\nonumber\hspace{40pt} =\ (1 \pm 2 \eps)\frac{1}{Z}\sum_{j \in \BB(a)}  \bigg[ \bigg| W(x_j) \setminus \big(\bigcup_{j' \in \BB(a) \,:\, j' < j} W(x_{j'}) \big) \bigg| \pm |\eval(x_j)| \frac{\eps}{r}\bigg] \\
&\nonumber \hspace{40pt} =\ (1 \pm 2 \eps)\frac{1}{Z}\bigg(\sum_{j \in \BB(a)} \bigg| W(x_j) \setminus \big(\bigcup_{j' \in \BB(a) \,:\, j' < j} W(x_{j'}) \big) \bigg| \pm\sum_{j \in \BB(a)} |\eval(x_j)| \frac{\eps}{r}\bigg) \\
&\hspace{40pt} =\ (1 \pm 2 \eps)\frac{1}{Z}\left(|\eval(s_i,aw)|  \pm  |\eval(s_i,aw)| \eps\right) \label{eqn11}\\
&\nonumber\hspace{40pt} =\ (1 \pm 4 \eps)\frac{1}{Z}|\eval(s_i,aw)| 
	\end{align}
	Where equation \eqref{eq:5} uses the fact that $\aN(A_j) = (1 \pm \eps_0)|A_j|$, and equation \eqref{eqn11} uses the fact that $\eval(s_i,aw) = (\bigcup_{j \in \BB(a)} W(x_j)) \cdot \{aw\}$.
The above demonstrates that on a \textbf{single} trial of the inner while loop in lines \ref{Line:forstart} to \ref{Line:forend}, conditioned on having chosen the sample $w$ so far, for each $a\in \eval(s_i,w)$ we choose $a$ with probability $(1 \pm 4 \eps) \frac{ |\eval(s_i,aw)| }{Z}$. However, we do not break out of the while loop on line \ref{Line:forend} and move to the next step in the outer for loop in line \ref{line:outerfor} until we have chosen an $a\in \eval(s_i,w)$ to append to $w$. If on a given trial of the loop in line \ref{Line:forstart}, the algorithm does not choose some element to append to $w$, we say that it outputs no sample. Call the event that we output some sample $\mathcal{E}_i$, and let $\mathcal{E}_i(a)$ denote the event that we specifically output $a \in \Gamma$.  Then
	\begin{eqnarray*}
	\pr{ \mathcal{E}_i} &=& \sum_{a \in \eval(s_i,w)}	\pr{ \mathcal{E}_i(a)}  \\
 &=& (1 \pm 4 \eps)\sum_{a \in \eval(s_i,w)} \frac{ |\eval(s_i,aw)| }{Z}\\
	& =& (1 \pm 4 \eps)\frac{ |\eval(s_i,w)| }{Z}.
	\end{eqnarray*}
	Therefore,
		\begin{eqnarray*}
	\pr{\mathcal{E}_i(a) \mid \mathcal{E}_i} &=& \frac{\pr{\mathcal{E}_i(a)} }{\pr{ \mathcal{E}_i} }  \\
	 &=&\frac{ (1 \pm 4 \eps)}{ (1 \pm 4 \eps)}\bigg(\frac{Z}{|\eval(s_i,w)| }\bigg)  \frac{ |\eval(s_i,aw)| }{Z} \\
	 	 &=&(1 \pm 10 \eps)\frac{ |\eval(s_i,aw)|}{|\eval(s_i,w)|}.
	\end{eqnarray*}
	Thus, conditioned on outputting a sample at this step, we choose $a \in \Gamma$ with probability 
	\begin{eqnarray}\label{eq:out-ep}
	(1 \pm 10 \eps) \frac{ |\eval(s_i,aw)|}{|\eval(s_i,w)|}
	\end{eqnarray}
	Observe the above is within $(1 \pm 10\eps)$ of the correct sampling probability.
	
	\paragraph{Estimating the probability that we sample a given $w \in \eval(s_i)$.} We now analyize the quantity $q$ in the algorithm, and argue that at the point where line \ref{Line:final} is executed, $q$ is a good approximation of the probability that our algorithm sample $w$ at this point. 
	Now let $\rho^*_\beta$ be the probability that, within step $\beta \in \{1,2,\dots,d(s_0,s_i)\}$ of the outer for loop on line \ref{line:outerfor}, a given run of the inner while loop between lines \ref{Line:forstart} to \ref{Line:forend} fails to append a new sample $a$ to $w$. Let $\rho_\beta$ be the value that \textit{we} assign to the variable $\rho$ at the end of the for loop in line \ref{line:estimatefail} (note that this loop is executed at most once within step $\beta$ of the outer loop \ref{line:outerfor}).  The variable $\rho_\beta$ will be our estimate of $\rho^*_\beta$.
	
	Note that each trial of the inner while loop is independent, so $\rho^*_\beta$ only depends on the $\beta$ from the outer loop, and the value of $w$ sampled so far. Let $\DD_\beta'(aw)$ be the exactly probability that entry $a$ is chosen on step $\beta$ of the outer loop of our algorithm, conditioned on having chosen $w$ so far. Being in step $\beta$ of the outer loop then implies that $|aw| = \beta $. Now fix any $w = w_1 w_2 \dots w_{d(s_0,s_i)} \in \eval(s_i)$.
	Let $\DD'(w)$ be the exact probability that $w$ is sampled at this point right before the execution of line~\ref{Line:final}.  By definition we have 
	 \[	 \DD'(w) =  \DD_0'(w_{d(s_0,s_i) } ) \cdot \left( \prod_{j=1}^{d(s_0,s_i) - 1} 
	 \DD'_j\left(w_{d(s_0,s_i) - j} \ldots w_{d(s_0,s_i)}\right) \right) \] 
	 so via \eqref{eq:out-ep} we obtain:
	 \begin{eqnarray}\label{eq:ddp-range}
	  \DD'(w) \ =\ \frac{1}{|\eval(s_i)|}\prod_{j=1}^{d(s_0,s_i) }(
	 1 \pm 10 \eps)  \ = \ \frac{(1 \pm 10\eps)^{r}}{|\eval(s_i)|}  \ =  \ \frac{(1 \pm 20r \eps)}{|\eval(s_i)|} 
	 \end{eqnarray}
	 
	 \begin{claim}\label{claim1}
	 	 If $q$ is the value the variable $q$ takes at the point where line \ref{Line:final} is executed, given that $w = w_1 \dots w_{d(s_0,s_i)}$ is the value of $e$ at this point, then 
	 	 \[\DD'(w) \  = \ \left(1 \pm \frac{\eps}{50 r^2}\right) q\]
	 	 with probability at least $1-r(\eps/N)^{2\gamma}$.
	 \end{claim}
	 \begin{proof}
	 To see this, consider step $\beta \in \{1,2,\dots,d(s_0,s_i)\}$ of the for outer loop in line \ref{line:outerfor}. We first claim that $\rho_\beta^* \leq 1- \frac{1}{r}$. To see this, note that the probability that $Z_1$ is chosen is at least $\frac{1}{k} \geq \frac{1}{r}$, since we ordered $Z_1 \geq Z_2 \geq \dots \geq Z_k$, and if $Z_1$ is chosen the sample $a \sim A_1$ is never rejected, which completes the claim. Now each iteration of the for loop in line \ref{line:estimatefail} defines a random variable $Z$ which indicates if a random trial of the inner loop in line \ref{Line:forstart} would result in a failure. Here, if $Z=1$ (a trials fails), then we increment $\rho = \rho+1$, otherwise we do not. Thus $\ex{Z} = \rho^*_\beta$, and by Hoeffding's inequality, after repeating $M=\Theta(\frac{\log(N/\eps)\gamma r^{10}}{\eps^2})$ times, it follows that with probability 	 $(1 - (\eps/N)^{2\gamma})$ that we have $\rho_\beta = \rho_\beta^* \pm \eps/(400r^5)$ and, therefore, $1-\rho_\beta = (1 \pm \eps/(400r^4))(1- \rho_\beta^*)$ since $1/r \leq 1 - \rho_\beta^*$. Thus, it holds that
	\begin{eqnarray}\label{eq:app_beta}
	\frac{1}{1- \rho_\beta^*} & = & \bigg(1 \pm \frac{\eps}{200r^4}\bigg)\frac{1}{1-\rho_\beta}
	\end{eqnarray}
	 Let $\tau = d(s_0,s_i) - \beta +1$, so that on step $\beta$ of the for outer loop in line \ref{line:outerfor} we are considering the probability that we sample a $w_\tau \in \Gamma$ given that we have already sampled $w_{-\tau} = w_{\tau+1} \dots w_{d(s_0,s_i)}$. Now as shown above, the probability that $w_\tau$ is accepted on one trial of the while loop is precisely: 
	\[	 q^*(w_\tau)  = \sum_{j \in \BB(w_\tau)} \frac{\aN(x_j) \aN(A_j)  }{Z} \cdot \frac{\epsilon'}{|A_j|} \cdot \frac{\big|\aW(x_j) \setminus \big(\bigcup_{j' \in \BB(w_{\tau}) \,:\, j' < j} W(x_{j'}) \big) \big|}{\big|\aW(x_j)\big|}.\]
Notice that we are not trying to bound $q^*(w_\tau)$ in this expression, we are computing the exact value of $q^*(w_\tau)$, but based on an unknown value $\eps'$. 
However, we know by Definition \ref{def:prop} that $1 - \eps_0 \leq \epsilon' \leq 1 + \eps_0$.
Thus, 
although we do not know the exact value of $q^*(w_\tau)$, we do know that $1-3\eps_0 \leq \aN(A_j) \cdot \epsilon' / |A_j| \leq 1 + 3\eps_0$ by the assumptions of Definition \ref{def:prop}. Thus, we can estimate $q^*(w_\tau)$ by 
\[	\hat{q}(w_\tau)= \sum_{j \in B(w_\tau)} \frac{\aN(x_j) }{Z}\frac{\big|\aW(x_j) \setminus \big(\bigcup_{j' \in \BB(w_{\tau}) \,:\, j' < j} W(x_{j'}) \big) \big|}{\big|\aW(x_j)\big|}\] 
so that $q^*(w_\tau)=
(1 \pm 3\eps_0)\hat{q}(w_\tau)$.
	 The probability that $w_\tau$ is accepted overall before moving to the next step of the loop is $\sum_{j=1}^\infty q^*(w_\tau)(\rho^*_\beta)^{j-1} = q^*(w_\tau)(\frac{1}{1-\rho^*_\beta})$, 
	 for which by equation \eqref{eq:app_beta} we have a $(1 \pm \eps/(200r^4))(1 \pm 3\eps_0) = (1 \pm \eps/(100r^3))$ estimate of via the value $\hat{q}(w_\tau)/(1-\rho_\beta)$ (recall that $r \geq 4$).
	 Note that this is precisely the value which we scale the variable $q$ by after an iteration of the inner loop that appends a new sample $a$ to $w$ in line \ref{alg:append_a_w} of the algorithm. 
	 It follows that at the end of the main loop, we have:
	 \begin{eqnarray*}
	\DD'(w)  \ =  \ \bigg(1 \pm \frac{\eps}{100r^3}\bigg)^{d(s_0,s_i)} \cdot q \ =  \ \bigg(1 \pm \frac{\eps}{50r^2}\bigg) \cdot q
	 \end{eqnarray*}
	 as needed.
	 Notice that this equality holds under the condition that for every $\beta = 1, \ldots, d(s_0,s_i)$, it holds that $\rho_\beta = \rho_\beta^* \pm \eps/(400r^5)$, 
	 which occurs with probability $1 - (\eps/N)^{2\gamma}$ for each $\beta$. By a union bound, we obtain the desired success probability of at least $1 - r (\eps/N)^{2\gamma}$. 
	 \end{proof}
	 Thus, by rejecting with probability $\frac{1}{2q \aN(s_i)}$, it follows from Claim  \ref{claim1}
	 that the true probability $\DD^\star(w)$ that we output a given $w\in \eval(s_i)$ is 
	 \begin{eqnarray}\label{eq:dstar}
	 \DD^\star(w) \ = \ \frac{\DD'(w)}{2q\aN(s_i)} \ =\  \bigg(1 \pm \frac{\eps}{50r^2}\bigg) \frac{1}{2\aN(s_i)}  
	 \end{eqnarray}
	 Note that for the above fact to be true, we need that $\frac{1}{2q} \leq \aN(s_i)$, else the above rejection probability could be larger than $1$. But again by Claim \ref{claim1}  we have that
	 \begin{eqnarray*}
	 \frac{1}{2q \aN(s_i)} &\leq & \bigg(1+ \frac{\eps}{50 r^2}\bigg)\frac{1}{2 \aN(s_i)\DD'(w)} \\
	 &\leq& (1+ 2\eps)\frac{1}{2 |W(s_i)|\DD'(w)} \\
	 	  &\leq& (1+ 122 r \eps)\frac{1}{2 } \\
	 	  & \leq& 3/4
	 \end{eqnarray*}
	 where the second to last inequality holds applying \eqref{eq:ddp-range}, and the last inequality holds give that $\eps < 1/(300r)$.	  
	 Therefore, the rejection probability is always a valid probability. Similarly:
  \begin{eqnarray*}
\frac{1}{2q \aN(s_i)} & \geq & \bigg(1-\frac{\eps}{50 r^2}\bigg)\frac{1}{2 \aN(s_i)\DD'(w)} \\
&\geq & (1- 2\eps)\frac{1}{2 |W(s_i)|\DD'(w)} \\
&\geq & (1- 42 r \eps)\frac{1}{2 } \\
& \geq & 1/4
\end{eqnarray*}
Thus, by the above, we can bound the probability that we output \textbf{FAIL} on this last step by $3/4$ as required. Now, we are ready to analyize the true output distribution $\mathcal{D}$ over $\eval(s_i)$, which is given by the distribution $D^\star$ conditioned on not outputting \textbf{FAIL}. Now for any  $w \in \eval(s_i)$, we can apply equation \eqref{eq:dstar} to compute $\mathcal{D}(w)$ via:
%
%
	 \begin{eqnarray*}
	 \DD(w) &=& 
	 \pr{\text{output } w \in \eval(s_i) \mid \neg \FAIL }\\ 
	 &=& \frac{\pr{\text{output } w \in \eval(s_i) \wedge \neg \FAIL }}{\pr{\text{not output } \FAIL}}\\
	 &=& \frac{\pr{\text{output } w \in \eval(s_i)}}{\pr{\neg \FAIL}}\\
	&=& 	\frac{\DD^\star(w)}{\sum_{w \in W(s_i)}D^\star(w) }\\
		&=& 	\frac{(1 \pm \eps/(50r^2))}{2 \aN(s_i)(\sum_{w \in W(s_i)}\frac{1 \mp \eps/(50r^2)}{2\aN(s_i)} ) }\\
				&=& 	\frac{(1 \pm \eps/(50r^2))}{ \sum_{w \in W(s_i)}(1 \mp \eps/(50r^2 ))}\\
					&=& 	\frac{(1 \pm \eps/(50r^2))}{ |W(s_i)|(1 \mp \eps/(50r^2 ))}\\
						&=& 	\bigg(1 \pm \frac{3\eps}{50r^2}\bigg)\frac{1}{ |W(s_i)|}\\
						&=& 	\bigg(1 \pm \frac{\eps}{10r^2}\bigg)\frac{1}{ |W(s_i)|}
	 \end{eqnarray*}
	 which is the desired result.
	 
%
	
	\paragraph{Oracle complexity and runtime}
	For the complexity of the sample procedure, note that each iteration to sample a $a \in \Gamma$ has failure probability at most $\frac{1}{r}$ independently, thus with probability $1-(\eps/(rN))^22^{-10r\gamma}$ it requires at most $10 r^3 \log(Nr/\eps) \gamma$ iterations. Thus with probability $1-(\eps/(N))^22^{-10r\gamma}$, the total number of iterations required to produce a single sample (or output \textbf{FAIL} at the end) is $10 r^4 \log(Nr/\eps) \gamma$. Note that each iteration that fails to accept an $a \in \Gamma$ produces one call to the unit oracle. Once an $a$ is accepted, we run an experiment $M$ times, which produces $M = O(\frac{\log(N/\eps)\gamma r^{10}}{\eps^2})$ oracle calls. Since this occurs at most $r$ times, the total number of oracle calls is $O(\frac{\log(N/\eps)\gamma r^{11}}{\eps^2})$. Note that the runtime is dominant by the cost of the $\rho$ estimation procedure, wherein the probability $q_{a_h,j}$ is computed at each step of line \ref{line:estimatefail}. Note that to compute $q_{a_h,j}$, we must test for each sample in $s \in \widetilde{W}(x_j)$ if $s$ is contained in the union of at most $r$ sets, which requires at most $Tr$ runtime by the assumptions of Definition \ref{def:prop}. Note that each set has size at most $O(\frac{\gamma r^3}{\eps^2})$. Thus the total runtime can be bounded by $O(T\frac{\log(N/\eps)\gamma^2 r^{15}}{\eps^4})$.

In summary, with probability $1-(\eps/(N))^22^{-10r\gamma}$, the total number of samples (unit oracle calls) required is $O(\frac{\log(N/\eps)\gamma r^{11}}{\eps^2})$ (and the runtime is as stated above). 	
	Now if the sample complexity becomes too large we can safely output anything we would like (specifically, we can output \textbf{FAIL}, or even an arbitrary sequence of bits). The probability that this occurs, or that any of our $O(r)$ estimate of the inner failure probabilities $\rho$ fails to be within our desired bounds, is at most  $(\eps/(N))^22^{-10r\gamma} + r(\eps/N)^22^{-\gamma} \leq (\eps/N) 2^{-\gamma}$. Call the event that the sample complexity becomes too large  $\mathcal{Q}$, and let $\mathcal{P}$ be the event that any of our $O(r)$ estimate of $\rho$ fail to be within our desired bounds. We have just proven that
	\[	\pr{\text{we output } w \in W(s_i) \; | \; \neg \mathcal{P}} = (1 \pm \eps/(10r^2) )\frac{1}{|W(s_i)|}.	\]
	Now since $\pr{\mathcal{P}\cup \mathcal{Q}} \leq (\eps/N) 2^{-\gamma}$, we have
	\begin{equation}
	\begin{split}
		\pr{\text{we output } w \in W(s_i) \; | \; \neg \mathcal{Q}}& =\pr{\text{ we output } w \in W(s_i) \; | \; \neg \mathcal{Q}, \neg \mathcal{P} } \pm (\eps/N)2^{-\gamma} \\
		&= \pr{\text{ we output } w \in W(s_i) \; | \; \neg \mathcal{P}} \pm 3(\eps/N) 2^{-\gamma},
	\end{split}
	\end{equation} so it follows that for each $w \in W(s_i)$, we have 
	\begin{equation}
	\begin{split}
		\pr{\text{we output } w \in W(s_i) \; | \; \neg \mathcal{Q} } &=  (1 \pm  \eps/(10r^2))\frac{1}{|\eval(s_i)|} \pm 3(\eps/N)2^{-\gamma} 	\\
		&=  (1 \pm  \eps/(10r^2))\frac{1}{u_i} \pm \eps\frac{3}{|\eval(s_i)|}2^{-\gamma} \\
			&=  (1 \pm  \eps/(3r^2) )\frac{1}{|\eval(s_i)|}, \\
	\end{split}
	\end{equation}
	which shows that our sampler is still correct even if we output random bits whenever $\mathcal{Q}$ fails to hold, which is the desired result taking $\gamma = \Omega(\log(r/\eps))$.

\end{proof}

We can use the above sampling regime to now show that having properties $3,4,5$ for $s_j$ with $j<i$ will imply them for $s_i$.

\begin{lemma}\label{lem:alphabeta}
	Fix any $\gamma>0$. Suppose Properties $3,4,5$ hold for all $s_j$ with $j < i$. Then with probability $1-2^{-10\gamma}$, properties $3,4$ and $5$ hold for $s_i$. Moreover, the total number of oracle calls is at most $O(\gamma^2 \log(N/\eps) r^{17} / \eps^4)$, and the total runtime is $O(T\frac{\log(N/\eps)\gamma^2 r^{17}}{\eps^4})$.
\end{lemma}
\begin{proof}
	We obtain property $3$ with probability $1-2^{-\gamma r}$ by Lemma \ref{lem:alpha}, which uses $O(\frac{\log(N/\eps)\gamma^2 r^{17}}{\eps^4})$ sampling oracle calls. 	
	   By Lemma \ref{lem:beta}, conditioned on property $4$ holding for all levels $j<i$ and property $3$ holding for all $j \leq i$, we now have a procedure which can sample each $w \sim W(s_i)$ with probability in the range $(1 \pm \eps/(3r^2))\frac{1}{|W(s_i)|}$, and such that the sampler satisfies the other conditions of Property $5$. Thus property $5$ for level $i$ now holds deterministically conditioned on property $3$ holding for $i$ and all $j<i$. 
	
Now for property $4$, we can take $s' = \Theta(\gamma r^3 /\eps^2)$ samples to build $\widetilde{W}(s_i)$. By Lemma \ref{lem:beta}, each run of the algorithm requires $O(\gamma r^{11} \log(N/\eps) /\eps^2)$ oracle calls, and fails to return a sample with probability at most $3/4$. Applying Hoeffding's inequality on the required number of trials of the sampling algorithm to obtain $s'$ independent samples, this requires $O(\gamma^2 r^{15} \log(N/\eps) /\eps^4)$ oracle calls with probability $1-2^{-100\gamma}$. Given this, we have that each sample in $\widetilde{W}(s_i)$ is a $(1 \pm \eps/(3r^2))$-relative error almost uniform sample. Applying Hoeffding's inequality again, it follows that for a fixed set $L \subset \{s_0,\dots,s_{i-1}\}$, we have	
	\[	\left|	\frac{|\widetilde{W}(s_i) \setminus \left( \cup_{s_j \in L} W(s_j) \right)|}{|\widetilde{W}(s_i)|}	- \frac{|W(s_i) \setminus \left( \cup_{s_j \in L} W(s_j) \right)|}{|W(s_i)|}	\right| \leq \frac{\eps}{3r^2} + \frac{\eps}{2r } \leq \frac{\eps}{r}	\]
	with probability $1-2^{-100 \gamma r}$, and since there are only at most $2^{r}$ such subsets $L$, by a union bound this holds for all such subsets with probability $1-2^{-100 \gamma r - r}$. 
	Thus the overall probability of success is $1-2^{-100 \gamma r -r} - 2^{-\gamma r} > 1- 2^{-10 \gamma}$. Note that the runtime is dominated by the time required to obtain Property $3$ via Lemma \ref{lem:alpha}, which is $O(T\frac{\log(N/\eps)\gamma^2 r^{17}}{\eps^4})$.
\end{proof}
We are now ready the prove the main theorem. 

\begin{proof}[Proof of Theorem \ref{thm:progmain}]
	By Lemma \ref{lem:alphabeta}, conditioned on having Properties $3,4,$ and $5$ for a level $i$, we get it for $i+1$ with probability $1-2^{-10\gamma}$ with at most $O(\gamma^2 \log(N/\eps) r^{17} / \eps^4)$ oracle calls. It follows inductively that with probability $1-r2^{-10\gamma}$, we have Property $3$ and $5$ for all levels with at most $O(\gamma^2 \log(N/\eps) r^{18} / \eps^4)$ oracle calls, which completes the proof after recalling that $\gamma:= \log(1/\delta)$. The runtime for each level is  $O(T\frac{\log(N/\eps)\gamma^2 r^{17}}{\eps^4})$ by Lemma \ref{lem:alphabeta}, thus the total runtime is $O(T\frac{\log(N/\eps)\gamma^2 r^{18}}{\eps^4})$. 
\end{proof}

%% file: applications-new.tex

\subsection{Constraint satisfaction problems}
\label{sec-csp}
Constraint satisfaction problems offer a general and natural setting to represent a large number of problems where solutions must satisfy some constraints, and which can be found in different areas such as artificial intelligence, satisfiability, programming languages, temporal reasoning, scheduling, graph theory, and databases~\cite{V00,DBLP:books/daglib/0004131,rossi2006handbook,DBLP:books/daglib/0013017,DBLP:series/faia/2009-185,russell2016artificial}. Formally, a constraint satisfaction problem (CSP) is a triple $\cP = (V,D,C)$ such that $V = \{x_1, \ldots, x_m\}$ is a set of variables, $D$ is a set of values and $C = \{C_1, \ldots, C_n\}$ is a set of constraints, where each constraint $C_i$ is a pair $(\bar t_i,R_i)$ such that 
$\bar t_i$ is a tuple of variables from $V$ of arity $k$, for some $k \geq 1$, and $R_i \subseteq D^k$. Moreover, an assignment $\nu : V \to D$ is said to be a solution for $\cP$ if for every $i \in [n]$, it holds that $\nu(\bar t_i) \in R_i$~\cite{russell2016artificial}, where $\nu(\bar t_i)$ is obtained by replacing each variable $x_j$ occurring in $\bar t_i$ by $\nu(x_j)$. The set of solutions for CSP $\cP$ is denoted by $\sol(\cP)$. 

The two most basic tasks associated to a CSP are the evaluation and the satisfiability problems. In the evaluation problem, we are given a CSP $\cP$ and an assignment $\nu$, and the question to answer is whether $\nu \in \sol(\cP)$. In the satisfiability problem, we are given a CSP $\cP$, and the question to answer is whether $\sol(\cP) \neq \emptyset$. Clearly, these two problems have very different complexities, as in the former we only need to verify the simple condition that $\nu(t) \in R$ for every constraint $(t,R)$ in $\cP$, while in the latter we need to search in the space of all possible assignments for one that satisfies all the constraints. In fact, these two problems also look  different in terms of our interest in the specific values for the variables of the CSP; in the former we are interested in the value of each one of them that is given in the assignment $\nu$, while in the latter the variables of $\cP$ are considered as existential quantifiers, as we are interested in knowing whether there exists a solution for $\cP$ even if we do not know how to construct it.
As a way to unify these two problems, and to indicate for which variables we are interested in their values, a projection operator has been used in the definition of CSPs~\cite{CJ06,W10}. Notice that the definition of this operator has also played an important role when classifying the complexity of CSPs in terms of algebraic properties of relations~\cite{CJ06}. Formally, an {\em existential} CSP (ECSP) is defined as a pair $\cE = (U, \cP)$, where $\cP = (V,D,C)$ is a CSP and $U \subseteq V$. 
Moreover, the set of solution for $\cE$ is defined as
\begin{eqnarray*}
\sol(\cE) & = & \{ \nu|_{U} \mid \nu \in \sol(\cP)\},
\end{eqnarray*}
where $\nu|_U$ is the restriction of function $\nu$ to the domain $U$. Notice that both the evaluation and the satisfiability problems for a CSP $\cP$ can be reduced to the evaluation problem for an ECSP.
In fact, the satisfiability problem for $\cP$ corresponds to the problem of verifying whether the assignment with empty domain belongs to $\sol(\cE)$, where $\cE$  is the ECSP $(\emptyset, \cP)$.
Moreover, the evaluation and satisfiability problems are polynomially interreducible for ECSPs, so ECSPs provide a uniform framework for these two problems allowing us to focus only on the evaluation problem.

Clearly the satisfiability problem for CSPs, as well as the evaluation problem for ECSPs, is $\np$-complete; in particular, $\np$-hardness is a consequence that the satisfiability of 3-CNF propositional formulae can be easily encoded as a constraint satisfaction problem. Thus, a large body of research has been devoted to understanding the complexity of the evaluation problem for ECSPs, and finding tractable cases. In particular, two prominent approaches in this investigation have been based on the idea of viewing an ECSP as a homomorphism problem where the target structure is fixed~\cite{FV98,B17,Z17} or on the use of decomposition methods that require of some acyclicity conditions on an ECSP to be satisfied~\cite{GLS00,GLS02}. In this section, we focus on the latter class of methods, and show how the main results of this article can be used to deal with the fundamental problem of counting the number of solutions to an~ECSP. 

The evaluation problem for ECSPs is equivalent to the evaluation problem for CQs~\cite{KV00}. 
To see why this is the case, take an existential CSP $\cE = (U, \cP)$, where $\cP = (V,D,C)$ is a CSP with $U = \{y_1, \ldots, y_m\} \subseteq V$ and $C = \{C_1, \ldots, C_n\}$  where each constraint $C_i$ is a pair $(\bar t_i,R_i)$ such that 
$\bar t_i$ is a tuple of variables from $V$ of arity $k$, for some $k \geq 1$, and $R_i \subseteq D^k$. For each $R_i$,  let $\bar{R}_i$ be a $k$-ary relational symbol. Then define the CQ: 
\begin{eqnarray} 
	Q_{\cE}(\bar{y})  & \leftarrow & \bar{R}_1(\bar t_i), \ldots, \bar{R}_n(\bar t_n),
\end{eqnarray} 
with $\bar{y} = (y_1, \ldots, y_m)$ and the database $D_{\cE}$ such that $\bar{R}_i(\bar{a}) \in D_{\cE}$ if, and only if, $\bar{a} \in R_i$ for each $i \leq n$. Then it is easy to see that for every assignment $\nu$ it holds that:
\begin{eqnarray*}
	\nu \in \sol(\cE) & \text{ if and only if } & \nu(\bar{y}) \in Q_{\cE}(D_{\cE})
\end{eqnarray*}
This tight connection can be used to extend the notions of acyclicity  given in Section~\ref{sec:cqreduction} to the case of ECSPs. More precisely, $\cE$ is said to be acyclic if and only if $Q_{\cE}$ is acyclic~\cite{Y81,GLS98}, and $\hw(\cE)$ is defined as $\hw(Q_{\cE})$~\cite{GLS02}. 

The notion of acyclic CSP coincides with the notion of $\alpha$-acyclicity for hypergraphs~\cite{F83,BFMY83}, and it has played an important role in finding tractable cases for ECSPs~\cite{GLS00}. In fact, if $\aecsp = \{ (\cE,\nu) \mid \cE$ is an acyclic ECSP and $\nu \in \sol(\cE)\}$, then it holds that \aecsp\ is \logcfl-complete under many-to-one logspace reductions~\cite{GLS98}. Recall that \logcfl\ consists of all decision problems that are logspace reducible to a context-free language, and it holds that $\nlog \subseteq \logcfl \subseteq \acone$. Thus, we have that all problems in \logcfl\ can be solved in polynomial time and are highly parallelizable. 

Concerning to our investigation, we are interested in the fundamental problem of counting the number of solution of a ECSP. In general, this problem is $\sharpp$-complete and cannot admit an FPRAS (unless $\np = \rp$, given that the evaluation problem for CSP is $\np$-complete). Thus, we focus on the following fundamental problems where the degree of cyclicity of ECSP is bounded. 
\begin{center}
	\begin{tabular}{c}
		\framebox{
			\begin{tabular}{lp{10.3cm}}
				\textbf{Problem:} &  $\saecsp$\\
				\textbf{Input:} & An acyclic ECSP $\cE$\\
				\textbf{Output:} & $|\sol(\cE)|$
			\end{tabular}
		}\\
		\\
		\framebox{
			\begin{tabular}{lp{10.3cm}}
				\textbf{Problem:} &  $\skhwecsp$\\
				\textbf{Input:} & An ECSP $\cE$ such that $\hw(\cE) \leq k$\\
				\textbf{Output:} & $|\sol(\cE)|$
			\end{tabular}
		}
	\end{tabular}	
\end{center}
From the characterization of the evaluation problem for ECSPs in terms of the evaluation problem of CQs and Theorem~\ref{thm:acq-skhw-2}, we conclude that these problems admit FPRAS.
\begin{theorem}\label{thm:ecsp-skhw}
	\saecsp\ and $\skhwecsp$ admit an FPRAS for every $k \geq 1$.
\end{theorem}

\subsection{Nested words}

Nested words have been proposed as a model for the formal verification of correctness of structured programs that can contain nested calls to subroutines~\cite{AEM04,AM04,AM09}. In particular, the execution of a program is viewed as a linear sequence of states, but where a matching relation is used to specify the correspondence between each point during the execution at which a procedure is called with the point when we return from that procedure call. 

Formally, a binary relation $\mu$ on an interval $[n]$ is a matching if the following conditions hold:
(a) if $\mu(i,j)$ holds then $i < j$; 
(b) if $\mu(i,j)$ and $\mu(i,j')$ hold
then $j = j'$, and if $\mu(i,j)$ and $\mu(i',j)$ hold then $i = i'$; 
(c) if $\mu(i,j)$ and $\mu(i',j')$ hold, where $i \neq i'$ and $j \neq j'$, then either $[i,j] \cap [i',j'] = \emptyset$ or $[i,j] \subseteq [i',j']$ or $[i',j'] \subseteq [i,j]$. Moreover, given a finite alphabet $\Sigma$,  a {\em nested word} of length $n$ over $\Sigma$ is a tuple $\w=(w,\mu)$, where
$w \in \Sigma^*$ is a string of length $n$, and $\mu$ is a matching on $[n]$.

A position $i$ in a nested word 
$\w$ is a call (resp., return) position if there exists $j$ such that $\mu(i,j)$ (resp., $\mu(j,i)$) holds.
If $i$ is neither a call nor a return position in $\w$, then $i$ is said to be an internal position in $\w$. 
Figure~\ref{fig:shape} shows a nested word (without the
labeling with alphabet symbols). Solid lines are used to draw the linear edges that define a standard word, while nesting edges are drawn using dashed
lines. Thus, the relation $\mu$ is $\{(2,4), (5,6), (1,7)\}$, the set of call positions is $\{1,2,5\}$, the set of return positions is $\{4,6,7\}$ and the set of internal positions is~$\{3,8\}$.

Properties to be formally verified are specified by using nested word automata. Such automata have the same expressiveness as monadic second order logic over
nested words~\cite{AM09}, so they are expressive enough to allow the specification and automatic verification of a large variety of properties over programs with nested calls to subroutines. 
Formally, a (nondeterministic) {\em nested word automaton} (NWA) $\cN$ 
is a tuple $(S, \Sigma, S_0, F, P, \cd, \id, \rd )$ consisting of
a finite set of states $S$,
an alphabet $\Sigma$, 
a set of initial states $S_0 \subseteq S$,
a set of final states   $F \subseteq S$,
a finite set of hierarchical symbols $P$,
a call-transition relation $\cd \subseteq S \times \Sigma \times S \times P $,
an  internal-transition relation $\id \subseteq S \times \Sigma \times S $, 
and a return-transition relation $\rd \subseteq S \times P \times \Sigma \times S  $.

An NWA $\cN = (S, \Sigma, S_0, F, P, \cd, \id, \rd )$ works as follows with input a nested word $\w$. $\cN$ starts in an initial state in $S_0$ and reads $\w$ from left to
right. The state is propagated along the linear edges of $\w$ as in case of
a standard word automaton.
However, at a call position in $\w$, the nested word automaton propagates a state
along the linear edge together with a hierarchical symbol along the nesting edge of $\w$. 
At a return position in $\w$, the new state is determined based on the
state propagated along the linear edge as well as the symbol along the incoming
nesting edge.
Formally, a run $\rho$ of the automaton $\cN$ over a nested word
$\w=(a_1 \cdots a_n, \mu)$ is a sequence $s_0, s_1, \ldots, s_n$ of
states along the linear edges,
and a sequence $p_i$, for every call position $i$, of hierarchical symbols 
along the nesting edges,
such that: (a) $s_0\in S_0$; (b) for each call position $i$, it holds that $(s_{i-1},a_i,s_i,p_i)\in\cd$; (c) for each internal position $i$, it holds that $(s_{i-1},a_i,s_i)\in\id$; and (d) for each return position $i$ such that $\mu(j, i)$ holds, we have that $(s_{i-1}, p_j, a_i,s_i)\in\rd$. Moreover, the run $\rho$ is accepting if $s_n \in F$, and 
\begin{eqnarray*}
&\cL(\cN) \ = \ \{ \w \mid \w \text{ is a nested word over } \Sigma^* \text{ and there exists an accepting run of } \cN \text{ with input }\w\}.&
\end{eqnarray*}
The emptiness problem for nested word automata ask whether, given a NWA $\cN$, there exists a nested word $\w$ accepted by $\cN$. This is a fundamental problem when looking for faulty executions of a program with nested calls to subroutines; if $\cN$ is used to encode the complement of a property we expect to be satisfied by a program, then a nested word $\w \in \cL(\cN)$ encodes a bug of this program. In this sense, the following is also a very relevant problem for understanding how faulty a program~is:
\begin{center}
	\framebox{
		\begin{tabular}{ll}
			\textbf{Problem:} &  $\snwa$\\
			\textbf{Input:} & A nested word automaton $\cN$ and a string $0^n$\\
			\textbf{Output:} & $|\{ \w \in \cL(\cN) \mid |\w| = n\}|$
		\end{tabular}
	}
\end{center}
As there exists a trivial polynomial-time parsimonious reduction from $\snfa$ to $\snwa$, we have that $\snwa$ is $\sharpp$-complete. Interestingly, from the existence of an FPRAS for $\stta$ (see Corollary \ref{cor:fpras-ta-bta}) and the results in \cite{AM09} showing how nested word automata can be represented by using tree automata over binary trees, it is possible to prove that:
\begin{theorem}
$\snwa$ admits an FPRAS.
\end{theorem}

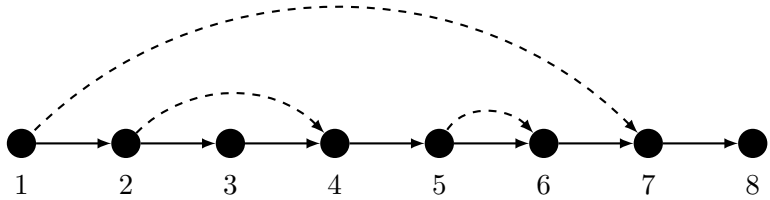
\begin{figure}[t]
\begin{center}
\begin{tikzpicture}
\node[circ, fill=black] (a1) {};
\node[circw, below=0mm of a1] (v1) {$1$};
\node[circ, fill=black,right=10mm of a1] (a2) {}
edge[arrin] (a1);
\node[circw, below=0mm of a2] (v2) {$2$};
\node[circ, fill=black,right=10mm of a2] (a3) {}
edge[arrin] (a2);
\node[circw, below=0mm of a3] (v3) {$3$};
\node[circ, fill=black,right=10mm of a3] (a4) {}
edge[arrin, dashed, bend right, in=-135, out =-45] (a2)
edge[arrin] (a3);
\node[circw, below=0mm of a4] (v4) {$4$};
\node[circ, fill=black,right=10mm of a4] (a5) {}
edge[arrin] (a4);
\node[circw, below=0mm of a5] (v5) {$5$};
\node[circ, fill=black,right=10mm of a5] (a6) {}
edge[arrin, dashed, bend right, in=-120, out =-45] (a5)
edge[arrin] (a5);
\node[circw, below=0mm of a6] (v6) {$6$};
\node[circ, fill=black,right=10mm of a6] (a7) {}
edge[arrin, dashed, bend right, in=-135, out =-45] (a1)
edge[arrin] (a6);
\node[circw, below=0mm of a7] (v7) {$7$};
\node[circ, fill=black,right=10mm of a7] (a8) {}
edge[arrin] (a7);
\node[circw, below=0mm of a8] (v8) {$8$};
\end{tikzpicture}
\end{center}
\caption{A nested word}
\label{fig:shape}
\end{figure}

\subsection{Knowledge compilation}

\input{DNNFcircuits}

%% file: DNNFcircuits.tex

Model counting is the problem of counting the number of satisfying assignments given a propositional formula. Although this problem is $\#\textsc{P}$-complete~\cite{valiant1979complexity}, there have been several approaches to tackle it~\cite{GomesSS09}.
One of them comes from the field of \emph{knowledge compilation}, a subarea in artificial intelligence~\cite{darwiche2002knowledge}.
Roughly speaking, this approach consists in dividing the reasoning process in two phases. 
The first phase is to compile the formula into a target language (e.g. Horn formulae, BDDs, circuits) that has good algorithmic properties. The second phase is to use the new representation to solve the problem efficiently.
The main goal then is to find a target language that is expressive enough to encode a rich set of propositional formulae and, at the same time, that allows for efficient algorithms to solve the counting problem.

One of the most used formalism in knowledge compilation are circuits in Negation Normal Form (NNF for short).
An NNF circuit $C = (V, E, g_0, \mu)$  is a directed acyclic graph $(V, E)$ where $V$ are called gates, edges $E$ are called wires, and $g_0 \in V$ is a distinguished gate called the output gate. The function $\mu$ assigns a type to each gate that can be $\wedge$ (AND), $\vee$ (OR), or a literal (i.e. a variable or the negation of a variable). We assume that all literals have in-degree $0$ and we call them input gates. 
Without loss of generality, we assume that all $\wedge$-gate and $\vee$-gate have in-degree two (if not, we can convert any NNF circuit in poly-time to binary gates). 
For a gate $g$ we define the set $\Vars(g)$ of all variables whose value can alter the value of $g$, formally, $v \in \Vars(g)$ if and only if there exists an input gate $g'$ with variable $v$ (i.e. $\mu(g) = v$ or $\mu(g) = \bar{v}$) and there is a path from $g'$ to $g$ in $(V, E)$. 
A valuation for $C$ is a mapping $\nu$ from the variables of $C$ to $\{0,1\}$. The valuation of $C$ with $\nu$, denoted by $\nu(C)$, is the value (i.e. $0$ or $1$) taken by $g_0$ when $C$ is evaluated in a bottom up fashion. 

A target language for knowledge compilation that has attracted a lot of attention is the class of DNNF circuits.
An NNF circuit $C$ is called decomposable~\cite{darwiche2001decomposable} if and only if for every $\wedge$-gate $g$ with incident gates $g_1, g_2$ it holds that $\Vars(g_1) \cap \Vars(g_2) = \emptyset$. In other words, if the incident gates of every $\wedge$-gate share no variables. For example, one can easily check that the NNF circuit of Figure~\ref{fig:DNNFcircuit} is decomposable. 
DNNF is the set of all NNF circuits that are decomposable. 
DNNF has good algorithmic properties in terms of satisfiability and logical operations. Furthermore, DNNF can be seen as a generalization of DNF formulae and, in particular, of binary decision diagrams (BDD), in the sense that every BDD can be transformed into a DNNF circuit in polynomial time. Nevertheless, DNNF is exponentially more succint than DNF or BDD, and then it is a more appealing language for knowledge compilation.

\begin{figure}
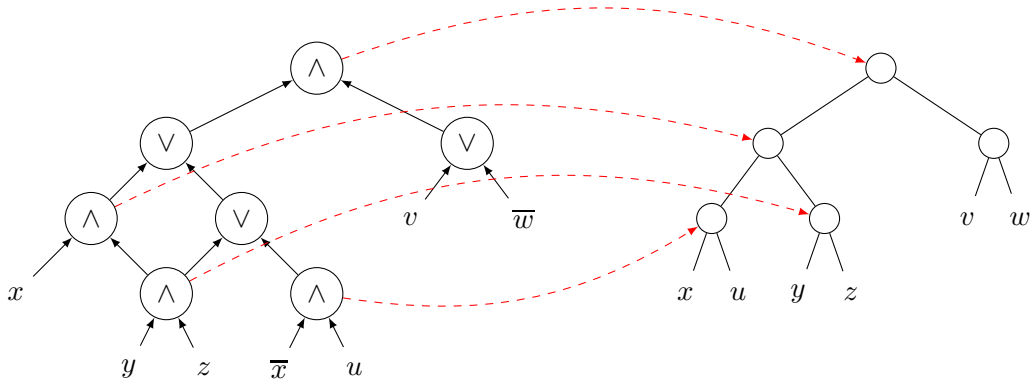

	\ctikzfig{dnnfcircuit}
	\caption{A structured DNNF circuit and its corresponding v-tree.}
	\label{fig:DNNFcircuit}
\end{figure}

Regarding model counting, DNNF circuits can easily encode \#P-complete problems (e.g. \#DNF) and, therefore, researchers have look into subclasses of DNNF with efficient counting properties. 
Deterministic DNNF (d-DNNF for short) is a subclass of DNNF where the counting problem can be solved in polynomial time (see~\cite{darwiche2001tractable} for a definition of d-DNNF). 
Indeed, several problems can be compiled into d-DNNF circuits, finding applications in probabilistic reasoning~\cite{chavira2006compiling}, query evaluation~\cite{Beame0RS17}, planning~\cite{bonet2006heuristics}, among others. 

However, as pointed out in~\cite{pipatsrisawat2008new} the compilation into d-DNNF circuits usually satisfies a structural property between variables, which naturally brings the class of structured DNNF circuits.  
A v-tree is a binary tree $t$ whose leaves are in one-to-one correspondence with a set of variables. Similar than for circuits, for a node $u$ in a v-tree, we denote by $\Vars(u)$ the set of all variables in the leaves of the subtree rooted at $u$. 
Then we say that a DNNF circuit $C$ respects a v-tree $t$ if for every
$\wedge$-gate $g$ and the two incident gates $g_1$ and $g_2$ of $g$, there exists a node $u$ in $t$ such that $\Vars(g_1) \subseteq \Vars(u_1)$ and $\Vars(g_2) \subseteq \Vars(u_2)$, where $u_1$ and $u_2$ are the left and right child of $u$ in $t$, respectively.
We say that a DNNF circuit $C$ is structured if and only if there exists a v-tree $t$ such that $C$ respects $t$. 
For example, in the right-hand side of Figure~\ref{fig:DNNFcircuit}, we show a v-tree for variables $\{x,y,z,u,v,w\}$. The red dashed lines show how $\wedge$-gates have to be assigned to the nodes in the v-tree in order for the circuit to respect this v-tree. 
Structured DNNF is the class of all DNNF circuits that are structured.
As it was already mentioned, the compilation into d-DNNF circuits usually produces circuits that are also structured~\cite{pipatsrisawat2008new}. Structured DNNF have been recently used for efficient enumeration~\cite{AmarilliBJM17,AmarilliBMN19} and in~\cite{oztok2014cv} it was shown that CNF formulae with bounded width (e.g. CV-width) can be efficiently compiled into structured DNNF circuits. 
Since structured DNNF circuits includes the class of DNF formulae, its underlying counting problem is $\#\textsc{P}$-complete. We now prove a positive approximation result. Specifically, consider the problem:
\begin{center}
	\framebox{
		\begin{tabular}{lp{11cm}}
			\textbf{Problem:} &  $\ssddnnf$\\
			\textbf{Input:} & A DNNF circuit $C$ and a v-tree $t$ such that $C$ respects $t$. \\
			\textbf{Output:} & $|\{\nu \mid \nu(C) = 1 \}|$
		\end{tabular}
	}
\end{center}
By using the existence of an FPRAS for $\sta$, we can show that  $\ssddnnf$ also admits an FPRAS. 

\begin{theorem}
	$\ssddnnf$ admits an FPRAS.
\end{theorem}
\begin{proof}
	The connection between structured DNNF and tree automata was already used in~\cite{AmarilliBJM17,AmarilliBMN19}, so this connection is not new. 
	Here, we show that there exists a parsimonious reduction from $\ssddnnf$ into $\sta$, which proves the FPRAS for structured DNNF.

	Let $t$ be a v-tree and $C = (V, E, g_0, \mu)$ be a DNNF circuit such that $C$ respects $t$. Let $\Vand$ be all gates in $V$ that are $\wedge$-gates or input gates.  Furthermore, let $f: \Vand \rightarrow t$ be the function that realizes that $C$ respects $t$, namely, for every $\wedge$-gate $g$ and gates $g_1$ and $g_2$ incident to $g$ it holds that $\Vars(g_1) \subseteq \Vars(f(g) \cdot 1)$ and $\Vars(g_2) \subseteq \Vars(f(g) \cdot 2)$. We also assume that if $g\in \Vand$ is a literal, then $f(g)$ is the leaf in $t$ that has the same variable as $g$ in $C$. One can easily check that for two gates $g_1, g_2 \in \Vand$ in $C$, if there is a path from $g_2$ to $g_1$ in $C$, then $f(g_2)$ is a descendant of $f(g_1)$ in~$t$.  
	Without loss of generality, we assume that $g_0$ is a $\wedge$-gate in $C$ and $f(g_0) = \epsilon$. 
	Finally, for any $\wedge$-gate $g$ we define the set $D(g)$ of all gates $g' \in \Vand$ such that there exists a path from $g'$ to $g$ in $(V, E)$ passing only through $\vee$-gates. Intuitively, if $g' \in D(g)$ then $g'$ is directly affecting the value of $g$ in the sense that if $g'$ is true, then at least one of the incident wires to $g$ is true. 
	
	The idea for the parsimonious reduction is to construct a tree automaton $\cT_C$ that will accept trees that encode valuations that makes $C$ true. To encode a  valuation, $\cT_C$ will only accept binary trees having exactly the same tree-shape as $t$, but its leaves are labeled with $0$ or $1$ (internal node will have any symbol, e.g. $@$). Given that leaves of $t$ are in one-to-one correspondence with the variables in $C$, then these encodings are in one-to-one correspondence with valuations in $C$. So, for a valuation $\nu$ let $t_\nu$ be the tree that has the same tree-shape as $t$ and whose leaves encode~$\nu$. Furthermore, suppose that $t_\nu$ is an input tree for the tree automaton $\cT_C$. 
	For checking that $\nu(C) = 1$, the states of $\cT_C$ will be either nodes $u \in t$ or pairs of the form $(u, g)$ where $u \in t$, $g \in \Vand$, and $f(g)$ is a descendant of $u$ in $t$. A node $u$ in a state (e.g. in the pair $(u,g)$) will take care of checking that the tree-shape of $t_\nu$ is the same as $t$. On the other hand, the gate $g$ in the pair $(u,g)$ will be use to navigate $C$ and find whether $g$ is evaluated to $1$ given the valuation encoded by $t_\nu$. When $f(g)$ is a strict descendant of $u$ (i.e. $f(g) \neq u$) we will continue down $t_\nu$ trying to find a node $u'$ such that $f(g) = u'$. When a node $u$ with $f(g) = u$ is found, then to evaluate $g$ to $1$ we need to find two gates $g_1, g_2 \in D(g)$ that are also evaluated to $1$ with $\nu$. Given that $C$ respects $t$ we know that $f(g_1)$ and $f(g_2)$ must be descendants of $f(g)$ in $t_\nu$, and then $\cT_C$ will recurse into the states $(u1, g_1)$ and $(u2, g_2)$ continuing into $g_1$ and $g_2$. If the non-deterministic decisions of $\cT_C$ are taken correctly, $\cT_C$ will reach the leaves $u$ of $t_\nu$ that has the same variable as the gate $g$ and it will check if the value $t_\nu(u)$ is correct with respect to $g$. 
	
	Let $\cT_C = (Q, \Sigma, \Delta, \sinit)$ be the tree automaton constructed from $C$ and $t$ such that $\Sigma = \{@, 0, 1\}$, $Q = t \, \cup \, \{(u, g) \in t \times \Vand \mid \text{$f(g)$ is a descendant of $u$ in $t$}\}$, and $\sinit = (\es, g_0)$ (recall that we assume that $g_0$ is a $\wedge$-gate and $f(g_0) = \es$). 
	We define the transition relation $\Delta$ by case analysis:
	\begin{itemize}
		\item For $u \in t$ and $u$ is not a leaf, then $(u, @, u1 \cdot u2) \in \Delta$.
		\item For $u \in t$ and $u$ is a leaf, then $(u, a, \es) \in \Delta$ for every $a \in \{0,1\}$.
		\item For $(u, g) \in Q$ such that $f(g) \neq u$, if $f(g)$ is a descendant of $u1$, then $((u,g), @, (u1,g) \cdot u2) \in \Delta$. Otherwise, if $f(g)$ is a descendant of $u2$, then $((u,g), @, u1 \cdot (u2,g)) \in \Delta$.
		
		\item For $(u, g) \in Q$ such that $f(g) = u$ and $u$ is not a leaf in $t$, then $((u,g), @, (u1, g_1)\cdot(u2,g_2)) \in \Delta$ for every $g_1, g_2 \in D(g)$ with $f(g_1)$ is a descendant of $u1$ and $f(g_2)$ is a descendant of $u2$. 
		
		\item For $(u, g) \in Q$ such that $f(g) = u$ and $u$ is a leaf in $t$, then $((u,g), 1, \es) \in \Delta$ iff $\mu(g) = t(u)$, that is, if $g$ is a positive literal and its variable coincide with the variable assign to $u$ in $t$.  Similarly, $((u,g), 0, \es) \in \Delta$ iff $\mu(g) = \neg t(u)$.
	\end{itemize}
	Finally, given a circuit $C$ the reduction produces the tree automaton $\cT_C$ and the value $0^{|t|}$. From the construction, it is straightforward to check that $\cT_C$ will accept trees that have the same tree-shape as $t$ and whose leaves encode a valuation of $C$. Furthermore, for a valuation $\nu$ and its tree $t_\nu$ one can check that $\nu(C) = 1$ if, and only if, $t_\nu \in \cL(\cT_C)$. Therefore, the reduction is parsimonious. Finally, the number of states and transitions of $\cT_C$ is polynomial in the size of $C$ and $t$, and thus the reduction can be computed in polynomial time. 
\end{proof}

%% file: main.bbl
\newcommand{\etalchar}[1]{$^{#1}$}
\begin{thebibliography}{BHvMW09}

\bibitem[ABJM17]{AmarilliBJM17}
Antoine Amarilli, Pierre Bourhis, Louis Jachiet, and Stefan Mengel.
\newblock A circuit-based approach to efficient enumeration.
\newblock In {\em 44th International Colloquium on Automata, Languages, and
  Programming, {ICALP} 2017, July 10-14, 2017, Warsaw, Poland}, pages
  111:1--111:15, 2017.

\bibitem[ABMN19]{AmarilliBMN19}
Antoine Amarilli, Pierre Bourhis, Stefan Mengel, and Matthias Niewerth.
\newblock Enumeration on trees with tractable combined complexity and efficient
  updates.
\newblock In {\em Proceedings of the 38th {ACM} {SIGMOD-SIGACT-SIGAI} Symposium
  on Principles of Database Systems, {PODS} 2019}, pages 89--103, 2019.

\bibitem[ACJR19]{Arenas19}
Marcelo Arenas, Luis~Alberto Croquevielle, Rajesh Jayaram, and Cristian
  Riveros.
\newblock Efficient logspace classes for enumeration, counting, and uniform
  generation.
\newblock In {\em Proceedings of the 38th ACM SIGMOD-SIGACT-SIGAI Symposium on
  Principles of Database Systems}, pages 59--73. ACM, 2019.

\bibitem[AD20]{AD20}
Serge Abiteboul and Gilles Dowek.
\newblock {\em The Age of Algorithms}.
\newblock Cambridge University Press, 2020.

\bibitem[AEM04]{AEM04}
Rajeev Alur, Kousha Etessami, and P.~Madhusudan.
\newblock A temporal logic of nested calls and returns.
\newblock In {\em Tools and Algorithms for the Construction and Analysis of
  Systems, 10th International Conference, {TACAS} 2004, Proceedings}, pages
  467--481, 2004.

\bibitem[{\`{A}}J93]{AJ93}
Carme {\`{A}}lvarez and Birgit Jenner.
\newblock A very hard log-space counting class.
\newblock {\em Theor. Comput. Sci.}, 107(1):3--30, 1993.

\bibitem[AM04]{AM04}
Rajeev Alur and P.~Madhusudan.
\newblock Visibly pushdown languages.
\newblock In {\em Proceedings of the 36th Annual {ACM} Symposium on Theory of
  Computing, Chicago, IL, USA, June 13-16, 2004}, pages 202--211, 2004.

\bibitem[AM09]{AM09}
Rajeev Alur and P.~Madhusudan.
\newblock Adding nesting structure to words.
\newblock {\em J. {ACM}}, 56(3):16:1--16:43, 2009.

\bibitem[BCM{\etalchar{+}}03]{DBLP:conf/dlog/2003handbook}
Franz Baader, Diego Calvanese, Deborah~L. McGuinness, Daniele Nardi, and
  Peter~F. Patel{-}Schneider, editors.
\newblock {\em The Description Logic Handbook: Theory, Implementation, and
  Applications}.
\newblock Cambridge University Press, 2003.

\bibitem[BFMY83]{BFMY83}
Catriel Beeri, Ronald Fagin, David Maier, and Mihalis Yannakakis.
\newblock On the desirability of acyclic database schemes.
\newblock {\em J. {ACM}}, 30(3):479--513, 1983.

\bibitem[BG06]{bonet2006heuristics}
Blai Bonet and Hector Geffner.
\newblock Heuristics for planning with penalties and rewards using compiled
  knowledge.
\newblock In {\em KR}, pages 452--462, 2006.

\bibitem[BGS00]{bertoni2000random}
Alberto Bertoni, Massimiliano Goldwurm, and Massimo Santini.
\newblock Random generation and approximate counting of ambiguously described
  combinatorial structures.
\newblock In {\em Annual Symposium on Theoretical Aspects of Computer Science},
  pages 567--580. Springer, 2000.

\bibitem[BHvMW09]{DBLP:series/faia/2009-185}
Armin Biere, Marijn Heule, Hans van Maaren, and Toby Walsh, editors.
\newblock {\em Handbook of Satisfiability}, volume 185 of {\em Frontiers in
  Artificial Intelligence and Applications}. {IOS} Press, 2009.

\bibitem[BLRS17]{Beame0RS17}
Paul Beame, Jerry Li, Sudeepa Roy, and Dan Suciu.
\newblock Exact model counting of query expressions: Limitations of
  propositional methods.
\newblock {\em {ACM} Trans. Database Syst.}, 42(1):1:1--1:46, 2017.

\bibitem[Bul17]{B17}
Andrei~A. Bulatov.
\newblock A dichotomy theorem for nonuniform csps.
\newblock In {\em 58th {IEEE} Annual Symposium on Foundations of Computer
  Science, {FOCS} 2017, Berkeley, CA, USA, October 15-17, 2017}, pages
  319--330, 2017.

\bibitem[Bv20]{10.1145/3389390}
Andrei~A. Bulatov and Stanislav \v{Z}ivn\'{y}.
\newblock Approximate counting {CSP} seen from the other side.
\newblock {\em ACM Trans. Comput. Theory}, 12(2), May 2020.

\bibitem[CDG{\etalchar{+}}07]{tata2007}
H.~Comon, M.~Dauchet, R.~Gilleron, C.~L\"oding, F.~Jacquemard, D.~Lugiez,
  S.~Tison, and M.~Tommasi.
\newblock Tree automata techniques and applications.
\newblock Available on: \url{http://tata.gforge.inria.fr/}, 2007.
\newblock release October, 12th 2007.

\bibitem[CDGL99]{calvanese1999reasoning}
Diego Calvanese, Giuseppe De~Giacomo, and Maurizio Lenzerini.
\newblock Reasoning in expressive description logics with fixpoints based on
  automata on infinite trees.
\newblock In {\em IJCAI}, volume~99, pages 84--89, 1999.

\bibitem[CDJ06]{chavira2006compiling}
Mark Chavira, Adnan Darwiche, and Manfred Jaeger.
\newblock Compiling relational bayesian networks for exact inference.
\newblock {\em International Journal of Approximate Reasoning}, 42(1-2):4--20,
  2006.

\bibitem[CJ06]{CJ06}
David~A. Cohen and Peter Jeavons.
\newblock The complexity of constraint languages.
\newblock In {\em Handbook of Constraint Programming}, pages 245--280.
  Elsevier, 2006.

\bibitem[CKS01]{DBLP:books/daglib/0004131}
Nadia Creignou, Sanjeev Khanna, and Madhu Sudan.
\newblock {\em Complexity classifications of Boolean constraint satisfaction
  problems}, volume~7 of {\em {SIAM} monographs on discrete mathematics and
  applications}.
\newblock {SIAM}, 2001.

\bibitem[CM77]{CM77}
Ashok~K. Chandra and Philip~M. Merlin.
\newblock Optimal implementation of conjunctive queries in relational data
  bases.
\newblock In {\em Proceedings of the 9th Annual {ACM} Symposium on Theory of
  Computing}, STOC'77, pages 77--90, 1977.

\bibitem[CMN99]{ChaudhuriMN99}
Surajit Chaudhuri, Rajeev Motwani, and Vivek~R. Narasayya.
\newblock On random sampling over joins.
\newblock In {\em SIGMOD}, pages 263--274, 1999.

\bibitem[Cou90]{C90}
Bruno Courcelle.
\newblock The monadic second-order logic of graphs. i. recognizable sets of
  finite graphs.
\newblock {\em Information and computation}, 85(1):12--75, 1990.

\bibitem[CR97]{CR97}
Chandra Chekuri and Anand Rajaraman.
\newblock Conjunctive query containment revisited.
\newblock In {\em Database Theory - {ICDT} '97, 6th International Conference,
  Delphi, Greece, January 8-10, 1997, Proceedings}, pages 56--70, 1997.

\bibitem[CY20]{Chen020}
Yu~Chen and Ke~Yi.
\newblock Random sampling and size estimation over cyclic joins.
\newblock In {\em ICDT}, pages 7:1--7:18, 2020.

\bibitem[Dar01a]{darwiche2001decomposable}
Adnan Darwiche.
\newblock Decomposable negation normal form.
\newblock {\em Journal of the ACM (JACM)}, 48(4):608--647, 2001.

\bibitem[Dar01b]{darwiche2001tractable}
Adnan Darwiche.
\newblock On the tractable counting of theory models and its application to
  truth maintenance and belief revision.
\newblock {\em Journal of Applied Non-Classical Logics}, 11(1-2):11--34, 2001.

\bibitem[DJ04]{DALMAU2004315}
V\'ictor Dalmau and Peter Jonsson.
\newblock The complexity of counting homomorphisms seen from the other side.
\newblock {\em Theoretical Computer Science}, 329(1):315 -- 323, 2004.

\bibitem[DM02]{darwiche2002knowledge}
Adnan Darwiche and Pierre Marquis.
\newblock A knowledge compilation map.
\newblock {\em Journal of Artificial Intelligence Research}, 17:229--264, 2002.

\bibitem[DM15]{DBLP:journals/mst/0001M15}
Arnaud Durand and Stefan Mengel.
\newblock {S}tructural {T}ractability of {C}ounting of {S}olutions to
  {C}onjunctive {Q}ueries.
\newblock {\em Theory Comput. Syst.}, 57(4):1202--1249, 2015.

\bibitem[EJ91]{emerson1991tree}
E~Allen Emerson and Charanjit~S Jutla.
\newblock Tree automata, mu-calculus and determinacy.
\newblock In {\em [1991] Proceedings 32nd Annual Symposium of Foundations of
  Computer Science}, pages 368--377. IEEE, 1991.

\bibitem[Fag83]{F83}
Ronald Fagin.
\newblock Degrees of acyclicity for hypergraphs and relational database
  schemes.
\newblock {\em J. {ACM}}, 30(3):514--550, 1983.

\bibitem[FG06]{FG06}
J{\"{o}}rg Flum and Martin Grohe.
\newblock {\em Parameterized Complexity Theory}.
\newblock Texts in Theoretical Computer Science. An {EATCS} Series. Springer,
  2006.

\bibitem[FV98]{FV98}
Tom{\'{a}}s Feder and Moshe~Y. Vardi.
\newblock The computational structure of monotone monadic {SNP} and constraint
  satisfaction: {A} study through datalog and group theory.
\newblock {\em {SIAM} J. Comput.}, 28(1):57--104, 1998.

\bibitem[GGLS16]{GGLS16}
Georg Gottlob, Gianluigi Greco, Nicola Leone, and Francesco Scarcello.
\newblock Hypertree decompositions: Questions and answers.
\newblock In {\em Proceedings of the 35th {ACM} {SIGMOD-SIGACT-SIGAI} Symposium
  on Principles of Database Systems, {PODS} 2016, San Francisco, CA, USA, June
  26 - July 01, 2016}, pages 57--74, 2016.

\bibitem[GJK{\etalchar{+}}97]{gore1997quasi}
Vivek Gore, Mark Jerrum, Sampath Kannan, Z~Sweedyk, and Steve Mahaney.
\newblock A quasi-polynomial-time algorithm for sampling words from a
  context-free language.
\newblock {\em Information and Computation}, 134(1):59--74, 1997.

\bibitem[GLS98]{GLS98}
Georg Gottlob, Nicola Leone, and Francesco Scarcello.
\newblock The complexity of acyclic conjunctive queries.
\newblock In {\em 39th Annual Symposium on Foundations of Computer Science,
  {FOCS}'98, November 8-11, 1998, Palo Alto, California, {USA}}, pages
  706--715, 1998.

\bibitem[GLS00]{GLS00}
Georg Gottlob, Nicola Leone, and Francesco Scarcello.
\newblock A comparison of structural {CSP} decomposition methods.
\newblock {\em Artif. Intell.}, 124(2):243--282, 2000.

\bibitem[GLS02]{GLS02}
Georg Gottlob, Nicola Leone, and Francesco Scarcello.
\newblock Hypertree decompositions and tractable queries.
\newblock {\em J. Comput. Syst. Sci.}, 64(3):579--627, 2002.

\bibitem[GSS01]{GSS01}
Martin Grohe, Thomas Schwentick, and Luc Segoufin.
\newblock When is the evaluation of conjunctive queries tractable?
\newblock In {\em Proceedings on 33rd Annual {ACM} Symposium on Theory of
  Computing, July 6-8, 2001, Heraklion, Crete, Greece}, STOC'01, pages
  657--666, 2001.

\bibitem[GSS09]{GomesSS09}
Carla~P. Gomes, Ashish Sabharwal, and Bart Selman.
\newblock Model counting.
\newblock In {\em Handbook of Satisfiability}, pages 633--654. 2009.

\bibitem[HN04]{DBLP:books/daglib/0013017}
Pavol Hell and Jaroslav Nesetril.
\newblock {\em Graphs and homomorphisms}, volume~28 of {\em Oxford lecture
  series in mathematics and its applications}.
\newblock Oxford University Press, 2004.

\bibitem[JVV86a]{JVV86}
Mark Jerrum, Leslie~G. Valiant, and Vijay~V. Vazirani.
\newblock Random generation of combinatorial structures from a uniform
  distribution.
\newblock {\em Theor. Comput. Sci.}, 43:169--188, 1986.

\bibitem[JVV86b]{jerrum1986random}
Mark~R Jerrum, Leslie~G Valiant, and Vijay~V Vazirani.
\newblock Random generation of combinatorial structures from a uniform
  distribution.
\newblock {\em Theoretical Computer Science}, 43:169--188, 1986.

\bibitem[KL83]{KL83}
Richard~M. Karp and Michael Luby.
\newblock Monte-carlo algorithms for enumeration and reliability problems.
\newblock In {\em 24th Annual Symposium on Foundations of Computer Science,
  Tucson, Arizona, USA, 7-9 November 1983}, FOCS'83, pages 56--64, 1983.

\bibitem[KLM89]{karp1989monte}
Richard~M Karp, Michael Luby, and Neal Madras.
\newblock Monte-carlo approximation algorithms for enumeration problems.
\newblock {\em Journal of algorithms}, 10(3):429--448, 1989.

\bibitem[KSM95]{kannan1995counting}
Sampath Kannan, Z~Sweedyk, and Steve Mahaney.
\newblock Counting and random generation of strings in regular languages.
\newblock In {\em Proceedings of the sixth annual ACM-SIAM symposium on
  Discrete algorithms}, pages 551--557. Society for Industrial and Applied
  Mathematics, 1995.

\bibitem[KV00]{KV00}
Phokion~G. Kolaitis and Moshe~Y. Vardi.
\newblock Conjunctive-query containment and constraint satisfaction.
\newblock {\em J. Comput. Syst. Sci.}, 61(2):302--332, 2000.

\bibitem[Mai94]{mairson1994generating}
Harry~G Mairson.
\newblock Generating words in a context-free language uniformly at random.
\newblock {\em Information Processing Letters}, 49(2):95--99, 1994.

\bibitem[Nev02]{N02}
Frank Neven.
\newblock Automata theory for {XML} researchers.
\newblock {\em {SIGMOD} Record}, 31(3):39--46, 2002.

\bibitem[OD14]{oztok2014cv}
Umut Oztok and Adnan Darwiche.
\newblock Cv-width: A new complexity parameter for cnfs.
\newblock In {\em ECAI}, pages 675--680, 2014.

\bibitem[PD08]{pipatsrisawat2008new}
Knot Pipatsrisawat and Adnan Darwiche.
\newblock New compilation languages based on structured decomposability.
\newblock In {\em AAAI}, volume~8, pages 517--522, 2008.

\bibitem[PS13]{PS13}
Reinhard Pichler and Sebastian Skritek.
\newblock Tractable counting of the answers to conjunctive queries.
\newblock {\em J. Comput. Syst. Sci.}, 79(6):984--1001, 2013.

\bibitem[Rab69]{R69}
Michael~O Rabin.
\newblock Decidability of second-order theories and automata on infinite trees.
\newblock {\em Transactions of the american Mathematical Society}, 141:1--35,
  1969.

\bibitem[RGG03]{ramakrishnan2003database}
Raghu Ramakrishnan, Johannes Gehrke, and Johannes Gehrke.
\newblock {\em Database management systems}, volume~3.
\newblock McGraw-Hill New York, 2003.

\bibitem[RN16]{russell2016artificial}
Stuart~J Russell and Peter Norvig.
\newblock {\em Artificial intelligence: a modern approach}.
\newblock Malaysia; Pearson Education Limited,, 2016.

\bibitem[RVBW06]{rossi2006handbook}
Francesca Rossi, Peter Van~Beek, and Toby Walsh.
\newblock {\em Handbook of constraint programming}.
\newblock Elsevier, 2006.

\bibitem[Sch07]{S07}
Thomas Schwentick.
\newblock Automata for xml?a survey.
\newblock {\em Journal of Computer and System Sciences}, 73(3):289--315, 2007.

\bibitem[Sei90]{S90}
Helmut Seidl.
\newblock Deciding equivalence of finite tree automata.
\newblock {\em {SIAM} J. Comput.}, 19(3):424--437, 1990.

\bibitem[Ter99]{T99}
Eugenia Ternovskaia.
\newblock Automata theory for reasoning about actions.
\newblock In {\em Proceedings of the Sixteenth International Joint Conference
  on Artificial Intelligence, {IJCAI} 99, Stockholm, Sweden, July 31 - August
  6, 1999. 2 Volumes, 1450 pages}, pages 153--159, 1999.

\bibitem[Tho97]{T97}
Wolfgang Thomas.
\newblock Languages, automata, and logic.
\newblock In {\em Handbook of formal languages}, pages 389--455. Springer,
  1997.

\bibitem[TW68]{TR68}
James~W. Thatcher and Jesse~B. Wright.
\newblock Generalized finite automata theory with an application to a decision
  problem of second-order logic.
\newblock {\em Mathematical systems theory}, 2(1):57--81, 1968.

\bibitem[Val79]{valiant1979complexity}
Leslie~G Valiant.
\newblock The complexity of enumeration and reliability problems.
\newblock {\em SIAM Journal on Computing}, 8(3):410--421, 1979.

\bibitem[Var95]{V95}
Moshe~Y Vardi.
\newblock Alternating automata and program verification.
\newblock In {\em Computer Science Today}, pages 471--485. Springer, 1995.

\bibitem[Var00]{V00}
Moshe~Y. Vardi.
\newblock Constraint satisfaction and database theory: a tutorial.
\newblock In {\em Proceedings of the Nineteenth {ACM} {SIGMOD-SIGACT-SIGART}
  Symposium on Principles of Database Systems, May 15-17, 2000, Dallas, Texas,
  {USA}}, pages 76--85, 2000.

\bibitem[VSBR83]{valiant1983fast}
LG~Valiant, S~Skyum, S~Berkowitz, and C~Rackoff.
\newblock Fast parallel computation of polynomials using few processors.
\newblock {\em SIAM Journal on Computing}, 12(4):641--644, 1983.

\bibitem[Wil10]{W10}
Ross Willard.
\newblock Testing expressibility is hard.
\newblock In {\em Principles and Practice of Constraint Programming - {CP} 2010
  - 16th International Conference, {CP} 2010, St. Andrews, Scotland, UK,
  September 6-10, 2010. Proceedings}, pages 9--23, 2010.

\bibitem[Yan81]{Y81}
Mihalis Yannakakis.
\newblock Algorithms for acyclic database schemes.
\newblock In {\em Very Large Data Bases, 7th International Conference,
  September 9-11, 1981, Cannes, France, Proceedings}, pages 82--94, 1981.

\bibitem[ZCL{\etalchar{+}}18]{ZhaoC0HY18}
Zhuoyue Zhao, Robert Christensen, Feifei Li, Xiao Hu, and Ke~Yi.
\newblock Random sampling over joins revisited.
\newblock In {\em SIGMOD}, pages 1525--1539, 2018.

\bibitem[Zhu17]{Z17}
Dmitriy Zhuk.
\newblock A proof of {CSP} dichotomy conjecture.
\newblock In {\em 58th {IEEE} Annual Symposium on Foundations of Computer
  Science, {FOCS} 2017, Berkeley, CA, USA, October 15-17, 2017}, pages
  331--342, 2017.

\end{thebibliography}
